\documentclass[acmsmall,screen]{acmart}

\usepackage{graphicx} 
\usepackage{todonotes}
\usepackage[linesnumbered,lined,boxed]{algorithm2e}
\usepackage{multirow}
\usepackage{adjustbox}
\usepackage{listings}
\usepackage[capitalise]{cleveref}
\usepackage[subrefformat=parens]{subcaption}
\usepackage{stackengine}
\usepackage{enumerate}
\usepackage[shortlabels,inline]{enumitem}
\usepackage{xparse}
\usepackage{amsthm}
\usepackage{csvsimple}
\usepackage[createShortEnv]{proof-at-the-end}
\usepackage{csvsimple}
\usepackage{dsfont}
\usepackage{enumerate}

\usepackage{tikz}
\usetikzlibrary{backgrounds}
\usetikzlibrary{calc}
\usetikzlibrary{positioning}
\usetikzlibrary{shadows}
\usetikzlibrary{fit}

\usepackage{framed}
\usepackage{adjustbox}

\setstackEOL{\#}
\setstackgap{L}{1.1\baselineskip}
\stackMath

\makeatletter
\renewcommand{\paragraph}{%
  \@startsection{paragraph}{4}%
  {\z@}{-.1\baselineskip \@plus -2\p@ \@minus -.2\p@}{-3.5\p@}%
  {\bfseries\@parfont\@adddotafter}%
}
\makeatother

\makeatletter
\newcommand\notsotiny{\@setfontsize\notsotiny\@viipt\@viiipt}
\makeatother

\definecolor{Gray}{gray}{0.6}
\definecolor{Emerald}{RGB}{0, 128, 0}
\definecolor{NavyBlue}{RGB}{0, 0, 128}
\definecolor{Magenta}{RGB}{128, 0, 0}
\definecolor{Brown}{RGB}{128, 128, 0}
\definecolor{Darkgreen}{RGB}{0, 100, 0}

\newcommand{\iql}[0]{GenSQL}

\newcommand{\tableModel}{rowModel}
\newcommand{\TableModel}{RowModel}
\newcommand{\tableModels}{rowModels}
\newcommand{\ami}{AMI}

\newcommand{\RR}{\mathbb{R}}
\newcommand{\BB}{\mathbb{B}}

\newcommand{\NN}{\mathbb{N}}
\newcommand{\ZZ}{\mathbb{Z}}

\newcommand{\syntaxcolor}[0]{\color{blue!65!black}}
\newcommand{\syntaxcolorgreen}[0]{\color{Darkgreen}}

\newcommand{\colforsyntax}[1]{\mathbf{\syntaxcolor{#1}}}

\newcommand{\revision}[1]{#1}

\newcommand{\oldcons}[1]{#1}
\newcommand{\newcons}[1]{}

\newcommand{\true}{\colforsyntax{true}}
\newcommand{\false}{\colforsyntax{false}}
\newcommand{\colorbasetype}[1]{\textbf{#1}}
\newcommand\real{\colorbasetype{Real}}
\newcommand\posreal{\colorbasetype{PosReal}}
\newcommand\ranged{\colorbasetype{Ranged}}
\newcommand\bool{\colorbasetype{Bool}}
\newcommand\nat{\colorbasetype{Nat}}
\newcommand\inte{\colorbasetype{Int}}
\newcommand\str{\colorbasetype{Str}}
\newcommand\strsem{\mathbf{Str}}
\newcommand\categorical{\colorbasetype{Cat}}
\newcommand\bag{\colorbasetype{Bag}}

\newcommand\vars{\mathbf{vars}}
\newcommand\condvars{\mathbf{condvars}}

\newcommand\Null{\colforsyntax{Null}}

\newcommand\gor{\mathrel{\lvert}}

\newcommand{\tabty}[1]{\{#1\}}
\newcommand{\col}{\textsc{col}}
\newcommand{\cols}{\textsc{cols}}

\newcommand{\id}{\textsc{id}}

\newcommand{\colorsyntaxiql}[1]{\textbf{\notsotiny \syntaxcolor{#1}}}
\newcommand{\union}{\ \colorsyntaxiql{UNION}\ }
\newcommand{\join}{\ \colorsyntaxiql{JOIN}\ }
\newcommand{\where}{\ \colorsyntaxiql{WHERE}\ }
\newcommand{\rename}{\ \colorsyntaxiql{RENAME}\ }

\newcommand{\dedup}{\ \colorsyntaxiql{DEDUP}\ }
\newcommand{\as}{\ \colorsyntaxiql{AS}\ }
\newcommand{\select}{\ \colorsyntaxiql{SELECT}\ }
\newcommand{\from}{\ \colorsyntaxiql{FROM}\ }
\newcommand{\group}{\ \colorsyntaxiql{GROUP}\ }
\newcommand{\by}{\ \colorsyntaxiql{BY}\ }
\newcommand{\exceptt}{\ \colorsyntaxiql{EXCEPT}\ }
\newcommand{\aggr}{\ \colorsyntaxiql{AGGREGATING}\ }
\newcommand{\generate}{\ \colorsyntaxiql{GENERATE}\ \colorsyntaxiql{UNDER}\ }
\newcommand{\probof}{\ \colorsyntaxiql{PROBABILITY}\ \colorsyntaxiql{OF}\ }

\newcommand{\mi}{\ \colorsyntaxiql{MUTUAL}\ \colorsyntaxiql{INFO}\ }
\newcommand{\limit}{\ \colorsyntaxiql{LIMIT}\ }
\newcommand{\under}{\ \colorsyntaxiql{UNDER}\ }

\newcommand{\generative}{\colorsyntaxiql{GENERATIVE}}
\newcommand{\modjoin}{\ \generative \join}

\newcommand{\given}{\ \colorsyntaxiql{GIVEN}\ }

\newcommand{\duplicate}{\ \colorsyntaxiql{DUPLICATE}\ }
\newcommand{\timess}{\ \colorsyntaxiql{TIMES}\ }

\newcommand{\with}{\ \colorsyntaxiql{WITH}\ }

\newcommand{\ttable}{\ \colorsyntaxiql{TABLE}\ }
\newcommand{\model}{\ \colorsyntaxiql{MODEL}\ }

\newcommand{\coloraggr}[1]{\colforsyntax{#1}}
\newcommand{\Sum}{\coloraggr{Sum}}
\newcommand{\Avg}{\coloraggr{Avg}}
\newcommand{\Max}{\coloraggr{Max}}
\newcommand{\Min}{\coloraggr{Min}}
\newcommand{\Count}{\coloraggr{Count}}
\newcommand{\Concat}{\coloraggr{Concat}}
\newcommand{\Countdis}{\coloraggr{CountDistinct}}
\newcommand{\ctx}{\ensuremath{\Gamma; \Delta}}

\newcommand{\dotmeas}{\text{.meas}}
\newcommand{\dotpdf}{\text{.pdf}}
\newcommand{\sem}[1]{\llbracket #1\rrbracket}

\newcommand{\semex}[1]{\llbracket #1\rrbracket_{\text{exact}}}
\newcommand{\semap}[1]{\llbracket #1\rrbracket_{\text{approx}}}

\newcommand{\tupsem}[1]{\mathsf{Tup}\llbracket #1\rrbracket}
\newcommand{\tabsem}[1]{\mathsf{Tab}\llbracket #1\rrbracket}
\newcommand{\measm}{\mathcal{P}}
\newcommand{\admissible}{\measm_{\text{adm}}}
\newcommand{\aadmissible}{\measm_{\text{adm}}^{\text{approx}}}
\newcommand{\dens}{\measm_{\text{dens}}}
\newcommand{\opset}{\textbf{Op}}
\newcommand{\op}{\textit{op}}
\newcommand{\bind}[0]{\mathop{\gg\mkern-10mu\scalebox{1}[1]{=}}}

\newcommand{\mathopcolor}[1]{\mathbf{#1}}
\newcommand{\return}[0]{\mathopcolor{return}\ }
\newcommand{\filter}[0]{\mathopcolor{filter}\ }
\newcommand{\map}[0]{\mathopcolor{map}\ }
\newcommand{\maptwo}[0]{\mathopcolor{map2}\ }

\newcommand{\splat}[0]{\mathopcolor{splat}\ }

\newcommand{\lowerjoin}[0]{\mathopcolor{join}}
\newcommand{\fold}[0]{\mathopcolor{fold}\ }

\newcommand{\llet}[0]{\mathopcolor{let}\ }
\newcommand{\iin}[0]{\mathopcolor{in}\ }
\newcommand{\iif}[0]{\mathopcolor{if}\ }
\newcommand{\tthen}[0]{\mathopcolor{then}\ }
\newcommand{\eelse}[0]{\mathopcolor{else}\ }

\newcommand{\val}[0]{\mathopcolor{val}\ }
\newcommand{\eventtype}{C^1}
\newcommand{\densitytype}{C^0}
\newcommand{\evtype}{\mathcal{E}}

\newcommand{\tabtype}{\mathcal{T}}
\newcommand{\modtype}{\mathcal{M}}
\newcommand{\Dis}{\mathbf{Dis}}
\newcommand{\Cond}{\mathbf{Cond}}

\newcommand{\mapreduce}[0]{\mathopcolor{mapreduce}\ }

\newcommand{\replicate}[0]{\mathopcolor{replicate}\ }

\newcommand{\colset}{\mathcal{C}}
\newcommand{\idset}{\mathcal{I}}

\newcommand{\interfacecolor}[0]{\color{green!55!black}}
\newcommand{\colforatmi}[1]{\mathbf{\interfacecolor{#1}}}
\newcommand{\simulate}{\colforatmi{simulate}}
\newcommand{\logpdf}{\colforatmi{logpdf}}
\newcommand{\prob}{\colforatmi{prob}}

\newcommand{\transymbol}[0]{\mathcal{T}}
\newcommand{\tran}[2]{\transymbol_{#1}\left\{#2\right\}}
\newcommand{\tranty}[1]{\transymbol\left\{#1\right\}}

\newcommand{\hostexp}{\mathbf{exp}}
\newcommand{\singleton}{\mathbf{singleton}}

\newcommand{\stseq}{{\mathrm{StochSeq}}}
\newcommand{\toas}{\xrightarrow{\mathrm{a.s.}}}
\newcommand{\safe}{\mathbf{safe?}}
\newcommand{\exact}{\mathbf{exact?}}
\newcommand{\continuous}{\mathbf{continuous?}}

\newcommand{\Model}{\mathrm{Model}}
\newcommand{\SPE}{\mathrm{SPE}}
\newcommand{\TMVG}{\mathrm{TMVG}}
\newcommand{\ancestral}{\mathbf{ancestral}}
\newcommand{\cond}{\mathbf{cond}}
\newcommand{\truncate}{\mathbf{truncate}}
\newcommand{\marginalize}{\mathbf{marginalize}}

\lstdefinestyle{iql}{
    language=SQL,
    basicstyle=\tiny\ttfamily,
    columns=fullflexible,
    keepspaces=true,
    upquote=true,
    showstringspaces=false,
    comment=[l]{\#},
    commentstyle=\color{Gray}\ttfamily\bfseries,
    keywordstyle=\tiny\bfseries\syntaxcolor,
    stringstyle=\color{Magenta},
    escapechar=\@,
    breaklines=true,
    captionpos=b,
    frame=single,
    numbers=none,
    numbersep=10pt,
    showtabs=false,
    deletekeywords={column,value},
    literate=
        {=}{{{{\bfseries\syntaxcolor=}}}}1
        {>}{{{{\bfseries\syntaxcolor>}}}}1
        {<}{{{{\bfseries\syntaxcolor<}}}}1
        {+}{{{{\bfseries\syntaxcolor+}}}}1
        {-}{{{{\bfseries\syntaxcolor-}}}}1
        {*}{{{{\bfseries\syntaxcolor*}}}}1
        {/}{{{{\bfseries\syntaxcolor/}}}}1
        {,}{{{{\bfseries\syntaxcolor,}}}}1
        {:}{{{{\bfseries\syntaxcolor:}}}}1
        {0}{{{{\color{Emerald}0}}}}1
        {1}{{{{\color{Emerald}1}}}}1
        {2}{{{{\color{Emerald}2}}}}1
        {3}{{{{\color{Emerald}3}}}}1
        {4}{{{{\color{Emerald}4}}}}1
        {5}{{{{\color{Emerald}5}}}}1
        {6}{{{{\color{Emerald}6}}}}1
        {7}{{{{\color{Emerald}7}}}}1
        {8}{{{{\color{Emerald}8}}}}1
        {9}{{{{\color{Emerald}9}}}}1,
    morekeywords={OF,PROBABILITY,VARIABLE,UNDER,CONDITIONED,EVERYTHING,IS,LOG,LOAD,
        DENSITY,CONSTRAINED,BUILD,WITH,STD,STDEV,FOR,GENERATE,MUTUAL,INFORMATION,COLUMNS,GENERATIVE,DUPLICATE,TIMES,
        PROBABILITY_DENSITY, PREDICT,CONDITION,CONSTRAIN,MUTUAL_INFO,MUTUAL_INFORMATION,
        MODEL,AVG,MEDIAN,SUM,COUNT,GIVEN,INFER,import},
  keywordstyle=[2]\textcolor{Magenta},
  morekeywords=[2]{Mil},
  sensitive=true,
}

\lstdefinestyle{iql-normalsize}{
    language=SQL,
    basicstyle=\small\ttfamily,
    columns=fullflexible,
    keepspaces=true,
    upquote=true,
    showstringspaces=false,
    comment=[l]{\#},
    commentstyle=\color{Gray}\ttfamily\bfseries,
    keywordstyle=\small\bfseries\syntaxcolor,
    stringstyle=\color{Magenta},
    escapechar=\@,
    breaklines=true,
    captionpos=b,
    frame=single,
    numbers=none,
    numbersep=10pt,
    showtabs=false,
    deletekeywords={column,value},
    literate=
        {=}{{{{\bfseries\syntaxcolor=}}}}1
        {>}{{{{\bfseries\syntaxcolor>}}}}1
        {<}{{{{\bfseries\syntaxcolor<}}}}1
        {+}{{{{\bfseries\syntaxcolor+}}}}1
        {-}{{{{\bfseries\syntaxcolor-}}}}1
        {*}{{{{\bfseries\syntaxcolor*}}}}1
        {/}{{{{\bfseries\syntaxcolor/}}}}1
        {,}{{{{\bfseries\syntaxcolor,}}}}1
        {:}{{{{\bfseries\syntaxcolor:}}}}1
        {0}{{{{\color{Emerald}0}}}}1
        {1}{{{{\color{Emerald}1}}}}1
        {2}{{{{\color{Emerald}2}}}}1
        {3}{{{{\color{Emerald}3}}}}1
        {4}{{{{\color{Emerald}4}}}}1
        {5}{{{{\color{Emerald}5}}}}1
        {6}{{{{\color{Emerald}6}}}}1
        {7}{{{{\color{Emerald}7}}}}1
        {8}{{{{\color{Emerald}8}}}}1
        {9}{{{{\color{Emerald}9}}}}1,
    morekeywords={OF,PROBABILITY,VARIABLE,UNDER,CONDITIONED,EVERYTHING,IS,LOG,LOAD,
        DENSITY,CONSTRAINED,BUILD,WITH,STD,STDEV,FOR,GENERATE,MUTUAL,INFORMATION,COLUMNS,GENERATIVE,DUPLICATE,TIMES,
        PROBABILITY_DENSITY, PREDICT,CONDITION,CONSTRAIN,MUTUAL_INFO,MUTUAL_INFORMATION,
        MODEL,AVG,MEDIAN,SUM,COUNT,GIVEN,INFER,import},
  keywordstyle=[2]\textcolor{Magenta},
  morekeywords=[2]{Mil},
  sensitive=true,
}

\lstdefinestyle{iql-small}{
    language=SQL,
    basicstyle=\footnotesize\ttfamily,
    columns=fullflexible,
    keepspaces=true,
    upquote=true,
    showstringspaces=false,
    comment=[l]{\#},
    commentstyle=\color{Gray}\ttfamily\bfseries,
    keywordstyle=\small\bfseries\syntaxcolor,
    stringstyle=\color{Magenta},
    escapechar=\@,
    breaklines=true,
    captionpos=b,
    frame=single,
    numbers=none,
    numbersep=10pt,
    showtabs=false,
    deletekeywords={column,value},
    literate=
        {=}{{{{\bfseries\syntaxcolor=}}}}1
        {>}{{{{\bfseries\syntaxcolor>}}}}1
        {<}{{{{\bfseries\syntaxcolor<}}}}1
        {+}{{{{\bfseries\syntaxcolor+}}}}1
        {-}{{{{\bfseries\syntaxcolor-}}}}1
        {*}{{{{\bfseries\syntaxcolor*}}}}1
        {/}{{{{\bfseries\syntaxcolor/}}}}1
        {,}{{{{\bfseries\syntaxcolor,}}}}1
        {:}{{{{\bfseries\syntaxcolor:}}}}1
        {0}{{{{\color{Emerald}0}}}}1
        {1}{{{{\color{Emerald}1}}}}1
        {2}{{{{\color{Emerald}2}}}}1
        {3}{{{{\color{Emerald}3}}}}1
        {4}{{{{\color{Emerald}4}}}}1
        {5}{{{{\color{Emerald}5}}}}1
        {6}{{{{\color{Emerald}6}}}}1
        {7}{{{{\color{Emerald}7}}}}1
        {8}{{{{\color{Emerald}8}}}}1
        {9}{{{{\color{Emerald}9}}}}1,
    morekeywords={OF,PROBABILITY,VARIABLE,UNDER,CONDITIONED,EVERYTHING,IS,LOG,LOAD,
        DENSITY,CONSTRAINED,BUILD,WITH,STD,STDEV,FOR,GENERATE,MUTUAL,INFORMATION,COLUMNS,GENERATIVE,DUPLICATE,TIMES,
        PROBABILITY_DENSITY, PREDICT,CONDITION,CONSTRAIN,MUTUAL_INFO,MUTUAL_INFORMATION,
        MODEL,AVG,MEDIAN,SUM,COUNT,GIVEN,INFER,import},
  keywordstyle=[2]\textcolor{Magenta},
  morekeywords=[2]{Mil},
  sensitive=true,
}

\lstdefinestyle{sqlite3}{
    language=SQL,
    basicstyle=\tiny\ttfamily,
    columns=fullflexible,
    keepspaces=true,
    upquote=true,
    showstringspaces=false,
    comment=[l]{\#},
    commentstyle=\color{Gray}\ttfamily\bfseries,
    keywordstyle=\tiny\bfseries\syntaxcolorgreen,
    stringstyle=\color{Magenta},
    escapechar=\@,
    breaklines=true,
    captionpos=b,
    frame=single,
    numbers=none,
    numbersep=10pt,
    showtabs=false,
    deletekeywords={column,value},
    literate=
        {=}{{{{\bfseries\syntaxcolorgreen=}}}}1
        {>}{{{{\bfseries\syntaxcolorgreen>}}}}1
        {<}{{{{\bfseries\syntaxcolorgreen<}}}}1
        {+}{{{{\bfseries\syntaxcolorgreen+}}}}1
        {-}{{{{\bfseries\syntaxcolorgreen-}}}}1
        {*}{{{{\bfseries\syntaxcolorgreen*}}}}1
        {/}{{{{\bfseries\syntaxcolorgreen/}}}}1
        {,}{{{{\bfseries\syntaxcolorgreen,}}}}1
        {:}{{{{\bfseries\syntaxcolorgreen:}}}}1
        {0}{{{{\color{Emerald}0}}}}1
        {1}{{{{\color{Emerald}1}}}}1
        {2}{{{{\color{Emerald}2}}}}1
        {3}{{{{\color{Emerald}3}}}}1
        {4}{{{{\color{Emerald}4}}}}1
        {5}{{{{\color{Emerald}5}}}}1
        {6}{{{{\color{Emerald}6}}}}1
        {7}{{{{\color{Emerald}7}}}}1
        {8}{{{{\color{Emerald}8}}}}1
        {9}{{{{\color{Emerald}9}}}}1,
    morekeywords={},
  keywordstyle=[2]\textcolor{Magenta},
  morekeywords=[2]{Mil},
  sensitive=true,
}
\lstdefinestyle{python-sppl}{
    language=Python,
    basicstyle=\tiny\ttfamily,
    columns=fullflexible,
    keepspaces=true,
    upquote=true,
    showstringspaces=false,
    comment=[l]{\#},
    commentstyle=\color{Gray}\ttfamily\bfseries,
    keywordstyle=\tiny\bfseries\syntaxcolor,
    stringstyle=\color{Magenta},
    escapechar=\@,
    breaklines=true,
    captionpos=b,
    frame=single,
    numbers=none,
    numbersep=10pt,
    showtabs=false,
    deletekeywords={column,value},
    literate=
        {=}{{{{\bfseries\syntaxcolor=}}}}1
        {>}{{{{\bfseries\syntaxcolor>}}}}1
        {<}{{{{\bfseries\syntaxcolor<}}}}1
        {+}{{{{\bfseries\syntaxcolor+}}}}1
        {-}{{{{\bfseries\syntaxcolor-}}}}1
        {*}{{{{\bfseries\syntaxcolor*}}}}1
        {/}{{{{\bfseries\syntaxcolor/}}}}1
        {,}{{{{\bfseries\syntaxcolor,}}}}1
        {:}{{{{\bfseries\syntaxcolor:}}}}1
        {0}{{{{\color{Emerald}0}}}}1
        {1}{{{{\color{Emerald}1}}}}1
        {2}{{{{\color{Emerald}2}}}}1
        {3}{{{{\color{Emerald}3}}}}1
        {4}{{{{\color{Emerald}4}}}}1
        {5}{{{{\color{Emerald}5}}}}1
        {6}{{{{\color{Emerald}6}}}}1
        {7}{{{{\color{Emerald}7}}}}1
        {8}{{{{\color{Emerald}8}}}}1
        {9}{{{{\color{Emerald}9}}}}1,
  keywordstyle=[2]\textcolor{Magenta},
  morekeywords=[2]{Mil},
  sensitive=true,
}
\lstdefinestyle{clojure-cat}{
    language=Lisp,
    basicstyle=\tiny\ttfamily,
    columns=fullflexible,
    keepspaces=true,
    upquote=true,
    showstringspaces=false,
    commentstyle=\color{Gray}\ttfamily\bfseries,
    keywordstyle=\tiny\bfseries\syntaxcolor,
    stringstyle=\color{Magenta},
    escapechar=\@,
    breaklines=true,
    captionpos=b,
    frame=single,
    numbers=none,
    numbersep=10pt,
    showtabs=false,
    deletekeywords={map,value},
    literate=
        {=}{{{{\bfseries\syntaxcolor=}}}}1
        {>}{{{{\bfseries\syntaxcolor>}}}}1
        {<}{{{{\bfseries\syntaxcolor<}}}}1
        {+}{{{{\bfseries\syntaxcolor+}}}}1
        {-}{{{{\bfseries\syntaxcolor-}}}}1
        {*}{{{{\bfseries\syntaxcolor*}}}}1
        {/}{{{{\bfseries\syntaxcolor/}}}}1
        {,}{{{{\bfseries\syntaxcolor,}}}}1
        {:}{{{{\bfseries\syntaxcolor:}}}}1
        {0}{{{{\color{Emerald}0}}}}1
        {1}{{{{\color{Emerald}1}}}}1
        {2}{{{{\color{Emerald}2}}}}1
        {3}{{{{\color{Emerald}3}}}}1
        {4}{{{{\color{Emerald}4}}}}1
        {5}{{{{\color{Emerald}5}}}}1
        {6}{{{{\color{Emerald}6}}}}1
        {7}{{{{\color{Emerald}7}}}}1
        {8}{{{{\color{Emerald}8}}}}1
        {9}{{{{\color{Emerald}9}}}}1,
  morekeywords={let,take,map,mapv,select,keys},
  keywordstyle=[2]\textcolor{Magenta},
  morekeywords=[2]{Mil},
  sensitive=true,
}

\lstdefinestyle{clojure}{
    language=Lisp,
    basicstyle=\footnotesize\ttfamily,
    columns=fullflexible,
    keepspaces=true,
    upquote=true,
    showstringspaces=false,
    commentstyle=\color{Gray}\ttfamily\bfseries,
    keywordstyle=\footnotesize\bfseries\syntaxcolor,
    stringstyle=\color{Magenta},
    escapechar=\@,
    breaklines=true,
    captionpos=b,
    frame=single,
    numbers=none,
    numbersep=10pt,
    showtabs=false,
    deletekeywords={map,value},
    literate=
        {=}{{{{\bfseries\syntaxcolor=}}}}1
        {>}{{{{\bfseries\syntaxcolor>}}}}1
        {<}{{{{\bfseries\syntaxcolor<}}}}1
        {+}{{{{\bfseries\syntaxcolor+}}}}1
        {-}{{{{\bfseries\syntaxcolor-}}}}1
        {*}{{{{\bfseries\syntaxcolor*}}}}1
        {/}{{{{\bfseries\syntaxcolor/}}}}1
        {,}{{{{\bfseries\syntaxcolor,}}}}1
        {:}{{{{\bfseries\syntaxcolor:}}}}1
        {0}{{{{\color{Emerald}0}}}}1
        {1}{{{{\color{Emerald}1}}}}1
        {2}{{{{\color{Emerald}2}}}}1
        {3}{{{{\color{Emerald}3}}}}1
        {4}{{{{\color{Emerald}4}}}}1
        {5}{{{{\color{Emerald}5}}}}1
        {6}{{{{\color{Emerald}6}}}}1
        {7}{{{{\color{Emerald}7}}}}1
        {8}{{{{\color{Emerald}8}}}}1
        {9}{{{{\color{Emerald}9}}}}1,
  morekeywords={let,take,map,mapv,select,keys},
  keywordstyle=[2]\textcolor{Magenta},
  morekeywords=[2]{Mil},
  sensitive=true,
}

\newcommand\YAMLcolonstyle{\color{red}\bfseries}
\newcommand\YAMLvaluestyle{\color{black}\mdseries}

\lstdefinelanguage{yaml}
{%
  keywords={true,false,null,y,n},
  keywordstyle=\color{darkgray}\bfseries,
  basicstyle=\tiny\ttfamily\bfseries\syntaxcolor,
  frame=single,
  rulecolor=\color{black},
  sensitive=false,
  comment=[l]{\#},
  morecomment=[s]{/*}{*/},
  commentstyle=\color{Gray}\ttfamily\bfseries,
  stringstyle=\YAMLvaluestyle\ttfamily,
  moredelim=[l][\color{orange}]{\&},
  moredelim=[l][\color{magenta}]{*},
  moredelim=**[il][\YAMLcolonstyle{:}\YAMLvaluestyle]{:},   
  morestring=[b]',
  morestring=[b]",
  columns=fullflexible,
  keepspaces=true,
}

\lstdefinelanguage{Julia}{%
    basicstyle       = \tiny\ttfamily,
    frame=single,
    keywordstyle     = \bfseries\syntaxcolor,
    stringstyle      = \color{magenta},
    commentstyle     = \color{ForestGreen},
    numberstyle=\color{Emerald},
    escapechar=\%,
    showstringspaces = false,
    morekeywords={abstract,break,case,catch,const,continue,do,else,elseif,%
      end,export,false,for,function,immutable,import,importall,if,in,%
      macro,module,otherwise,quote,return,switch,true,try,type,typealias,%
      using,while},
    morekeywords=[2]{@gen},
    keywordstyle=[2]\bfseries\textcolor{Brown},
    morekeywords=[3]{:period},
    keywordstyle=[3]\bfseries\textcolor{Magenta},
    sensitive=true,
    literate=
        {0}{{{{\color{Emerald}0}}}}1
        {1}{{{{\color{Emerald}1}}}}1
        {2}{{{{\color{Emerald}2}}}}1
        {3}{{{{\color{Emerald}3}}}}1
        {4}{{{{\color{Emerald}4}}}}1
        {5}{{{{\color{Emerald}5}}}}1
        {6}{{{{\color{Emerald}6}}}}1
        {7}{{{{\color{Emerald}7}}}}1
        {8}{{{{\color{Emerald}8}}}}1
        {9}{{{{\color{Emerald}9}}}}1,
}

\usepackage{pythonhighlight}

\AtBeginDocument{%
  \providecommand\BibTeX{{%
    \normalfont B\kern-0.5em{\scshape i\kern-0.25em b}\kern-0.8em\TeX}}}

\setcopyright{rightsretained}
\acmDOI{10.1145/3656409}
\acmYear{2024}
\copyrightyear{2024}
\acmSubmissionID{pldi24main-p182-p}
\acmJournal{PACMPL}
\acmVolume{8}
\acmNumber{PLDI}
\acmArticle{179}
\acmMonth{6}
\received{2023-11-16}
\received[accepted]{2024-03-31}

\begin{document}

\title{\iql{}: A Probabilistic Programming System for Querying Generative Models of Database Tables}
\author{Mathieu Huot}
\email{mhuot@mit.edu}
\affiliation{%
  \institution{Massachusetts Institute of Technology}
  \city{Cambridge}
  \state{MA}
  \country{USA}}
  \orcid{0000-0002-5294-9088}

\author{Matin Ghavami}
\email{mghavami@mit.edu}
\affiliation{%
  \institution{Massachusetts Institute of Technology}
  \city{Cambridge}
  \state{MA}
  \country{USA}}
  \orcid{0000-0003-3052-7412}

\author{Alexander K. Lew}
\email{alexlew@mit.edu}
\affiliation{%
  \institution{Massachusetts Institute of Technology}
  \city{Cambridge}
  \state{MA}
  \country{USA}}
  \orcid{0000-0002-9262-4392}

\author{Ulrich Schaechtle}
\email{ulli-schaechtle@garage.co.jp}
\affiliation{%
  \institution{Digital Garage}
  \city{Tokyo}
  \country{Japan}}
\orcid{0009-0005-8897-6394}

\author{Cameron E. Freer}
\email{freer@mit.edu}
\affiliation{%
  \institution{Massachusetts Institute of Technology}
  \city{Cambridge}
  \state{MA}
  \country{USA}}
\orcid{0000-0003-1791-6843}

\author{Zane Shelby}
\email{zane-shelby@garage.co.jp}
\affiliation{%
  \institution{Digital Garage}
  \city{Tokyo}
  \country{Japan}}
\orcid{0009-0003-2976-4581}

\author{Martin C. Rinard}
\email{rinard@csail.mit.edu}
\affiliation{%
  \institution{Massachusetts Institute of Technology}
  \city{Cambridge}
  \state{MA}
  \country{USA}}
  \orcid{0000-0001-8095-8523}

\author{Feras A. Saad}
\email{fsaad@cmu.edu}
\affiliation{%
  \institution{Carnegie Mellon University}
  \city{Pittsburgh}
  \state{PA}
  \country{USA}}
\orcid{0000-0002-0505-795X}

\author{Vikash K. Mansinghka}
\email{vkm@mit.edu}
\affiliation{%
  \institution{Massachusetts Institute of Technology}
  \city{Cambridge}
  \state{MA}
  \country{USA}}
  \orcid{0000-0003-2507-0833}


\begin{abstract}
This article presents \iql{}, a probabilistic programming system for
querying probabilistic generative models of database tables.
By augmenting SQL with only a few key primitives for querying
probabilistic models, \iql{} enables complex Bayesian inference
workflows to be concisely implemented.
\iql{}'s query planner rests on a unified programmatic interface for
interacting with probabilistic models of tabular data, which
makes it possible to use models written in a variety of probabilistic
programming languages that are tailored to specific workflows.
Probabilistic models may be automatically learned via probabilistic
program synthesis, hand-designed, or a combination of both.
\iql{} is formalized using a novel type system and denotational
semantics, which together enable us to establish proofs that precisely
characterize its soundness guarantees.
We evaluate our system on two case real-world studies---an anomaly
detection in clinical trials and conditional synthetic data generation
for a virtual wet lab---and show that \iql{} more accurately captures
the complexity of the data as compared to common baselines.
We also show that the declarative syntax in \iql{} is more concise and
less error-prone as compared to several alternatives.
Finally, \iql{} delivers a $1.7$-$6.8\mathrm{x}$ speedup compared to
its closest competitor on a representative benchmark set and runs in
comparable time to hand-written code, in part due to its reusable
optimizations and code specialization.
\end{abstract}



\maketitle
\renewcommand{\shortauthors}{Huot, Ghavami, Lew, Schaechtle, Freer, Shelby, Rinard, Saad, Mansinghka}

\section{Introduction}

Building generative models of tabular data is a central focus in Bayesian
data analysis \citep{gelman1995bayesian}, probabilistic machine
learning \citep{murphy2022probabilistic} and in applications such as
econometrics \citep{bauwens2000bayesian}, healthcare
\citep{hernan2006estimating} and systems biology
\citep{wilkinson2018stochastic}.
Motivated by these applications, researchers have developed
techniques for automatically learning rich probabilistic models of
tabular
data~\citep{mansinghka2016crosscat,saad2019bayesian,gens2013learning,grosse2012exploiting,adams2010learning}.
To fully exploit these models for solving complex tasks, users
must be able to easily interleave operations that access both
tabular data records and probabilistic models.
Examples computations include
\begin{enumerate*}[label=(\roman*)]
\item generating synthetic data records that satisfy user constraints;
\item conditioning distributions specified by probabilistic models given
observed data records; and
\item using database operations to aggregate the results of combined
queries against tabular and model data.
\end{enumerate*}
However, the majority of existing probabilistic programming systems
are designed for specifying generative models and estimating
parameters given observations.
They do not support complex database queries that combine tabular data
with generative models specified by probabilistic programs.

\paragraph{\iql}
This article introduces \iql{}, a novel probabilistic programming
system for querying generative models of database tables.
\iql{} is structured as a declarative extension to SQL which seamlessly
enables queries that integrate access to the tabular data with
operations against the probabilistic model.
Examples include predicting new data, detecting anomalies, imputing
missing values, cleaning noisy entries, and generating synthetic
observations~\citep{gelman2020,gabry2019visualization,saad2021sppl,lew2021pclean}.
\iql{} introduces a novel interface and soundness guarantees that
decouple user-level specification of high-level queries against
probabilistic models from low-level details of probabilistic
programming, such as probabilistic modelling, inference algorithm
design, and high-performance machine implementations.
\iql{} extends SQL with several constructs:

\begin{itemize}
    \item To complement$\select$clauses that \emph{retrieve} existing
    records from a table, \iql{} includes the clause$\generate m$ to
    \textit{generate} synthetic records from a probabilistic
    model $m$.

    \item To complement$\where$clauses that \emph{filter} data via
    constraints, \iql{} introduces the clause $m \given e$
    to \emph{condition} a probabilistic model $m$ on an event (i.e., a
    set of constraints) $e$.

    \item To complement \textit{joins} between tables, \iql{} introduces
    a new \emph{mixed join} clause $t\ \generative$ $\join m$ to join each row
    of a data table $t$ with a synthetic row generated from a
    probabilistic model $m$, whose generation can be conditioned in a
    per-row fashion on the values of $t$.

    \item To complement arithmetic expressions, \iql{} introduces
    $\probof e$ $\under m$ expressions, which
    compute the probability (density) of an event
    $e$ under a probabilistic model $m$.
\end{itemize}

\begin{figure}[t]
\begin{adjustbox}{max width=\linewidth}
\begin{tikzpicture}
\tikzstyle{question} = [inner sep=0pt, rectangle, font=\bfseries\footnotesize]
\tikzstyle{result} = [rectangle, draw=none, inner sep=1pt, font=\scriptsize]
\tikzstyle{query} = [draw=none, rectangle, text width=3cm, align=center, font=\scriptsize]
\tikzstyle{iql} = [draw=none,outer sep=0pt, inner sep=0pt, rectangle, text width=3cm]

\node[name=planner,rectangle, draw=black,fill=lightgray, minimum width=11cm,]{\iql{} Planner};

\node[question,name=imputation,above = 3.75 of planner.north west,anchor=north west,]{\begin{tabular}{@{}c@{}}Imputation\\\phantom{x} \end{tabular}};
\node[query,name=imputation-query, below = 0 of imputation]{
  Guess all missing values \\ of the ``Experience'' field \\ with confidences. \phantom{p}};
\node[iql,name=imputation-iql,draw=none,outer sep=0pt, inner sep=0pt, rectangle, text width=2.8cm,
  at={($(imputation-query)!0.5!(planner-|imputation)$)}
  ]{
\begin{lstlisting}[style=iql,aboveskip=0pt,belowskip=0pt]
SELECT exp AS Experience, 
  PROBABILITY OF exp UNDER 
  model GIVEN * AS Pr_Exp
FROM data GENERATIVE JOIN 
model GIVEN *
\end{lstlisting}};
\node[result,name=imputation-result, below = 3 of imputation-query]{
\begin{tabular}{|c|c|}
\hline
Experience & Pr\_Exp \\ \hline\hline
5 years & 0.6 \\
3 years & 0.3 \\
1 year  & 0.1 \\\hline
\end{tabular}};

\node[question,name=anomaly, above = 3.75 of planner.north east,anchor=north east]{\begin{tabular}{@{}c@{}}Anomaly\\ Detection \end{tabular}};
\node[query,name=anomaly-query, below = 0 of anomaly]{
  Who are the 3 \\ most anomalous \\ respondents?};
\node[iql,name=anomaly-iql,
  at={($(anomaly-query)!0.5!(planner-|anomaly)$)}
  ]{
\begin{lstlisting}[style=iql,aboveskip=0pt,belowskip=0pt]
SELECT *, PROBABILITY OF 
  * AS Pr_Data
UNDER model FROM data
ORDER BY ASC Pr_Data 
LIMIT 3
\end{lstlisting}};
\node[result,name=anomaly-result, below = 3 of anomaly-query]{
\begin{tabular}{|c|c|c|}
\hline
name & \ldots & Pr\_Data \\ \hline\hline
Tom & \dots & 0.11 \\
Lisa & \ldots & 0.15\\
John & \ldots & 0.27\\
\hline
\end{tabular}};

\node[question,name=prediction,
  at={($(imputation)!0.325!(anomaly)$)}
  ]{\begin{tabular}{@{}c@{}}Prediction \phantom{y}\\\phantom{x} \end{tabular}};
\node[query,name=prediction-query, below = 0 of prediction]{
  How likely is it that \\ a developer from \\ Seattle knows Rust? \phantom{p}};
\node[iql,name=prediction-iql, text width=2.6cm,
  at={($(prediction-query)!0.5!(planner-|prediction)$)}
  ]{
\begin{lstlisting}[style=iql,aboveskip=0pt,belowskip=0pt]
SELECT 
  PROBABILITY OF Rust
UNDER model
GIVEN city = "Seattle"
\end{lstlisting}};
\node[result,name=prediction-result,below = 3 of prediction-query]{
\begin{tabular}{|c|c|}
\hline
value & probability \\ \hline\hline
Yes  & 0.78 \\
No   & 0.22 \\ \hline
\end{tabular}};

\node[question,name=synthetic,
  at={($(imputation)!0.65!(anomaly)$)}
    ]{\begin{tabular}{@{}c@{}}Synthetic Data\\ Generation \end{tabular}};
\node[query,name=synthetic-query, below = 0 of synthetic]{
  Generate 1,000 \\\ synthetic rows for \\ women in MA.\phantom{p}};
\node[iql,name=synthetic-iql,text width=2.6cm,
  at={($(synthetic-query)!0.5!(planner-|synthetic)$)}
  ]{
\begin{lstlisting}[style=iql,aboveskip=0pt,belowskip=0pt]
SELECT age, salary FROM
GENERATE UNDER model
GIVEN gender = "Woman"
AND state = "MA"
LIMIT 1000
\end{lstlisting}};
\node[result,name=synthetic-result,below = 3 of synthetic-query]{
\begin{tabular}{|c|c|}
\hline
age & salary \\ \hline\hline
26  & 61k \\
31   & 77k \\
\dots & \dots \\\hline
\end{tabular}};

\begin{scope}[on background layer]
\foreach \x in {imputation, prediction, synthetic, anomaly} {
  \draw[thick] (\x-query) -- (\x-iql.north);
  \draw[thick,-latex] (\x-iql.south) -- (\x|-planner.north);
  \draw[thick,-latex] (\x|-planner.south) -- (\x-result);
}
\end{scope}

\node[name=data, right=.95 of planner, yshift=9em, font=\bfseries\footnotesize]{
\begin{tabular}{@{}c@{}}    
Data\\Table
\end{tabular}};

\node[name=model, right=.5 of planner, yshift=-1.1em ,font=\bfseries\footnotesize]{\begin{tabular}{@{}c@{}}Probabilistic\\Program\end{tabular}};
\begin{scope}[transform canvas={yshift=-.1em}]
\draw[-latex,thick]
  ([yshift=-2pt]model.west |- planner.east) -- ([yshift=-2pt]planner.east);
\end{scope}
\begin{scope}[transform canvas={xshift=.5em}]
\draw[-latex,dashed,thick] (data) -- node[pos=.5,rotate=-90,anchor=south]{\tiny program synthesis (optional)} (model.north);
\end{scope}
\draw[-latex,thick] (data.south) |- (planner.1);
\tikzstyle{heading} = [fill=lightgray,text=black, minimum height=1cm]
\node[name=heading-question,heading, anchor=north,at={([xshift=-1.75cm]planner.west |- anomaly.east), }]{
  \begin{tabular}{@{}c@{}}User\\Question\end{tabular}};

\node[name=heading-iql,heading, at={([xshift=-1.75cm]planner.west |- anomaly-iql)}]{
  \begin{tabular}{@{}c@{}}\iql{}\\ \hspace{.5em}Query\hspace{.5em}  \end{tabular}};

\node[name=heading-result,heading, at={([xshift=-1.75cm]planner.west |- anomaly-result)}]{\begin{tabular}{@{}c@{}}  \hspace{.2em}Answer\hspace{.2em} \end{tabular}};

\draw[line width=2pt,-latex] (heading-question) -- (heading-iql);
\draw[line width=2pt,-latex] (heading-iql) -- (heading-result);

\end{tikzpicture}
\end{adjustbox}
\captionsetup{skip=4pt}
\caption{Overview of \iql{}.}
\vspace{-1em}
\label{fig:schematic}
\end{figure}

In this work, we assume that an existing probabilistic program
synthesis tool has been used to automatically generate a probabilistic
model of the user's data satisfying a certain formal interface.
The user then uploads the data and model to \iql{} which
automatically integrates them.
The user can then issue queries for a variety of tasks, as illustrated
in \cref{fig:schematic}.
Although we envision most users using automatically discovered models
on their data, the \iql{} implementation also supports
hand-implemented or partly-learned probabilistic models.
For instance, a user can develop custom models for
harmonization across different sources, as shown in
\cref{sub:gen-clj-code-for-emission-functions}.

The core of \iql{} is formalized as a simply-typed extension of SQL
(\cref{sub:language}).
This extension includes standard SQL scalar expressions and tables as
well as \emph{\tableModels} (probabilistic models of tables) and
\emph{events} (a set of constructs that allow users to issue
probabilistic queries that leverage Bayesian conditioning).
Together, \tableModels{} and events enable a seamless integration of
standard SQL databases with probabilistic models, which include
queries that interleave accesses to the database records and
probabilistic models.

The \iql{} \emph{query planner} (\cref{sec:impl-table-model}) lowers
queries into plans that execute against a new model interface for
probabilistic models of tabular data.
This \emph{Abstract Model Interface} (\ami{}) (\cref{sub:\ami{}})
provides a unifying specification of probabilistic models that are
compatible with \iql{}.
To implement the \ami{}, the model must be able to:
\begin{enumerate*}[label=(\roman*)]
\item generative samples from a (potentially approximate) conditional
distribution;
\item compute probability densities for specified points;
\item compute probabilities of sets in the support of the
conditional distribution.
\end{enumerate*}

The open source \iql{} system includes a number of implementations of
the \ami{}, including
\begin{itemize}
\item a Clojure implementation~\citep{GenCLJ} of Gen~\citep{towner2019gen}, a general
purpose probabilistic programming language; see \cref{sub:gen-clj-code-for-emission-functions} for an example.

\item models produced by CrossCat~\citep{mansinghka2016crosscat}, a probabilistic program synthesis tool;

\item SPPL~\citep{saad2021sppl}, a probabilistic programming language
for exact inference.
\end{itemize}

We provide a measure-theoretic denotational semantics for the language
(\cref{sub:semantics}).
This semantics captures the interaction between deterministic SQL
operations and probabilistic operations on the probabilistic model,
enabling us to prove several correctness guarantees that query results
satisfy.
Specifically, we prove guarantees for \begin{enumerate*}[label=(\roman*)]
\item the \emph{exact} case, where exact inference about
marginal and conditional distributions of the probabilistic model
is available (\cref{thm:exact-\ami{}-guarantee}); and
\item a range of \emph{approximate} cases, where answers to marginal
and conditional queries are obtained via approximate inference
algorithms (\cref{thm:consistent-ami-guarantee}).
\end{enumerate*}

We benchmark \iql{} on a set of representative queries, testing
the runtime performance, overhead of the query planner, and effect of
our optimizations.
The results show that all queries execute in milliseconds against data
tables of sizes up to 10,000 rows, with a speedup in the range
$1.7$--$6.8\mathrm{x}$ against the most closely related baseline, and
that the query planner's overhead as compared to hand-written code is
small.
We evaluate our system on two case studies to test its applicability
to solving real-world problems (conditional synthetic data generation
for a virtual wet lab and an anomaly detection in clinical trials),
comparing against a generalized linear model (GLM) and a conditional
tabular generative adversarial network (CTGAN~\citep{xu2020modeling})
baseline.

\paragraph{Contributions}
This paper makes the following contributions:
\begin{enumerate}
\item The \textbf{\iql{} language} (\cref{sub:language}), an extension
of SQL with probabilistic models of tabular data as first-class
constructs and probabilistic constructs to allow the integration of
queries on these models with queries on the data.

\item A \textbf{unifying abstract interface for models of tabular
data} (\cref{sub:\ami{}}), which bridges the query language and
probabilistic models of database tables, to which all models must
conform. The query planner lowers \iql{} queries on models to
queries on this interface.

\item \textbf{Soundness theorems}, which fall into two classes:

\begin{itemize}

\item {\bf Exact:} We show that if models satisfy the exact interface,
all deterministic computations will be exact
(\cref{thm:exact-\ami{}-guarantee}). This theorem works with an exact
denotational semantics (\cref{sub:exact-backend}) that precisely
characterize the behavior of exact models.

\item {\bf Approximate:} If approximate models implement
consistent estimators (i.e., estimators that converge to the true
value), we prove that all queries return consistent results
(\cref{thm:consistent-ami-guarantee}). This theorem works with a novel
denotational semantics that combines measure-theoretic aspects with
sequences of random variables.
\end{itemize}

Together, these guarantees highlight some of the tradeoffs between
using an exact model (which deliver stronger guarantees but may be
difficult to obtain in some use cases) and an approximate model (which
deliver weaker guarantees but are more easily available).

\item An \textbf{open-source implementation} of \iql{} in
 Clojure (\url{https://github.com/OpenGen/GenSQL.query}), which can be compiled into JavaScript and
 run natively in the browser.

\item A \textbf{performance evaluation} of our approach
(\cref{sec:perf-eval}) which establishes that \iql{} is competitive
with hand-coded implementations and gives improved performance over a
competitive baseline. Two case studies further demonstrate the utility
of \iql{}.
\end{enumerate}

\section{Example}
\label{sec:example}

\begin{figure}[b]
\small
\begin{minipage}{0.9\textwidth}
\begin{lstlisting}[style=iql-small, firstnumber=1,numbers=left,stepnumber=1]
SELECT weight, AVG(log_pxy_div_px_py) AS mutual_information
FROM (
  SELECT weight, LOG(pxy) - (LOG(px) + LOG(py)) AS log_pxy_div_px_py
  FROM (
    SELECT weight,
      PROBABILITY OF h_model.age = table.age AND h_model.bmi = table.bmi
        UNDER h_model GIVEN h_model.weight = table.weight AS pxy,
      PROBABILITY OF h_model.age = table.age
        UNDER h_model GIVEN h_model.weight = table.weight AS px,
      PROBABILITY OF h_model.bmi = table.bmi
        UNDER h_model GIVEN h_model.weight = table.weight AS py
    FROM (
      SELECT table.weight, table.age, table.bmi
      FROM (
        health_data DUPLICATE 1000 TIMES
        GENERATIVE JOIN h_model
        GIVEN h_model.weight = health_data.weight) AS table)))
GROUP BY weight\end{lstlisting}
\end{minipage}
\captionsetup{skip=4pt}
\caption{Estimating the conditional mutual information between \texttt{age} and \texttt{bmi} given patient weights.}
\label{fig:example-cmi-section2}
\end{figure}

\Cref{fig:example-cmi-section2} presents an example \iql{} query. 
In this example, we work with a probabilistic model
(\texttt{health\_model}) derived from a national database of patient
information, as well as a data table (\texttt{health\_data}) from a
set of local hospitals.
The query uses the probabilistic model to estimate the mutual
information---an information-theoretic measure used in data
analysis--- between the \texttt{age} and \texttt{bmi} columns (from
the probabilistic model) for specific valueso f patient weights
(selected from the data table).
The mutual information is a statistical measure of the strength of the
association between these two columns, defined as a sum or integral,
over the joint distribution of \texttt{age} and \texttt{bmi}, of the
logarithm of the ratio of the joint density and the product of the
marginal density.

The query estimates the mutual information by Monte Carlo integration,
i.e., it approximates the integral by sampling.
We first generate 1000 copies of each row in the \texttt{health\_data}
table (line 15) and then use the \iql{} {\em generative join}
construct (line 16) to complete each row as follows.
For each such row $r$:
\begin{enumerate}
  \item a row $r'$ is sampled from a version of the model conditioned
  on the \texttt{weight} value of row $r$;
  \item the rows $r$ and $r'$ are concatenated.
\end{enumerate}
The resulting intermediate table is called \texttt{table} (line 17). 
Each synthetic row $r'$ is used as a sample for the Monte Carlo
integration of the conditional mutual information for the
corresponding \texttt{weight} value.
From this intermediate table, we select the \texttt{weight},
\texttt{age}, and \texttt{bmi} columns (line 13).
Note that the \texttt{weight} column comes from the patient data while
the \texttt{age} and \texttt{bmi} columns come from the rows sampled
from the probabilistic model.

For each \texttt{weight} in the patient data, we compute the Monte
Carlo approximation of the mutual information between \texttt{age} and
\texttt{bmi} for that weight as
\begin{equation*}
\textstyle\frac{1}{1000k}\sum_{i=1}^{1000k}
  \log \frac{p(\texttt{age}_i,\texttt{bmi}_i)}{p(\texttt{age}_i)p(\texttt{bmi}_i)},
  \end{equation*}
where $k$ is the number of patients with that specific weight, and
$(\texttt{age}_i,\texttt{bmi}_i)$ is a sample from the model for that
weight.
To do so, lines 6--11 compute the probability densities
$p(\texttt{age}_i,\texttt{bmi}_i)$ (lines 6--7), $p(\texttt{age}_i)$
(lines 8-9), and $p(\texttt{bmi}_i)$ (lines 10--11).
For instance, the \texttt{GIVEN} clause conditions the model on the
\texttt{weight} column of the model being equal to the \texttt{weight}
column of \texttt{table} (line 11).
Line 10 then computes the probability density that the \texttt{bmi}
column of the conditioned model is equal to the corresponding column
in \texttt{table}.
\iql{} computes these probability densities by invoking the $\logpdf$
function in the probabilistic model interface (\cref{sub:\ami{}}).

A traditional SQL select statement (line 5) propagates the patient weights
and corresponding probabilities \texttt{pxy}, \texttt{px}, and
\texttt{py} to generate a table with four columns: the \texttt{weight}
and the corresponding probability densities for that weight.
Line 3 computes $\log
{p(\texttt{age}_i,\texttt{bmi}_i)}-\log{p(\texttt{age}_i)p(\texttt{bmi}_i)}$
for each of the rows, naming this ratio
\texttt{log\_pxy\_div\_px\_py}.
Note that there are $1000k$ \texttt{log\_pxy\_div\_px\_py} values for
each weight in the local patient data, where $k$ is the number of
patients with that specific weight.
Finally, line 1 computes the mutual information estimate between
\texttt{age} and \texttt{bmi}, for each weight, as the average of the
\texttt{log\_pxy\_div\_px\_py} values for that weight. 
This example illustrates the expressivity of \iql{}, but we note that our implementation has a primitive which directly estimates conditional mutual information without the need to materialize intermediate tables.

\section{Syntax and Semantics}
\label{sec:language}

\subsection{Language}
\label{sub:language}

The core calculus extending SQL for querying from probabilistic models
of tabular data is given in \cref{fig:syntax-iql}, and the
type-system is given in \cref{fig:type-system-iql}.\footnote{
The core SQL-part of the language is minimal for expository purposes.
\Cref{sec:full-language} presents the formalization for a
richer language, including the operations$\group\by$and$\duplicate$used in \cref{fig:example-cmi-section2}.
}
As SQL is a subset of \iql{}, this calculus also includes a
simply-typed formalization of SQL where terms are given in a pair of
context: a local and a global one.
We found this formalization interesting in its own right, as we could
not find an equivalent formalization in the programming languages
literature.

There are several noteworthy differences with variables and contexts
from traditional simply-typed languages based on the
lambda-calculus, which are explained below.
We note early that we distinguish two types of conditioning through constructs called events and events-0. 
Events-0 come from a technical difficulty well-known in the PPL and measure theory literature \citep{shan2017,wu2018discrete} when conditioning on a continuous variable taking a specific value. This creates a \emph{possible event of probability 0}, and requires special treatment.

\begin{figure*}[t]
    \captionsetup{skip=4pt}
\setlength{\tabcolsep}{0.4em}
\centering
\begin{adjustbox}{max width=\textwidth}
\begin{tabular}{|l|l c l|l c l|}
\hline
\multicolumn{1}{|c|}{\textbf{Description}} & \multicolumn{3}{c|}{\textbf{SQL}} & \multicolumn{3}{c|}{\textbf{Probabilistic Extension}} \\\hline \hline
Base/Event Type & $\sigma$ & $::=$ & $\sigma_c \gor
\sigma_d$
& $\evtype$ & $::=$ & $\eventtype\tabty{\col_1: \sigma_1,\ldots,\col_n: \sigma_n}$ \\
&&&
&&& $\gor\densitytype\tabty{\col_1: \sigma_1,\ldots,\col_n: \sigma_n}$  \\
Table/\TableModel{} Type & $\tabtype$ & $::=$ & $T[\id]\tabty{\col_1: \sigma_1,\ldots,\col_n: \sigma_n}$ & $\modtype$ & $::=$ & $M[\id]\tabty{\col_1: \sigma_1,\ldots,\col_n: \sigma_n}$ \\ \hline
Table Expression & $t$ & $\Coloneq$ & $\id \gor \rename t \as \id$ & $t$ & $::=$ & $\ldots$ \\
& && $\gor t_1\ \join\ t_2 \gor t\ \where\ e$ &&& $\gor \generate m \limit e$ \\
&&& $\gor \select \widebar{e} \as \widebar{\col} \from t$&&& $\gor t \modjoin m$\\
\hline
\TableModel{}  Expression &&&& $m$ & $::=$ & $\id \gor m \given c^i
    \gor \rename m \as \id $ \\
     \hline
Scalar Expression  & $e$ & $::=$ & $\id.\col \gor \op(e_1,\ldots,e_n)$ & $e$ & $::=$ & $\ldots \gor \probof c^i \under m $  \\
\hline
Event Expressions &&&& $c^1$ & $::=$ & $c^1_1 \wedge c^1_2 \gor c^1_1 \vee c^1_2 \gor \id.\col\ \op\ e$ \\
Event-0 Expressions &&&&   $c^0$ & $::=$ & $c^0_1 \wedge c^0_2 \gor \id.\col = e$ \\ \hline
\end{tabular}
\end{adjustbox}
\begin{adjustbox}{max width=\textwidth}
\begin{tabular}{|lc|}
    \hline
    Primitive Domains: & $\op \in \opset$, $\id\in\idset$, $\col\in\colset$, $\sigma_c \in \{\real, \posreal, \ranged(a,b), \ldots\}$, $\sigma_d \in \{\inte, \str, \nat, \bool, \ldots\}$
    \\ \hline
    Syntactic Sugar: &  $\protect\widebar{e} \as \protect\widebar{\col}\equiv e_1 \as \col_1,\ldots,e_n\as \col_n$.
    \\ \hline
    \end{tabular}
\end{adjustbox}
\caption{Syntax of \iql{}.}
    %
\label{fig:syntax-iql}
\end{figure*}

\paragraph{Names and Identifiers.} We assume a countable set
$\colset$ of names $\col\in\colset$ for the columns of tables and
\tableModels, as well as a countable set $\idset$ of identifiers $\id\in\idset$
for naming tables and \tableModels.

\paragraph{Base types.}  Cells of tables can have a base type $\sigma$,
which is either a continuous type $\sigma_c$ or a discrete type $\sigma_d$.
Continuous types are $\real$ for reals, $\posreal$ for
non-negative reals, or $\ranged(a,b)$ for reals in the range $[a,b]$.
Discrete types are $\nat$ for natural numbers, $\inte$ for
integers, $\str$ for strings,
$\categorical(\textsc{n}_1,\ldots,\textsc{n}_k)$ for a categorical
type over $k$ attributes, and $\bool$ for Boolean.

\paragraph{Table and \tableModel{} types.} We denote these types by $D[?\id]\{\col_1:\sigma_1,\ldots,\col_n:\sigma_n\}$. 
$D$ is either $T$ for tables or $M$ for \tableModels{}. $?\id$ is an optional identifier, allowing access to columns of a table or \tableModel{} in a query.
For instance, in$\select \id.\text{weight}$, the identifier $\id$ refers to a table and $\text{weight}$ to a column of that table. 
The identifier can be be optional, e.g. there is no default identifier for a table created after a join.
The notation $\{\col_1:\sigma_1,\ldots,\col_n:\sigma_n\}$ indicate that the table has columns $\col_i$ of type $\sigma_i$. 
Therefore, we can think of each row of a table as an element of a record type $\{\col_1:\sigma_1,\ldots,\col_n:\sigma_n\}$, a bag of rows as a table, and a \tableModel{} as a row generator. 

\begin{figure}[t]
    \captionsetup{skip=4pt}
\footnotesize
\FrameSep=4pt
\begin{framed}
\begin{subfigure}[t]{.5\linewidth}
\centering
\caption{\textbf{Type System for Table Expressions}}
\begin{tabular}{c}
   \hline
   $\Gamma, T[\id]\{\cols\}; \Delta \vdash \id: T[\id]\{\cols\}$
\end{tabular}
\smallskip

\begin{tabular}{c}
   $\ctx \vdash t:T[?\id]\{\cols\}$\quad    $\id'$ fresh
   \\ \hline
   $\ctx \vdash \rename t\as \id':T[\id']\{\cols\}$
\end{tabular}
\smallskip

\begin{tabular}{c}
     $\ctx \vdash t_1:T[?\id_1]\{\cols\}$ \\ $\ctx \vdash t_2:T[?\id_2]\{\cols'\}$ \quad
     $\cols\cap\cols'=\emptyset$
   \\ \hline
   $\ctx \vdash t_1 \join t_2:T[]\{\cols,\cols'\}$
\end{tabular}
\smallskip

\begin{tabular}{c}
   $\ctx \vdash t:T[\id]\{\cols\}$ \quad  $\ctx, T[\id]\{\cols\} \vdash e:\bool$
   \\ \hline
   $\ctx \vdash t \where e:T[\id]\{\cols\}$
\end{tabular}
\smallskip

\begin{tabular}{c}
   $\ctx \vdash m:M[\id]\{\cols\}$    \quad
   $\ctx \vdash e:\nat$
   \\ \hline
   $\ctx \vdash  \generate m$ \\
   $\quad\qquad \limit e :T[]\{\cols\}$
\end{tabular}
\smallskip

\hspace{-2mm}
\begin{tabular}{c}
   $\ctx \vdash t:T[\id]\{\cols\}$ \quad  $\cols\cap\cols'=\emptyset$ \\
   $\ctx, T[\id]\{\cols\} \vdash m:M[\id']\{\cols'\}$
   \\ \hline
   \hspace{-1mm}$\ctx \vdash  t \modjoin m:T[]\{\cols,\cols'\}$
\end{tabular}
\smallskip

\oldcons{
\begin{tabular}{c}
    $\ctx \vdash t: T[\id]\{\cols\}$ \\
    $\ctx, T[\id]\{\cols\} \vdash e_i: \sigma_i$ for $1\leq i\leq n$ \\
    $\widebar{e} \as \widebar{\col} := e_1 \as \col_1,\ldots,e_n\as \col_n $
    \\ \hline
    $\ctx\vdash \select \widebar{e} \as \widebar{\col} $\\
    $\quad\from t: T[]\{\col_1:\sigma_1,\ldots,\col_n: \sigma_n\}$
\end{tabular}
\smallskip
}

\newcons{
\begin{tabular}{c}
    $\ctx \vdash t: T[\id]\{\cols\}$ \quad
    $\ctx \vdash m: M[\id']\{\cols'\}$ \quad $\cols,\cols'$ disjoint names\\
    $\ctx, T[\id]\{\cols\}, M[\id']\{\cols'\} \vdash e_i: \sigma_i$ for $1\leq i\leq n$
    \\ \hline
    $\ctx\vdash \select e_1 \as \col_1,\ldots,e_n \as \col_n \from \ttable t, \model m: T[]\{\col_1:\sigma_1,\ldots,\col_n: \sigma_n\}$
\end{tabular}
\smallskip
}

\caption{\textbf{Type System for Event Expressions}}

\begin{tabular}{c}
     $\ctx \vdash e:\sigma_i$
     \quad $i\in\{1,\ldots,n\}$\\
      $\op \in \{<,>,=\}$
     \quad  $\forall \sigma_c. (\sigma_i,\op) \neq (\sigma_c,=)$ \\
$\cols=\col_1:\sigma_1,\ldots,\col_i:\sigma_i,\ldots,\col_n:\sigma_n$
     \\ \hline
      $\ctx, M[\id]\{\cols\} \vdash \id.\col_i\ \op\ e:\eventtype\{\cols\}$
\end{tabular}
\smallskip

\begin{tabular}{c}
     $\ctx \vdash c^1_1:\eventtype\{\cols\}$
     \quad
     $\ctx \vdash c^1_2:\eventtype\{\cols\}$
     \\ \hline
      $\ctx \vdash c^1_1 \wedge c^1_2:\eventtype\{\cols\}$
\end{tabular}
\smallskip

\begin{tabular}{c}
     $\ctx \vdash c^1_1:\eventtype\{\cols\}$
     \quad
     $\ctx \vdash c^1_2:\eventtype\{\cols\}$
     \\ \hline
      $\ctx \vdash c^1_1 \vee c^1_2:\eventtype\{\cols\}$
\end{tabular}

\end{subfigure}%
\begin{subfigure}[t]{.5\linewidth}
\centering
\caption{\textbf{Type System for \tableModel{}  Expressions}}
\begin{tabular}{c}
   \hline
   $\Gamma, M[\id]\{\cols\}; \Delta \vdash \id: M[\id]\{\cols\}$
\end{tabular}
\smallskip

\begin{tabular}{c}
   $\ctx \vdash m:M[\id]\{\cols\}$ \\
   $\ctx, M[\id]\{\cols\} \vdash c^1:\eventtype\{\cols\}$
   \\ \hline
   $\ctx \vdash m \given c^1 :M[\id]\{\cols\}$
\end{tabular}
\smallskip

\begin{tabular}{c}
   $\ctx, M[\id]\{\cols\} \vdash c^0:\densitytype\{\cols'\}$
   \\ \hline
   $\ctx \vdash \id \given c^0 :M[\id]\{\cols\}$
\end{tabular}
\smallskip

\begin{tabular}{c}
   $\ctx \vdash m:M[\id]\{\cols\}$  \quad $\id'$ fresh
   \\ \hline
   $\ctx \vdash \rename m \as \id':M[\id']\{\cols\}$
\end{tabular}
\smallskip

\caption{\textbf{Type System for Scalar Expressions}}
\begin{tabular}{c}
  $i \in \{1,\ldots,n\}$ \\ \hline
$\ctx, T[\id]\{\col_1:\sigma_1,\ldots,\col_n:\sigma_n\} \vdash \id.\col_i:\sigma_i$
\end{tabular}
\smallskip

\newcons{
\begin{tabular}{c}
  $i \in \{1,\ldots,n\}$ \\ \hline
$\ctx, M[\id]\{\col_1:\sigma_1,\ldots,\col_n:\sigma_n\} \vdash \id.\col_i:\sigma_i$
\end{tabular}
\smallskip
}

\begin{tabular}{c}
   $\ctx \vdash m:M[\id]\{\cols\}$    \\
   $\ctx, M[\id]\{\cols\} \vdash c^1:\eventtype\{\cols\}$
   \\ \hline
   $\ctx \vdash  \probof c^1 \under m: \ranged(0,1)$
\end{tabular}
\smallskip

\begin{tabular}{c}
    $\ctx \vdash m:M[\id]\{\cols\}$\quad $\vars(c^0)\cap \condvars(m)=\emptyset$    \\ 
    $\ctx, M[\id]\{\cols\} \vdash c^0:\densitytype\{\cols'\}$
    \\ \hline
    $\ctx \vdash  \probof c^0 \under m: \posreal$
 \end{tabular}
 \smallskip

\begin{tabular}{c}
  $\ctx \vdash e_i:\sigma_i$ for $1\leq i\leq n$ \quad $\op:\sigma_1,\ldots,\sigma_n\to \sigma$
   \\ \hline
     $\ctx \vdash \op(e_1,\ldots,e_n):\sigma$
\end{tabular}
\smallskip

\caption{\textbf{Type System for Event-0 Expressions}}
\begin{tabular}{c}
     $\ctx \vdash e:\sigma$\quad $i\in \{1,\ldots,n\}$\\
     $\cols = \ldots,\col_i:\sigma,\ldots$
     \\ \hline
      $\ctx, M[\id]\{\cols\} \vdash
      \id.\col_i\ =\ e:\densitytype \{\col_i:\sigma\} $
\end{tabular}
\smallskip

\begin{tabular}{c}
     $\ctx \vdash c^0_1:\densitytype\{\cols\} \quad \ctx \vdash c^0_2:\densitytype\{\cols'\}$\\
     $\cols\cap\cols'=\emptyset$ \\ \hline
      $\ctx \vdash c^0_1 \wedge c^0_2:\densitytype\{\cols,\cols'\}$
\end{tabular}

\end{subfigure}
\end{framed}
\caption{Type system of \iql{}.}
\label{fig:type-system-iql}
\end{figure}

\paragraph{Scalar Expressions.}
$\opset$ is a set of primitive operations on base
types including standard operations such as $+,*,<,>,=$ on integers and reals, $\wedge,\vee$ on Booleans, as well as constants for every value of a
base type.
For any $\op\in\opset$, we write $\op:\sigma_1,\ldots,\sigma_n\to\sigma$
if $\op$ has arity $n$, takes arguments of base types
$\sigma_1,\ldots,\sigma_n$ and returns a value of type $\sigma$. 
In particular, operations with no arguments are constants of the appropriate type such as $\true$ and $\false$ at the boolean type.
All base types have an additional constant $\Null$ representing a
missing value.
This constant is preserved by primitive operations (e.g.
$\Null+ 3 \mapsto \Null$, $\Null * 4.1 \mapsto \Null$).
By convention, $\where \Null$ clauses act as $\where \false$.

\paragraph{Table expressions.}
Apart from typical SQL operations, we have two ways to generate synthetic data.
$\generate$returns a synthetic table with a given number of rows specified by the$\limit$clause, where each row is generated by sampling from a given \tableModel.
$\modjoin$takes a \tableModel{} and a table, and returns a synthetic
table with the same number of rows, where each row is generated by
concatenating the current row of the table with a sample from the
\tableModel{}. 
The model generating the samples can be conditioned on the current row of the table. 
$\rename$renames a table or \tableModel{} with a new identifier, therefore changing the identifier in its type and the way to access their column in a select of event clause. 

\paragraph{Event and event-0 expressions.}
Events are Boolean expressions on tables and \tableModels, which include equality on discrete values but not on continuous values, which is reserved for events-0.
The only probability 0 events are impossible under a given model, e.g. $x > 6 \wedge x < 3$, and those do not require a separate treatment.
Events and events-0 are used in the$\probof$and$\given$constructs.
$\probof$ takes an event (or event-0) expression and a \tableModel{} to query and returns the probability (or probability density) of the
event under the model.\footnote{
It may be confusing for people familiar with probabilistic programming languages (PPLs) to use$\probof$for both a probability mass and a probability density. 
Our implementation has two versions of the syntax: a strict one and a permissive one.
The strict syntax distinguishes between the two, and in particular on events-0 one the primitive is $\colorsyntaxiql{PROBABILITY}\ \colorsyntaxiql{DENSITY}\ \colorsyntaxiql{OF}$.
The permissive syntax allows to use$\probof$for both, and the system will automatically choose the right version based on the type of the event.  
}

\paragraph{RowModel expressions.}
$\given$takes a \tableModel{} and an event (or event-0) expression, and returns a new \tableModel{}, the conditional distribution of the original
\tableModel{} on the event.
The event expression can be given by a list of inequalities on \emph{arbitrary variables} and equalities on discrete variables, in which case$\given$ acts as a set of constraints on the possible returned values of the model. 
Otherwise, the event expression can be a set of equalities on possibly continuous values and is understood as conditioning the model on the given values.

\paragraph{Contexts.} Expressions are typed in a pair of contexts
$\ctx$ containing table and \tableModel{} types.
As these types include identifiers, there is no need for the more
classical notation $x:\tau$ pairing a variable with its type.
$\Gamma$ is a \emph{set} of types, while $\Delta$ is an
\emph{ordered list} of types. In the premise of a typing rule such as$\probof$,
only the last element of $\Delta$ will be accessible to an expression.
We denote the empty context by $[]$.
Intuitively, $\Gamma$ contains the ambient tables in the database
schema and any loaded models, and within $\Gamma$, all identifiers $\id$ are assumed distinct.
$\Delta$ is the \emph{value environment}, and contains only tables
that are ``in scope'' for a particular expression.
Scalar, event, and event-0 expressions all depend on the value
environment.
If an expression has a table in scope, it will be iterated over the
rows of that table and can only access the current row.
If a$\probof$ expression has a \tableModel{} in scope, it will query
the model for the probability of an event under that model.
If it's a$\given$ expression, it will condition that model on an
event.

\paragraph{Typing rules.}
Judgments are of the form $\ctx \vdash \texttt{e}: \texttt{t}$ where
\texttt{e} is an expression ($t$, $e$, $m$, $c^1$, $c^0$ in \Cref{fig:syntax-iql});
\texttt{t} is a type ($\sigma$, $\mathcal{T}$, $\mathcal{E}$, $\mathcal{M}$ in \Cref{fig:syntax-iql}),
and $\ctx$ is a context.
Given some loaded tables and \tableModels{} forming environment
$\Gamma$, the objects of interest are ``closed expressions''
of table type, i.e., expressions of the form
$\Gamma; [] \vdash t: T[?\id]\{\cols\}$.
``Closed'' here refers to $\Delta$ being empty, not $\Gamma$.
Notable rules include those that need the same identifier twice, such as the$\probof$ or the$\where$ rule.
For instance, in the $t \where e$ rule, where $t$ has identifier $\id$, a valid SQL $e$ would be $\id.\col = 3$ where $\col$ is a column of $t$.
This reflects the fact that the expression $e$ should have access to the identifier $\id$ in its local environment, and that the column $\col$ of $t$ will be iterated over by the expression $e$.

\paragraph{Notations used in the type system}
$?\id$ indicates an optional identifier and $\id$.
''$\id'$ fresh'' means that $\id'$ is not in the contexts $\Gamma$, $\Delta$ or in the type of a subterm of the expression.
We will often abbreviate $\{\col_1:\sigma_1,\ldots,\col_n:\sigma_n\}$
as $\{\cols\}$.
We write $\cols\cap\cols'=\emptyset$ when the set of column names in
$\cols$ and in $\cols'$ should be disjoint.
In the first typing rule for events, we write $\forall \sigma_c. (\sigma_i,\op) \neq (\sigma_c,=)$ to mean that $\op$ cannot be an equality on a continuous type.
We recursively define the following two macros:
{ \setlength{\abovedisplayskip}{0pt}
 \setlength{\belowdisplayskip}{0pt}
\begin{align*}
&\vars(\id.\col~\op~t) = \{\col\}\quad
\vars(c \land c') = \vars(c) \cup \vars(c')\quad
\vars(c \lor c') = \vars(c) \cup \vars(c') \\
&\condvars(m \given c^0) =\vars(c^0) \quad
\condvars(m \given c^1) =\condvars(m) \quad
\condvars(\id) = \emptyset
\end{align*}
}
\paragraph{Restrictions imposed by the type system.}
If the same identifier $\id$ appears twice in the
premise of a typing rule, the two identifiers must equal, and two different identifiers $\id$ and $\id'$ must be distinct.
The$\join$ and$\modjoin$ operations require that the columns of the two tables have disjoint names. 
As explained above, events are disallowed to be equalities on continuous types.
A model can only be conditioned once on an event-0, which is enforced by the restriction $\id\given c^0$.
Events-0 follow a linear typing system to avoid contradictory statements such as $\id.\col = 1.0 \wedge \id.\col = 2.0$.
Events-0 in a$\probof$query on a conditioned model cannot refer to the conditioned columns of the model, which is enforced by the restriction $\vars(c^0)\cap\condvars(m)=\emptyset$.\footnote{Our implementation is less restricted. It allows join variants such as SQL's left join where the tables do not have disjoint columns. It also allows multiple conditionings on the same model, which are then normalized to the form above. See \cref{sub:normalization} for details about the normalization.}

\paragraph{Syntactic sugar.} Our implementation includes various syntactic sugars that are not present in the formalization but which are used in
several figures. 
Given $t{:}T[\id]\{\col_1{:}\sigma_1,\ldots,\col_n{:}\sigma_n\}$, $m{:} M[\id']\{\col_1'{:}\sigma_1,\ldots,\col_n'{:}\sigma_n\}$, we have the following equivalences:
\begin{itemize}[noitemsep,leftmargin=*]
\item $\select * \from t$ $\rightsquigarrow$ $\select
\id.\col_1,\ldots,\id.\col_n \from t$
\item
$\probof *$ $\rightsquigarrow$ $\probof e$ for any query of the form $\select\probof * \under m \given c \from t$, where $e \coloneqq \id'.\col_1' =
\id.\col_1 \wedge \ldots \wedge \id'.\col_n' = \id.\col_n$.
\item $m \given *$ $\rightsquigarrow$ $m \given e$ within a
$\select \from t$ query. The event $e \coloneqq
\id'.\col_{i_1}' = \id.\col_{i_1} \wedge
\ldots \wedge \id'.\col_{i_k}' = \id.\col_{i_k}$, where the
$\col_{i_j}$ are columns $t$ that do not appear in the $\select$
clause.
\item $ * \exceptt \id.\col$ removes the column $\col$ from list of columns that $*$ selects.
\end{itemize}

\subsection{Semantics}
\label{sub:semantics}

We define denotational semantics using measure theory, shown in
\cref{fig:semantics-iql}.
Even though the SQL subset of \iql{} is not probabilistic, our
probabilistic semantics ensures compositional reasoning about the
semantics of SQL queries combined with probabilistic \iql{}
expressions, such as synthetic tables generated by \tableModels.
Per usual, the semantics of expressions is defined compositionally on
typing judgement derivations, and $\sem{\texttt{e}}$ is a shorthand
for $\sem{\Gamma; \Delta \vdash \texttt{e}:\texttt{t}}$.

\paragraph{Base types
(\cref{subfig:semantics-base-type}).} We assign to each type $\sigma$
a measure space $\sem{\sigma}:=(X_\sigma,
\Sigma_{X_\sigma},\nu_\sigma)$ consisting of a set $X_\sigma$, a
sigma-algebra $\Sigma_{X_\sigma}$, and reference measure $\nu_\sigma$.
$\ZZ$ denotes the set of integers, $\NN$ natural numbers, and $\BB$ Booleans, which are equipped with the discrete sigma-algebra.
We equip the reals $\RR$ with the Borel sigma-algebra.
We interpret $\Null$ by adding a fresh element $\{\star\}$ to the standard interpretation of each base type, equipped with the discrete sigma-algebra.
The semantics of a base type $\sigma$ is then given by the smallest
sigma-algebra making $\{\star\}$ measurable, as well as ensuring that
every previously measurable set remains measurable.
(This construction is also called the ``direct-sum sigma-algebra''
\citep[214L]{fremlin2001}).

The base measure on discrete types $\sigma_d$ such as
$\inte,\nat,\bool,\str$ is the counting measure.
On continuous types $\sigma_c$ such as $\real$, the base measure is
the Lebesgue measure $\lambda$.
These are extended to base measures $\nu_\sigma$ on $\sem{\sigma}$ by
using the dirac measure $\delta_{\{\star\}}$ on $\{\star\}$, e.g. the
base measure on $\sem{\RR}$ is $\lambda_\RR+\delta_{\{\star\}}$.
We write $\mu \otimes \nu$ for the product of measures. 
We extend the reference measure to the product space $\prod_{1\leq i\leq
n}\sem{\sigma_i}$ by taking the product of the reference measures
$\nu\coloneq\bigotimes_{1\leq i\leq n}\nu_{\sigma_i}$.

\paragraph{Table types (\cref{subfig:semantics-table-type}).}
Our semantics has two modes of interpreting table types, a ``tuple
mode'' $\tupsem{-}$, and a ``table mode'' $\tabsem{-}$.
$\tabsem{-}$ interprets tables as measures on bags of tuples, while $\tupsem{-}$
interprets a table as a tuple, representing the current row of the
table being processed by a scalar, event or event-0 expression.
More precisely, we denote by $\measm(X)$ the measurable space of probability measures on the standard Borel space $X$ \citep{giry2006categorical}.
The table semantics interprets table types as measures on bags of
tuples $\tabsem{T[?\id]\{\cols\}}=\bag(\tupsem{T[?\id]\{\cols\}})$,
where $\bag(X)=\{f :X\to \NN \mid f(x)=0 \text{ except for finitely
many }x\}$.
$\bag(X)$ is equipped with the least sigma-algebra containing the generating sets $\{b\in\bag(X) \mid b\text{ contains exactly }k\text{ elements in }A\}$ for
measurable sets $A$ of $X$~\citep{dash2021monads}.

\paragraph{Contexts (\cref{subfig:semantics-context}).} We
interpret the global context $\Gamma$ with the table semantics
$\tabsem{-}$ and the local context $\Delta$ with the tuple semantics
$\tupsem{-}$.
We write $\gamma$ for an element of $\tabsem{\Gamma}$, and see it as a
finite map from identifiers to values.
Likewise, we write $\delta$ for an element of $\Delta$.
We write $\delta[\id \mapsto v]$ for the extended finite map mapping
$\id$ to $v$.

\begin{figure}
\captionsetup[subfigure]{belowskip=-10pt, aboveskip=2pt}
\captionsetup{skip=4pt}
\setlength{\abovedisplayskip}{0pt}
\setlength{\belowdisplayskip}{0pt}
\setlength{\abovedisplayshortskip}{0pt}
\setlength{\belowdisplayshortskip}{0pt}

\setlength{\FrameSep}{0pt}
\begin{framed}
\footnotesize
\begin{subfigure}[t]{.6\textwidth}
\caption{\textbf{Semantics of Table and \tableModel{} Types}}
\label{subfig:semantics-table-type}
\begin{align*}
   \tupsem{T[?\id]\tabty{\col_1: \sigma_1,\ldots,\col_n: \sigma_n}} &= \textstyle\prod_{1\leq i\leq n} \sem{\sigma_i} \\
   \sem{M[?\id]\tabty{\col_1: \sigma_1,\ldots,\col_n: \sigma_n}} &= \textstyle\dens\Big(\prod_{1\leq i\leq n} \sem{\sigma_i}\Big)\\
   \tabsem{T[?\id]\tabty{\col_1: \sigma_1,\ldots,\col_n: \sigma_n}} &= \textstyle\measm\bag\Big(\prod_{1\leq i\leq n} \sem{\sigma_i}\Big) \\
   \tabsem{M[?\id]\tabty{\col_1: \sigma_1,\ldots,\col_n: \sigma_n}} &= \textstyle\admissible\Big(\prod_{1\leq i\leq n} \sem{\sigma_i}\Big)
\end{align*}
\end{subfigure}
\begin{subfigure}[t]{.32\textwidth}
\caption{\textbf{Semantics of Contexts}}
\label{subfig:semantics-context}
\begin{align*}
    \tupsem{\Delta := \tabtype,\Delta'} &= \tupsem{\tabtype} \times \tupsem{\Delta'} \\
  \tupsem{\Delta := \modtype,\Delta'} &= \sem{\modtype} \times \tupsem{\Delta'} \\
    \tabsem{\Gamma := \tabtype,\Gamma'} &= \tabsem{\tabtype} \times \tabsem{\Gamma'} \\
        \tabsem{\Gamma := \modtype,\Gamma'} &= \tabsem{\modtype} \times \tabsem{\Gamma'}
\end{align*}
\end{subfigure}

\smallskip
\hrule
\smallskip
\begin{subfigure}[t]{.4\textwidth}
\caption{\textbf{Semantics of Base Types}}
\label{subfig:semantics-base-type}
\begin{align*}
  \sem{\bool} &:= (\BB\cup\{\star\},\mathcal{P}(\BB\cup\{\star\}),\nu_\BB) \\
  \sem{\inte} &:= (\ZZ\cup\{\star\},\mathcal{P}(\ZZ\cup\{\star\}),\nu_\ZZ) \\
  \sem{\str} &:= (\strsem\cup\{\star\},\mathcal{P}(\strsem\cup\{\star\}),\nu_\strsem) \\
  \sem{\real} &:= (\RR\cup\{\star\},\mathcal{B}(\RR\cup\{\star\}),\nu_\RR) \\
  \sem{\posreal} &:= (\RR^+\cup\{\star\},\mathcal{B}(\RR^+\cup\{\star\}),\nu_{\RR^+}) 
\end{align*}
\end{subfigure}%
\begin{subfigure}[t]{.6\textwidth}
\caption{\textbf{Semantics of Scalar Expressions}}
\label{subfig:semantics-scalar}
\begin{flalign*}
\sem{\id.\col_i}(\gamma,\delta) &= \begin{array}[t]{@{}l}
  \pi_i(\delta(\id))
  \quad\mbox{where } T[?\id]\{\cols\}\in \Delta
  \end{array}\\
\sem{\op(e_1,\ldots,e_n)} (\gamma,\delta) &=
  \op_s\left(\sem{e_1} (\gamma,\delta),\ldots,
      \sem{e_n} (\gamma,\delta)\right)
\\
\begin{aligned}
&\llbracket
  \probof c^1
  \under m
  \rrbracket(\gamma,\delta)\span= \sem{m}(\gamma,\delta)\dotmeas\big(\sem{c^1}(\gamma,\delta)\big)\\
&\llbracket
  \probof c^0
  \under m \rrbracket(\gamma,\delta) = \llet (\pi,v)=\span \\ 
  &\sem{c^0}(\gamma,\delta[\id\mapsto \sem{m}(\gamma,\delta)])\iin \sem{m}(\gamma,\delta)\dotpdf(v)
\end{aligned}\span
\end{flalign*}
\end{subfigure}
\smallskip
\hrule
\smallskip

\begin{subfigure}{\textwidth}
\caption{\textbf{Semantics of Table Expressions}}
\label{subfig:semantics-table}
\begin{align*}
\sem{\id: T[?\id]\{\cols\}}(\gamma,\delta) &= \gamma(\id) \qquad  \sem{\rename t \as \id'}(\gamma,\delta) = \sem{t}(\gamma,\delta)\\
\sem{t_1 \join t_2} (\gamma,\delta) &=
(\lambda x,y.\maptwo (\lambda r_1,r_2. (r_1,r_2))\ x\ y)_*
(\sem{t_1}(\gamma,\delta) \otimes  \sem{t_2}(\gamma,\delta))\\
\sem{t:T[?\id]\{\cols\}\where e}(\gamma,\delta) &=
 \Big(\lambda x.\filter (\lambda r. \sem{e}(\gamma,\delta[\id\mapsto r]))\ x\Big)_*
 \sem{t}(\gamma,\delta) \\
 \sem{\generate m \limit e}(\gamma,\delta) &=
 \llet n=\sem{e}(\gamma,\delta)\ \iin
\Big(\lambda (x_1,\ldots,x_n). \textstyle\bigcup_{1\leq i\leq n}\{x_i\}\Big)_*
\bigotimes_{1\leq i\leq n}\sem{m}(\gamma,\delta)\dotmeas
\\
\begin{aligned}
&\sem{\select e_1 \as \col_1,\ldots,e_n \as \col_n \from t:T[?\id]\{\cols\}} (\gamma,\delta) = \\
&\Big(\lambda x. \map (\lambda r. \big(\sem{e_1}(\gamma,\delta[\id\mapsto r]),\ldots,\sem{e_n}(\gamma,\delta[\id\mapsto r])\big) )\ x \Big)_*
\sem{t}(\gamma,\delta)\\
&\sem{(t:T[?\id]\{\cols\}) \modjoin m}(\gamma,\delta) = \sem{t}(\gamma,\delta) \bind \Big(\lambda y.\fold (\lambda \mu,r. \mu \bind\\
&\Big(\lambda x. (\lambda \{r'\}. x\cup \{(r,r')\})_*
\sem{\generate m\limit 1}(\gamma,\delta[\id\mapsto r])\dotmeas\Big)\ \delta_{\{\}}\ y\Big)
\end{aligned}\span
\end{align*}
\end{subfigure}
\hrule
\smallskip

\begin{subfigure}{\textwidth}
\caption{\textbf{Semantics of \tableModel{} Expressions}}
\label{subfig:semantics-\tableModel}
\begin{align*}
     \sem{\id: M[?\id]\{\cols\}}(\gamma,\delta) &= (\gamma(\id)\dotmeas,\gamma(\id)\dotpdf) \qquad\qquad\sem{\rename m \as \id'}(\gamma,\delta) = \sem{m}(\gamma,\delta) \\
     \sem{m \given c^1:\eventtype\{\cols\}} (\gamma,\delta)
     &= \Cond(\sem{m}(\gamma,\delta),\sem{c^1}(\gamma,\delta[\id\mapsto \sem{m}(\gamma,\delta)])) \\
\sem{\id \given c^0:\densitytype\{\cols'\}} (\gamma,\delta) &= \llet (\pi,v) = \sem{c^0}((\gamma,\delta[\id\mapsto \sem{\id}(\gamma,\delta)]))\ \iin \Dis(\gamma(\id),\pi,v)
\end{align*}
\end{subfigure}
\hrule
\smallskip

\begin{subfigure}{\linewidth}
\centering
\begin{subfigure}[t]{.48\textwidth}
\caption{\textbf{Semantics of Event Expressions}}
\label{subfig:semantics-event}
\begin{align*}
    &\textstyle\sem{\bigwedge_{1\leq i\leq 2} c^1_i} (\gamma,\delta) = 
    \textstyle\bigcap_{1\leq i\leq 2}\sem{c^1_i}(\gamma,\delta) \\
    &\textstyle\sem{\bigvee_{1\leq i\leq 2} c^1_i} (\gamma,\delta) = 
    \textstyle\bigcup_{1\leq i\leq 2}\sem{c^1_i}(\gamma,\delta) \\
\end{align*}
\end{subfigure}
\hspace{-.4cm}
\begin{subfigure}[t]{.48\textwidth}
\caption{\textbf{Semantics of Event-0 Expressions}}
\label{subfig:semantics-event-0}
\begin{align*}
     \sem{\bigwedge_{1\leq i\leq 2} c^0_i} (\gamma,\delta) &=\left\{ \!\!\Centerstack[l]{\ \llet_{1\leq i\leq 2}(f_i,v_i)= \sem{c^0_i}(\gamma,\delta)\ \iin
     \#
     \ \big(\lambda x.(f_1(x),f_2(x)),(v_1,v_2)\big)
     }\right.
\end{align*}
\end{subfigure}
\vspace{-5mm}
\begin{align*}
\sem{\id.\col_i\ \op\ e : \eventtype\{\cols\} } (\gamma,\delta) =
\{(x_1,\ldots,x_n)\in \sem{\cols} \mid x_i\ \op_l\ \sem{e}(\gamma,\delta) \}
\quad  \sem{\id.\col_i\ =\ e} (\gamma,\delta) = (\pi_i, \sem{e}(\gamma,\delta))\end{align*}

\smallskip
\end{subfigure}
\end{framed}
\caption{Denotational semantics of \iql{}.}
\label{fig:semantics-iql}
\end{figure}

\paragraph{Scalar expressions (\cref{subfig:semantics-scalar}).} 
We then interpret scalar expressions $\ctx \vdash e:\sigma$ as
measurable functions $\tabsem{\Gamma}\times \tupsem{\Delta}\to
\sem{\sigma}$. 
We lift operations $\op$ to interpret $\Null$, and write $\op_s$ for the extended version of $\op$ which sends $\star$ to $\star$.

\paragraph{Event expressions (\cref{subfig:semantics-event}).} 
An event expression $\sem{c^1:\eventtype\{\cols\}}(\gamma,\delta)$ is interpreted as a measurable subset $S$ of $\sem{\cols}$ (disjoint union of
hyper-rectangles~\citep{saad2021sppl}).
Depending on the expression, this set $S$ is used in
different ways.
We interpret the probability clause$\probof~c^1
\under m$ as $\int_{\sem{\cols}} \mathds{1}_S d\mu$, where $\mu$ is the measure denoting the model $m$, i.e. $S$ is used in an indicator function
$\mathds{1}_S$.
When used in a$\given$ clause, we constrain the model to the event
$S$, which is then renormalized.
If the event has probability 0, we instead return a row of $\Null$.
A similar situation to$\where \Null$ arises for$\given$, e.g. in
$\given \id.\col ~\op~ \Null$.
Following the principle of least surprise, $\Null$ acts by convention
as a unit for conditioning, i.e. $\id \given \id.\col ~\op~ \Null$
behaves the same as $\id$.
To ensure this we interpret boolean expressions $\op$ differently in
the semantics of events, and write $\op_l$ for the extended version of $\op$ which sends $\star$ to $\true$.
The denotation of $\id.\col~ \op~ \Null$ will therefore be the entire
space, and conditioning a model on this event will not change its
denotation.

\paragraph{Event-0 expressions (\cref{subfig:semantics-event-0}).}
$\sem{c^0:\densitytype\{\cols\}}(\gamma,\delta)$ is interpreted as a pair of a projection function $\pi$ and a value $v$ in the codomain of the projection. 
$v$ is used to specify the point at which we want to condition or evaluate a density, and $\pi$ is used to project the model to the relevant subspace, which we detail in the paragraph on \tableModel{} expressions.

\paragraph{Table expressions (\cref{subfig:semantics-table}).} 
We interpret closed tables expressions $\vdash t:T[?\id]\{\cols\}$ as measures on their columns, i.e. elements of $\measm\Big(\tabsem{T[?\id]\{\cols\}}\Big)$.
We write $\mu \bind \kappa$ for the composition of a measure $\mu$ on $X$ with a kernel $X\to \measm Y$, defined by $\mu \bind \kappa(dy)=\int \kappa(x,dy) \mu(dx)$.
Given a measurable function $f:X\to Y$, we
denote the pushforward measure by $f_*\mu(A):=\mu(f^{-1}(A))$.

We use functional programming notation for the mathematical
functions $\filter$, $\map$, $\maptwo$, $\fold$.
Given a bag $S$ and a function $f:S\to \BB$, we define the bag
$\filter f\ S:=\{x\in S\mid f(x)\neq 0\}$.
Likewise, we define $\map f\ S:=\{f(x)\mid x\in S\}$ and $\maptwo f\
S\ T:=\{f(x,y)\mid x\in S,y\in T\}$.
A function $f: X\times Y\to Y$
is commutative \citep{dash2021monads} if
$f(x_1,f(x_2,y))=f(x_2,f(x_1,y))$ for all $x_1,x_2,y$.
Given a commutative function $f: X\times Y\to Y$, we further define
$\fold f:Y\times \bag(X)\to Y$ by $\fold f\ y_0 \{x_1,\ldots,x_n\}=
f(x_1,f(x_2,\ldots f(x_n,y_0)\ldots))$.

\paragraph{\TableModel{} expressions (\cref{subfig:semantics-\tableModel}).} The semantics of \tableModels{} is more involved, as conditioning statements$\given c^0$ require conditioning on events of probability 0.
We first review the minimal setting that helps us define conditioning on event-0 expressions.
Given measurable spaces $A,B$ with reference measure $\nu_A,\nu_B$,
a measure $\mu$ on $A\times B$ admits an \emph{$(A,B)$ disintegration}
if we can write $\mu=\nu_A\otimes \kappa$ for some measure kernel
$\kappa$ such that for all $a\in A$, $\kappa(a)$ has a density
$p(-\mid a)$ w.r.t. $\nu_B$.
A \emph{valid decomposition} $(A,B)$ for $\prod_{1\leq i\leq
n}\sem{\sigma_i}$ is given by $A=\prod_{j\in J}\sem{\sigma_j}$ for
some $J\subseteq \{1,\ldots,n\}$ and $B=\prod_{j\in
\{1,\ldots,n\}-J}\sem{\sigma_j}$.
A measure $\mu\in \measm(\prod_{1\leq i\leq n} \sem{\sigma_i})$ is
\emph{admissible} if it admits an $(A,B)$ disintegration for all valid
decompositions $(A,B)$ of $\prod_{1\leq i\leq n}\sem{\sigma_i}$.

We consider measures $\mu$ on spaces $X$ with chosen
disintegrations and (marginal) densities w.r.t. the reference measure.
More precisely, we interpret a \tableModel{} $\id$ from the global context $\Gamma$ as a quadruple $\tabsem{\id}:=(\mu,p,\{\kappa_A\}_A,\{p\}_A)$.
Here, $\mu$ is a measure denoting the unconditioned model, and $p$ a density of $\mu$ w.r.t. the reference measure.
For each valid decomposition $(A,B)$ of the columns of $\id$, the kernel $\kappa_A$ is an $(A,B)$-disintegration of $\mu$.
For all $a\in A$, $p_A(-\mid a)$ is a density for $\kappa_A(a)$ w.r.t. the reference measure $\nu_B$.
If $v$ is a partial assignment of the variables in $B$, we also write $p_A(v \mid a)$ for the marginal density of $\kappa_A(a)$ at $v$ obtained from $p_A(-\mid a)$ by integrating out the missing variables in $v$.
We denote by $\admissible(X)$ the set of such quadruples $(\mu,p,\{\kappa_A\}_A,\{p\}_A)$, where $\mu$ is a measure on $X$.
Given $m\in \admissible(X)$, we write $m\dotmeas$ for its first component $\mu$, $m\dotpdf$ for the density $p$, $m.A$ for the kernel $\kappa_A$, and $m.A\dotpdf$ for the density $p_A$.
Using this notation, given an event-0 $c^0$ denoting a projection $\pi$ and value $v$, the expression $m\dotpdf(v)$ gives a marginal density of the model $m$ at $v$; i.e. $m\dotpdf(v)$ is a version of the density of $\pi_*m\dotmeas$ evaluated at $v$.
We assume that all the models in the context are admissible, which is enforced in the semantics of contexts.

The models used in queries are built from admissible models and will carry chosen densities, which is enforced in the semantics of \tableModel{} expressions. 
We write $\dens(X)$ for the set of pairs $(\mu,p)$ where $\mu$ is a measure on $X:= X_1\times \ldots \times X_n$ and $p$ is either a density of $\mu$ w.r.t. the reference measure, or of the form $\lambda (x_1,\ldots,x_n).q(x_{i_1},\ldots,x_{i_k})$ for some $i_1,\ldots,i_k$, and where $q$ is a marginal density of $\mu$ on $X_{i_1}\times \ldots \times X_{i_k}$ w.r.t. the reference measure. 
The second case is used to represent the density of a model conditioned on an event-0 expression.

Conditioning on events-0 requires access to a disintegration of the model at the point $v$, which is possible thanks to the restriction from the type-system. 
For $m\in \admissible(X)$, $\pi:X\to Y$ a projection function, and $v\in Y$, we define $\Dis(m,\pi,v):= (m.\pi(X)(v)\otimes \delta_v, m.\pi(X)\dotpdf(v))$.

For conditioning on events, given $m\in\dens(X)$ and a measurable $S\subseteq X$,
we define
{\setlength{\abovedisplayskip}{0pt}
\setlength{\belowdisplayskip}{0pt}
\[ \cond(m,S):=\left\{ \begin{aligned} 
& \left(\lambda S'. \frac{m\dotmeas (S\cap S')}{m\dotmeas(S)}
, \lambda x.\frac{\mathds{1}_S(x)m\dotpdf(x)}{m\dotmeas(S)}\right) & \text{if }m\dotmeas(S)>0 \\
& \left(\delta_{\{\star,\ldots,\star\}}, \mathds{1}_{\{(\star,\ldots,\star)\}}\right)  & \text{otherwise} 
    \end{aligned} \right. \]
}


\section{Abstract  Model Interface and Query Planner}
\label{sec:impl-table-model}

This section presents a query planner that automatically lowers \iql{}
queries to programs that operate on tables and
\tableModels{}.
This lowering depends on the Abstract \tableModel{}  Interface
(\ami{}) which we assume all loaded \tableModels{} must satisfy.
The \ami{} is a flexible interface that many \tableModel{} implementations
can easily satisfy.
This flexibility means that model implementations can strike
different expressiveness-speed-accuracy trade-offs, and give
different guarantees. 

\noindent \cref{sub:comparison-against-baselines-using-approximate-inference} compares an exact SPPL backend to an approximate Gen.clj backend on 5 queries.

In what follows, we define the \ami{} and show how to lower
\iql{} queries to programs that access \tableModels{} through the \ami{}
interface.
We showcase the flexibility of the \ami{} by proving formal guarantees for
two different implementations of the \ami{}.
We show in \cref{sub:exact-backend} that if the \ami{} is implemented in a PPL with exact inference, then \iql{} queries can be lowered to programs in a
semantics-preserving way.
We then show in \cref{sub:approximate-backend} that if the \ami{} is
implemented in a PPL with approximate inference, then \iql{} queries
can be lowered to programs that encode asymptotically sound estimators
for$\probof$expressions and asymptotically sound samplers for$\generate$expressions.

\subsection{Abstract Model Interface (\ami{})}
\label{sub:\ami{}}

A \tableModel{} represents a probability distribution on rows with a
fixed set of columns.
The \ami{} captures the intuition that a model should be able to
produce samples and compute probabilities and densities for all
conditioned versions of the distribution it represents.
For each \tableModel{} $M[?\id]\{\cols\}$, the \ami{} requires the
existence of the following three methods:
\begin{align*}
\simulate_\id &: (\densitytype\{\cols'\}, \eventtype\{\cols\})\to T[?\id]\{\cols\} \\
\logpdf_\id &: (\densitytype\{\cols'\}, \eventtype\{\cols\}, \densitytype\{\cols''\}) \to \real \\
\prob_\id &: (\densitytype\{\cols'\}, \eventtype\{\cols\}, \eventtype\{\cols\}) \to  \ranged(0,1).
\end{align*}
where $\cols',\cols''\subseteq \cols$. These methods should behave as follows:
\begin{itemize}
\item $\simulate_\id(c^0,c^1)$
returns a sample from a model with identifier $\id$, conditioned on the event-0
$c^0$ and event $c^1$.
\item $\logpdf_\id(c^0,c^1,c^0_2)$ returns the (marginal if $\cols''\subsetneq \cols$) log-density of the model $\id$ conditioned on the event-0 $c^0$ and event $c^1$, at the point $c^0_2$. 
\item $\prob_\id(c^0,c^1,c^1_2)$ returns the probability of the event $c^1_2$ under the model $\id$, conditioned on the event-0 $c^0$ and event $c^1$.
\end{itemize}

A non-conditioned model is recovered by letting the subset $\cols'$ to be empty.
The precise usage of these methods is given in the next section.
The \ami{} methods can have different formal semantics, capturing
different aspects of the backend probabilistic model it abstracts.
These semantics reflect different implementation strategies
implementing conditional sampling and probability evaluation.
\Cref{sec:implementations-atmi} shows how different model classes can
implement the \ami{}.
In particular, we show that SPPL~\cite{saad2021sppl} and truncated
multivariate Gaussians satisfy the exact \ami{}, and that any PPL
implementing ancestral sampling will satisfy the approximate \ami{}.
We next describe how the \iql{} query planner lowers queries to
programs that rely on the \ami{}, before giving details about the
semantics and correctness guarantees.

\subsection{Lowering \iql{} to Queries on the \ami{}}

The lowering procedure from \iql{} to a lowered language is given in
two steps:
\begin{enumerate*}[label=(\roman*)]
\item a normalization transform for \iql{} queries; and
\item a program transform to the lowered language.
\end{enumerate*}

\paragraph{Normalization of \iql{} Queries.}
The normalization (see \cref{sub:normalization}) simply simplifies$\rename$ statements and aggregates events in a single conditioning statement. It leads to the following normal forms, where$\given$clauses are optional:
\begin{itemize}
    \item Probability queries:$\probof c^i_1 \under (\id \given c^0 \given c^1)$.
    \item Generate queries:$\generate (\id \given c^0 \given c^1) \limit e$ and \\ $t \modjoin (\id \given c^0 \given c^1)$.
\end{itemize}

\paragraph{Lowering Language (\cref{fig:lowered-syntax}).}
It is a first-order simply-typed lambda calculus with second-order operations acting on bags, and primitives for the \ami{}.
It also contains a version of events and events-0 which can be used by \ami{} primitives.
Operations like $\map$, $\filter$ and $\hostexp$ have their
usual meaning, and their typing along with those for constants,
tuples, projections, and arithmetic operations are standard and
recalled in \cref{sec:lowering-appendix}
(\cref{fig:lowered-syntax-full}). 
$\lowerjoin$ takes two bags of tuples and returns their Cartesian product.
$\replicate$ evaluates its bag argument $n$ times and returns the
union of all the resulting bags. 
$\mapreduce$ takes a bag of tuples and a function $f$ from tuples to bags, and returns the union of all the bags obtained by applying $f$ to each tuple in the input bag.

\begin{figure}
    \captionsetup{skip=4pt}
    \setlength{\abovedisplayskip}{0pt}
\setlength{\belowdisplayskip}{0pt}
\setlength{\abovedisplayshortskip}{0pt}
\setlength{\belowdisplayshortskip}{0pt}
\setlength{\FrameSep}{0pt}
    \begin{framed}
            \small      
\begin{align*}
    \text{base type } \sigma &::= \sigma_c \gor  \sigma_d  
    &\hspace{-2mm}\text{ground type } \sigma_g&::= \sigma \gor (\sigma_1, \ldots, \sigma_n) 
    &\text{event type } \evtype &::= \eventtype[\sigma_g] \gor \densitytype[\sigma_g]
    \\
    \text{type } \tau &::= \bag[\sigma_g]
    &\text{operator } \op\ &::= + \gor  - \gor  \times \gor \div \gor \land \gor  \lor \gor\ = 
    &\text{\tableModel{} } \modtype &::= M[\sigma_g]
\end{align*}

\vspace{-3mm}
\begin{minipage}{.68\linewidth}
\begin{align*}
    \text{primitives } f\ &::= \mapreduce \gor \map \gor \filter \gor \replicate \gor \lowerjoin \gor \hostexp\\
    &\quad \gor \singleton \gor \simulate_\id\gor \prob_\id\gor \logpdf_\id\\
    \text{term } t &::= c \gor \id \gor f(t_1, \ldots, t_n) \gor x \gor (t_1, \ldots, t_n) \gor \pi_i\ t \gor t_1\ op\ t_2 \\
    \end{align*}
\end{minipage}
\begin{minipage}{.3\linewidth}
    $$
    \frac{\Gamma\vdash t_1 : \densitytype[\sigma_g^1]\quad \Gamma\vdash t_2 : C^0[\sigma_g^2]}
    {\Gamma \vdash t_1 \land t_2: \densitytype[\sigma_g^1, \sigma_g^2]}
    $$
\end{minipage}
   \vspace{-6mm}
    
    $$
    \frac{\Gamma\vdash t_1 : \eventtype[\sigma_g]\quad \Gamma\vdash t_2 : \eventtype[\sigma_g]\quad \op\in\{\land,\lor \}}
    {\Gamma \vdash t_1\ \op\ t_2: \eventtype[\sigma_g]} 
    \quad
    \frac{\Gamma\vdash t:\sigma_i\quad \op\in\{=,<,>\}\quad (\sigma_i,\op)\neq (\sigma_c,=)}
         {\Gamma, \id: M[(\sigma_1, \ldots,\sigma_n)] \vdash (\id, i)\ \op\ t : \eventtype[(\sigma_1,\ldots,\sigma_n)]}
    $$
    $$
    \frac{\Gamma, \id: M[\sigma_g]\vdash 
    c^i: C^i[\sigma_g]\ \ \Gamma, \id: M[\sigma_g]\vdash c^1_1: \eventtype[\sigma_g]}
         {\Gamma, \id: M[\sigma_g]\vdash \prob_{\id}(c^0,c^1,c^1_1,): \ranged(0,1)}
    \ \
    \frac{\Gamma, \id: M[\sigma_g]\vdash 
    c^i: C^i[(\sigma_1, \ldots,\sigma_n)]}
         {\Gamma, \id: M[\sigma_g]\vdash \simulate_{\id}(c^0,c^1): \bag[\sigma_g]}
    $$
    $$
    \frac{\Gamma \vdash t:\sigma_i}
    {\Gamma,\id: M[(\sigma_1, \ldots,\sigma_n)] \vdash (\id, i) = t : \densitytype[\sigma_i]}
    \ \
    \frac{\Gamma, \id: M[\sigma_g]\vdash 
    c^i: C^i[\sigma_g]\ \ \Gamma, \id: M[\sigma_g]\vdash c^0_1: \densitytype[\sigma_g]}
         {\Gamma, \id: M[\sigma_g]\vdash \logpdf_{\id}(c^0, c^1,c^0_1): \real}
    $$
    \end{framed}
    \caption{A selected subset of the syntax and type system of the lowered language.}
    \label{fig:lowered-syntax}
\end{figure}

\paragraph{Lowering program transform (\cref{fig:lowering-transform}).}
After obtaining a normal form query, the planner applies a program
transformation $\tran{\delta}{\cdot}$ from normalized \iql{} queries
to the lowered language, defined by pattern matching
on the structure of the query.
It carries a local context $\delta$ of variables (a finite map from
identifiers to variable names) which are bound in the surrounding program.
Similarly to the local context $\Delta$ in \iql{}, $\delta$ will start
empty $[]$ at the root of the syntax tree. It is used to
rename variables in the lowered query.
The rationale is that a table identifier $\id$ in $\Delta$ will be
transformed to a variable $r$ representing a tuple being iterated over
by a $\map$ or $\fold$ primitive.
A \tableModel{} identifier $\id$, on the other hand, will be uniquely
accessible and identified from the global context $\Gamma$, thanks to
the normalization procedure which ensures that no \tableModel{} is
renamed in the normalized query.
A simple proof by induction shows that the transformation preserves
typing.

\begin{propositionE}
If $\Gamma,[] \vdash t: T[?\id]\{\cols\}$, then $\tranty{\Gamma} \vdash \tran{[]}{t}: \tranty{T[?\id]\{\cols\}}$.
\end{propositionE}

\begin{figure}[t]
    \captionsetup{skip=4pt}
\footnotesize
\setlength{\abovedisplayskip}{-5pt}
\setlength{\belowdisplayskip}{0pt}
\captionsetup[subfigure]{skip=0pt,belowskip=1pt,aboveskip=1pt}

\setlength{\FrameSep}{1pt}
\begin{framed}
\begin{subfigure}[t]{.5\linewidth}
    \caption{\textbf{Translating Types and Contexts}}
    \centering
    \begin{align*}
        &\tranty{ T[?\id]\{\col_1: \sigma_1,\ldots,\col_n: \sigma_n\}} = \bag[(\sigma_1,\ldots,\sigma_n)] \\
        &\tranty{ M[?\id]\{\col_1: \sigma_1,\ldots,\col_n: \sigma_n\}} = M[(\sigma_1,\ldots,\sigma_n)] \\
        &\tranty{ C^i\{\col_1: \sigma_1,\ldots,\col_n: \sigma_n\} } = C^i[(\sigma_1,\ldots,\sigma_n)] \\
        &\tranty{ \Gamma, T[?\id]\{\cols\}} = \tranty{ \Gamma}, \id: \tranty{T[?\id]\{\cols\}}  \\
         &\tranty{ \Gamma, M[?\id]\{\cols\}} = \tranty{ \Gamma}, \id: \tranty{M[?\id]\{\cols\}} \\
         &\tranty{ \sigma} = \sigma \quad
         \tranty{ []} = []
    \end{align*}
\end{subfigure}%
\begin{subfigure}[t]{.48\linewidth}
    \caption{\textbf{Translating Event and Event-0 Expressions}}
    \begin{align*}
    \tran{\delta}{\id.\col_i = e} &= (\id,i) = \tran{\delta}{ e} \\
    \tran{\delta}{ \id.\col_i\ >\ e} &= (\id,i) > \tran{\delta}{ e} \\
    \tran{\delta}{ \id.\col_i\ <\ e} &= (\id,i) < \tran{\delta}{ e} \\
    \tran{\delta}{ c_1 \land c_2} &= \tran{\delta}{ c_1} \land \tran{\delta}{ c_2} \\
    \tran{\delta}{ c_1 \lor c_2} &= \tran{\delta}{ c_1} \lor \tran{\delta}{ c_2} \\
    \end{align*}
\end{subfigure}

\vspace{-4mm}
\begin{subfigure}[t]{\linewidth}
\caption{\textbf{Translating \TableModel{} Queries}}
\centering
\begin{align*}
&\tran{\delta}{ \probof c_2^0 \under \id \given c^0 \given c^1} =
    \hostexp\big(\logpdf_{\id}(\tran{\delta}{ c^0},\tran{\delta}{ c^1},\tran{\delta}{ c_2^0})\big)\\
&\tran{\delta}{ \probof c_2^1 \under \id \given c^0 \given c^1} =
    \prob_\id(\tran{\delta}{ c^0},\tran{\delta}{ c^1},\tran{\delta}{ c_2^1},)\\
&\tran{\delta}{ \generate \id \given c^0 \given c^1 \limit e} = \replicate(\tran{\delta}{ e}, \simulate_\id(\tran{\delta}{ c^0},\tran{\delta}{ c^1})) \\
&\tran{\delta}{t:T[\id']\{\cols\} \modjoin \id \given c^0 \given c^1} = \\
&\qquad \mapreduce \left(\lambda r.\lowerjoin\left(\singleton(r), \simulate_\id(\tran{\delta[\id' \mapsto r]}{c^0}, \tran{\delta[\id'\mapsto r]}{c^1})\right),
    \tran{\delta}{ t}\right)
\end{align*}
\end{subfigure}

\begin{subfigure}[t]{\linewidth}
\caption{\textbf{Translating Scalar Expressions}}
\begin{minipage}{\linewidth}
\begin{align*} \tran{\delta}{ c} = c\quad
\tran{\delta[\id\mapsto r]}{\id.\col_i} = \pi_i(r) \quad
\tran{\delta}{ \op(e_1,\ldots,e_n)} = \op(\tran{\delta}{ e_1}, \ldots
\tran{\delta}{ e_n})
\end{align*}
\end{minipage}
\end{subfigure}%

\begin{subfigure}[t]{\linewidth}
\caption{\textbf{Translating Table Expressions}}
\begin{minipage}{.47\linewidth}
\begin{align*}
\tran{\delta}{ \rename t \as \id} &= \tran{\delta}{ t};\; \tran{\delta}{ \id} = \id \\
\tran{\delta}{ t_1 \join t_2} &= \lowerjoin(\tran{\delta}{ t_1}, \tran{\delta}{ t_1}) \\
\tran{\delta}{ t:T[\id]\{\cols\} \where e} &= \filter(\lambda r.\tran{\delta[\id\mapsto r]}{ e},
\tran{\delta}{ t})
\end{align*}
\end{minipage}
\begin{minipage}{.5\linewidth}
\begin{align*}
&\tran{\delta}{ \begin{aligned}
    &\select \overline{e} \as \overline{\col}\\
    &\from t: T[\id]\{\cols\}
    \end{aligned}
    }
    =\\
&\map(\lambda r. \tran{\delta[\id\mapsto r]}{\overline{e}}, \tran{\delta}{ t})
\end{align*}
\end{minipage}
\end{subfigure}
\end{framed}

\caption{The lowering transformation $\tranty{\cdot}$.}
\label{fig:lowering-transform}
\end{figure}

\subsection{Lowering Guarantees for Exact Backend}
\label{sub:exact-backend}

A large class of models supports exact inference, e.g. those expressible in SPPL~\citep{saad2021sppl} and truncated multivariate Gaussians.
These models satisfy the exact \ami{} and are able to return exact
samples from $\simulate$, and compute exact marginal $\logpdf$ and $\prob$ queries, even for conditioned models.
We make this precise by giving a measure semantics on the lowered
language (\cref{fig:lower-exact-semantics}) and show that the
program transform $\tranty{\cdot}$ preserves the semantics of the
lowered query (\cref{thm:exact-\ami{}-guarantee}).
In particular, all the scalar computations in
the query are deterministic and that the generated synthetic data comes from exact conditional distributions.

The denotational semantics (\cref{sec:lowering-appendix}, \cref{fig:lower-exact-semantics}) of the lowered language is mostly standard
and resembles the measure-theoretic semantics of \iql{} given in
\cref{fig:semantics-iql}.
Terms $\Gamma \vdash e: \sigma_g$ are interpreted as deterministic
measurable functions $\semex{\Gamma} \to \semex{\sigma_g}$.
Terms $\Gamma \vdash e: \bag[\sigma_g]$ are interpreted as
probability kernels $\semex{\Gamma} \to \measm\bag(\semex{\sigma_g})$,
where substitution for these programs is interpreted using the Kleisli
composition for the point process monad \citep{dash2021monads}.
By induction on the structure of \iql{} programs $t$ in context
$\Gamma;\Delta$, we can show (proof in \cref{sub:proof-exact}):

\begin{theoremE}[Exact \ami{} Guarantee][end, text link=]
\label{thm:exact-\ami{}-guarantee}
Let $\Gamma, [] \vdash t: T[?\id]\{\cols\}$.
Then, for every evaluation of the context $\gamma$,
{\setlength{\abovedisplayskip}{0pt}
\setlength{\belowdisplayskip}{0pt}
\begin{equation*}
\sem{t}(\gamma,[]) = \semex{\tran{[]}{t}}(\gamma).
\end{equation*}}
\end{theoremE}

\subsection{Approximate Backend Guarantee}
\label{sub:approximate-backend}

By relying on approximate probabilistic inference, general-purpose
PPLs can express large classes of models in which exact inference is
intractable.
In addition, programmable inference
\citep{mansinghka2018probabilistic} ensures PPLs can support a diverse class of probabilistic models without sacrificing inference quality. 
We give a new denotational semantics for the lowered language that is appropriate for reasoning in scenarios where the \tableModels{} are
implemented in PPLs with approximate Monte Carlo inference.

Monte Carlo algorithms are typically parameterized by a
positive integer $n$ specifying a compute budget, such as the
number of particles in a sequential Monte Carlo (SMC) algorithm
\citep{chopin2020introduction} or the number of samples in a Markov
Chain Monte Carlo (MCMC) algorithm \citep{robert1999monte}. 
The algorithm specifies a sequence of distributions or estimators that
converge in some sense to a quantity of interest as $n \to \infty$. 
In the case of approximate sampling algorithms, most typically the
distribution of the generated samples converges weakly to the target
distribution, and in the case of parameter estimation the algorithm
produces a strongly consistent estimator of the target parameter
\citep{chopin2020introduction, robert1999monte}.

\paragraph{Random variable semantics.} Our denotational semantics for
approximate \ami{} implementations is motivated by the above
discussion.
We assume the existence of an ambient probability space $(\Omega,
\mathcal{F}, \mathbb{P})$ and associate with each term a sequence of
random variables approximating the term in the exact semantics.
As an example, the approximate semantics of $\semap{\map (x.t_1)\
t_2}$ in the context $\gamma$ and at the ``random seed''
$\omega\in\Omega$ is given at the $n$-th approximation by
{\setlength{\abovedisplayskip}{1pt}
\setlength{\belowdisplayskip}{1pt}
\begin{equation*}
\semap{t_2}(\gamma,\omega)_n
    \bind \left(\lambda S. \return \left\{\lambda
    x'.\semap{t_1}(\gamma[x\mapsto x'], \omega)_n\ y\ \middle|\ y \in
    S\right\}\right).
\end{equation*}}
This means that we first obtain the $n$-th approximation of the input
$t_2$, which is a measure on tables, which we then evaluate to obtain
a concrete table, $S$.
We then apply the function to each row obtained by the $n$-th
approximation of $t_1$. The full semantics is given in
\cref{sub:appendix-approx-guarantee},
\cref{fig:lower-approximate-semantics}.
We assume the following hold:
\begin{itemize}
\item For each \tableModel{} identifier $\id:
M[(\sigma_1,\ldots,\sigma_k)]$ in environment $\gamma$, event $c^1:
C^1[(\sigma_1,\ldots,\sigma_k)]$, and event-0 $c^0 :
C^0[(\sigma_1,\ldots,\sigma_k)]$, there exists a sequence of probability
measures \\$\{\mu^n_{\id; \semap{c^0}(\gamma)_n, \semap{c^1}(\gamma)_n}\}$ on
$\bag\ \prod_{i=1}^k\sem{\sigma_i}$;

\item for $\id$, $\gamma$, $c^1$ and $c^0$ as above, and $c^1_2:
C^1[(\sigma_1,\ldots,\sigma_k)]$, there exists a sequence of real random
variables $\{P^n_{\id;\semap{c^0}(\gamma)_n,
\semap{c^1}(\gamma)_n,\semap{c_2^1}(\gamma)_n}\}$ which takes values in $[0, 1]$ $\mathbb{P}$-almost
surely;

\item for $\id$, $\gamma$, $c^1$ and $c^0$ as above, and $c^0_2:
C^0[(\sigma_1,\ldots,\sigma_k)]$, there exists a sequence of real random
variables $\{L^n_{\id;\semap{c^0}(\gamma)_n,\semap{c^1}(\gamma)_n,
\semap{c^0_2}(\gamma)_n}\}$.
\end{itemize}
These random sequences represent the sequences of approximations
produced by the implementation of the \ami{}.
In general, for a given term $t$ the convergence of sequences
associated with its sub-terms do not imply that the sequence
associated with $t$ converges.
For instance, consider evaluating the following query in an appropriate context $(\gamma,\delta)$:
{\setlength{\abovedisplayskip}{1pt}
\setlength{\belowdisplayskip}{1pt}
\begin{equation*}
\select \id.\col \from \id \where \id.\col \leq (\probof \id'.\col' = 7).
\end{equation*}}
If the value of the term$\probof \id'.\col' = 7$ is approximated, even if we
can make this approximation arbitrarily accurate, the output of the query need
not converge.
For example, if the table $\id$ contains a row in which the value of
$\col$ is exactly $\sem{\probof \id'.\col' = 7}(\gamma,
\delta)$ but the approximation converges to the true
value from below, this row will not be included in the query result no
matter the accuracy of the approximation.
Intuitively, this arises from the fact that the indicator functions of
half intervals are not continuous.

In order for the lowered queries to denote asymptotically sound
estimators for the original queries, we require that the
implementation of the \ami{} methods are asymptotically sound, and write $\lim_n \gamma_n$ to denote an evaluation of the context $\gamma$ in which each random variable is replaced by its limit as $n \to \infty$.
In \cref{sub:appendix-approx-guarantee}, we formalize the notions of
\emph{safe} queries and asymptotically \emph{sound} \ami{}
implementations and details of the proofs.
We then give the following guarantee.

\begin{theorem}[Consistent \ami{} Guarantee]
    \label{thm:consistent-ami-guarantee}
Let $\Gamma, [] \vdash t: T[?\id]\{\cols\}$ be a safe query and
suppose the \ami{} methods have asymptotically sound implementations.
Then, for every evaluation of the context $\gamma$, $\mathbb{P}$-almost surely
\begin{equation*}
\lim_n(\semap{\tran{[]}{t}})(\gamma) = \sem{t}(\lim_n \gamma,[]).
\end{equation*}
\end{theorem}

\section{Evaluation}
\label{sec:perf-eval}

The performance of an open-source Clojure implementation of \iql{} is
evaluated against other systems that have similar capabilities.
We test runtime, the effect of optimizations, and runtime
overhead of our system over alternative implementations of the same task.
Experiments were run on an Amazon EC2 C6a instance with Ubuntu 22.04,
4 vCPUs and 8.0 GiB RAM.

The probabilistic models used in the evaluation are obtained using
probabilistic program synthesis~\citep[Chapter 3]{saad2022}.
Each model is an ensemble of ``MultiMixture'' probabilistic
programs~\citep[Section 6]{saad2019bayesian}, which are posterior
samples from the CrossCat model class~\citep{mansinghka2016crosscat},
generated using ClojureCat~\citep{CharchutMEngThesis}.
An ensemble of 10 probabilistic programs is used in \cref{sec:perf-eval-runtime}
and 12 programs in \cref{sec:perf-eval-studies}.
%

\begin{table}[t]
\footnotesize
\begin{minipage}[t]{.525\linewidth}
\centering
\captionsetup{skip=4pt}
\caption{Runtime (sec) comparison of \iql{} and BayesDB~\citep{mansinghka2015bayesdb}
on 10 benchmark queries (\cref{appx:list-of-queries}) for evaluating
probability densities of measure-zero events.}
\label{tab:runtime}
\begin{tabular*}{\linewidth}{|l@{\extracolsep{\fill}}ccc|}
\hline
~ & \iql{} & BayesDB & Speedup \\
~ & (ClojureCat Backend) & (CGPM Backend) & ~ \\\hline\hline
Q1  & 0.24 $\pm$ 0.03 & 0.59 $\pm$ 0.16 & { 2.5x} \\
Q2  & 0.29 $\pm$ 0.03 & 1.15 $\pm$ 0.2  & { 4.0x} \\
Q3  & 0.43 $\pm$ 0.06 & 1.72 $\pm$ 0.28 & { 4.0x} \\
Q4  & 0.48 $\pm$ 0.06 & 2.25 $\pm$ 0.27 & { 4.7x} \\
Q5  & 0.57 $\pm$ 0.07 & 2.68 $\pm$ 0.36 & { 4.7x} \\
Q6  & 0.33 $\pm$ 0.06 & 0.55 $\pm$ 0.23 & { 1.7x} \\
Q7  & 0.49 $\pm$ 0.05 & 1.53 $\pm$ 0.26 & { 3.1x} \\
Q8  & 0.46 $\pm$ 0.03 & 1.81 $\pm$ 0.21 & { 3.9x} \\
Q9  & 0.37 $\pm$ 0.03 & 2.51 $\pm$ 0.32 & { 6.8x} \\
Q10 & 0.45 $\pm$ 0.04 & 2.87 $\pm$ 0.39 & { 6.4x} \\ \hline
{Mean} & 0.41 $\pm$ 0.11 & 1.77 $\pm$ 0.83 & { 4.3x}
\\\hline
\end{tabular*}
\end{minipage}\hfill
\begin{minipage}[t]{.45\linewidth}
\captionsetup{belowskip=0pt,aboveskip=6pt}
\vspace{-0.3cm}
\label{fig:benchmark-variance-rejection-sampling}
\includegraphics[width=\textwidth]{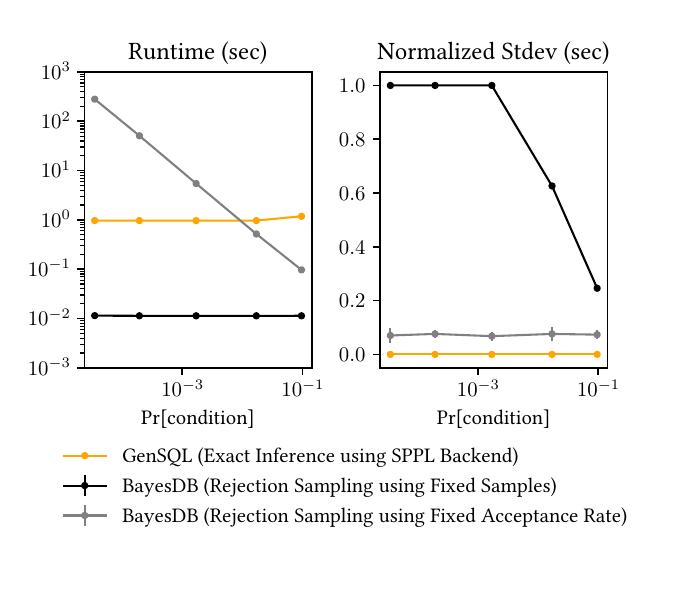}
\vspace{-1cm}
\captionof{figure}{
Runtime/stdev comparison of \iql{} and BayesDB~\citep{mansinghka2015bayesdb}
on 5 benchmark queries for evaluating probabilities of positive measure events.
}
\end{minipage}
\end{table}

\subsection{Performance and Usability}
\label{sec:perf-eval-runtime}

\paragraph{Runtime comparison}
\Cref{tab:runtime} compares the runtime on 10
benchmark queries (\cref{appx:list-of-queries}) adapted from
\citet[Tables 4.2 and 4.3]{CharchutMEngThesis} using \iql{} (with the
ClojureCat backend) and BayesDB (with the CGPM
backend~\citep{saad2016bayesdb}) for evaluating exact probability
densities.
\Cref{fig:benchmark-variance-rejection-sampling} compares the
runtime and standard deviation for computing the probabilities of positive measure events.
\iql{} (with the SPPL backend~\citep{saad2021sppl}) delivers exact
solutions, whereas BayesDB delivers approximate solutions using
rejection sampling.
Two rejection strategies in BayesDB are shown in
\cref{fig:benchmark-variance-rejection-sampling}: a fixed number of
samples (faster but higher variance) or a fixed acceptance rate
(slower but lower variance), which both are inferior to exact
solutions from \iql{}.

The performance gains in \iql{} are due to three main reasons: the
ClojureCat backend is faster than the CGPM backend in BayesDB, \iql{} has
optimizations (discussed below) that exploit repetitive computations,
and \iql{} itself is implemented in Clojure, a performant
language.
%

\paragraph{Optimizations and system overhead}
\iql{} leverages two classes of optimizations: caching (of the
likelihood queries and conditioned models) and exploiting independence
relations between columns.
The latter allows us to simplify a query such as$\probof \id.x>42
\under \id \given$ $\id.y = 17$ to the semantically equivalent query$\probof \id.x > 42 \under \id$ if the columns $x$ and $y$ are
independent.
\cref{sec:correctness-optimizations} gives a detailed account of the
optimizations.

In \cref{fig:effect-optim}, the unoptimized \iql{} queries have a
$1.1$-$1.6\mathrm{x}$ overhead compared to the pure ClojureCat baseline.
The optimizations reduce the overhead and can sometimes drastically
improve performance, while caching significantly reduces the
variance in the runtime of the queries.
In \cref{fig:effect-optim:b}, the effect of the independence
optimization varies between replicates, as these are different
CrossCat model samples, which explains the higher variance in query
runtime.

\paragraph{Code comparison.}
\Cref{fig:code-proba} compares the code required in \iql{}, pure
Python using SPPL~\citep{saad2021sppl}, and pure Clojure using
ClojureCat~\citep{CharchutMEngThesis}, for a conditional probability query.
\Cref{fig:code-proba:a} shows how \iql{} gains clarity by specializing
in data that comes from database tables.
In contrast, both SPPL and ClojureCat require users to hand-write the
looping/mapping over the data, which is error prone.
For instance, the code in \cref{fig:code-proba:c} will crash if the
table has missing values.
%
In \cref{fig:code-proba:b}, ClojureCat requires conditions to be maps.
Users can decide if they encode columns with strings, symbols, or
keywords.
If this choice does not align with the key type returned by the CSV
reader, the query will run but conditioning will result in a null-op.

In \cref{sub:code-comparison-with-scikit-learn}, we compare a single line query on a conditioned model in \iql{} to the equivalent code in Scikit-learn~\citep{pedregosa2011scikit} on the Iris data from the UCI ML repository.
The model querying alone in Scikit-learn is more than 50 lines long and clearly error prone, and we find that \iql{} offers a significant advantage in simplicity over such baselines.

\paragraph{Code comparison with BayesDB}
\Cref{fig:compare-code-bayesdb} shows \iql{} and its closest relative,
BayesDB~\citep{mansinghka2015bayesdb}, on a$\modjoin$query on
synthetic data.
The \iql{} code is more concise and simpler than BayesDB's code,
which is possible due to the language abstractions for manipulating
models.
In BayesDB, the user must exit to SQL and hand-code column
manipulations to fit the expected fixed pattern to query a model.
\Cref{sec:related-work} provides a detailed comparison of \iql{} and BayesDB.

\begin{figure}
    \captionsetup{skip=3pt,belowskip=0pt}
    \captionsetup[subfigure]{skip=0pt,labelfont=footnotesize,textfont=footnotesize}
    \begin{subfigure}[t]{0.31\textwidth}
    \centering
    \includegraphics[width=\textwidth]{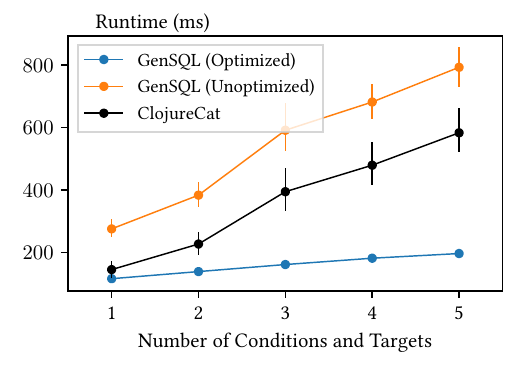}
    \caption{Varying number of conditions and targets in the
    \probof queries shown in \cref{tab:runtime}.}
    \end{subfigure}
    \hfill
    \begin{subfigure}[t]{0.31\textwidth}
    \centering
    \includegraphics[width=\textwidth]{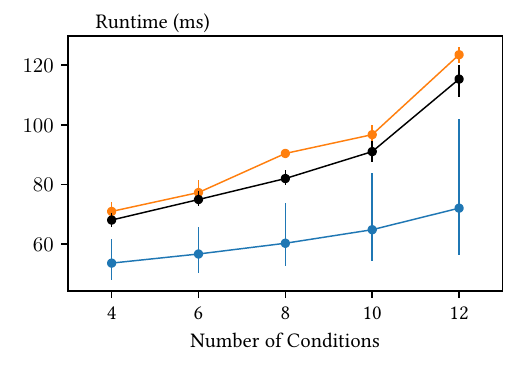}
    \caption{Varying number of conditions in $\given$ clause for
    \generate queries (caching does not apply).}
    \label{fig:effect-optim:b}
    \end{subfigure}
    \hfill
    \begin{subfigure}[t]{0.31\textwidth}
    \centering
    \includegraphics[width=\textwidth]{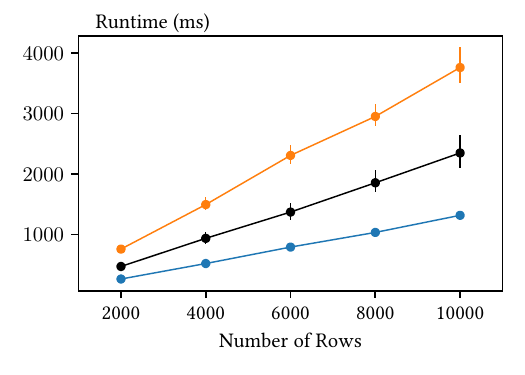}
    \caption{Varying number of rows in a data table
    used in the $\from$ clause of \select with a $\probof$ query.}
    \end{subfigure}
    
    \caption{Runtime comparison between \iql{} (ClojureCat backend) and
    raw ClojureCat \citep{CharchutMEngThesis}.}
    \label{fig:effect-optim}
    \smallskip

    \captionsetup[subfigure]{skip=0pt,belowskip=0pt,aboveskip=0pt}
    \begin{subfigure}[b]{0.42\textwidth}
     \begin{subfigure}[b]{\textwidth}
         \centering
         \input{figures/iql-code/prob-for-comparison}
         \caption{\iql{}}
         \label{fig:code-proba:a}
     \end{subfigure}
     \begin{subfigure}[b]{\textwidth}
         \centering
        \input{figures/other-code/clojure-prob}
         \caption{ClojureCat (Clojure)}
          \label{fig:code-proba:b}
     \end{subfigure}
    \end{subfigure}
    \qquad
    \begin{subfigure}[b]{0.5\textwidth}
         \centering
         \input{figures/other-code/sppl-prob}
         \caption{SPPL (Python)}
         \label{fig:code-proba:c}
    \end{subfigure}
    \caption{Comparison of \iql{}, ClojureCat~\citep{CharchutMEngThesis}, and SPPL~\citep{saad2021sppl} code for a conditional probability query.}
    \label{fig:code-proba}
   \smallskip

    \captionsetup[subfigure]{skip=0pt,belowskip=0pt,aboveskip=0pt}
    \captionsetup{skip=3pt}
    \begin{subfigure}[b]{0.5\textwidth}
    \centering
    \begin{subfigure}[b]{0.3\textwidth}
    \centering
    \scriptsize
    \begin{tabular}{|ccc|}
    \hline
    \textbf{a} & \textbf{b} & \textbf{c} \\\hline\hline
    $a_0$      & $b_0$      & $c_0$ \\
    $a_1$      & $b_1$      & $c_1$ \\
    $a_0$      & $b_0$      & $c_0$ \\
    $a_0$      & $b_1$      & $c_0$ \\
    $\dots$    & $\dots$    & $\dots$\\\hline
    \end{tabular}
    \caption{Table \texttt{foo}}
    \end{subfigure}
    \quad
    \begin{subfigure}[b]{0.5\textwidth}
    \centering
    \scriptsize
    \begin{tabular}{|ccccc|}
    \hline
    \textbf{a}&\textbf{b}&\textbf{x}&\textbf{y}&\textbf{z}\\\hline\hline
    $a_0$ & $b_1$ & 4.2  & 4.1  & 0.6 \\
    $a_1$ & $b_1$ & -4.4 & -5.4 & 0.2\\
    $a_1$ & $b_0$ & -3.7 & -6.2 & 0.5\\
    $a_1$ & $b_1$ & -6.2 & -4.2 & 0.1\\
    \dots&\dots&\dots&\dots&\dots\\\hline
    \end{tabular}
    \caption{Table \texttt{bar} to build model}
    \end{subfigure}
    
    \captionsetup[subfigure]{skip=0pt}
    \begin{subfigure}[b]{.8\textwidth}
    \centering
    \input{figures/iql-code/gen-join}
    \caption{\iql{}}
    \end{subfigure}
    \end{subfigure}
    \quad
    \begin{subfigure}[b]{0.4\textwidth}
    \input{figures/bql-code/gen-join}
    \caption{BayesDB}
    \end{subfigure}
    \caption{Comparison of \iql{} and BayesDB~\citep{mansinghka2015bayesdb} code.
    The latter does not support \modjoin.}
    \label{fig:compare-code-bayesdb}
    \end{figure}

\subsection{Case Studies on Real World Data}
\label{sec:perf-eval-studies}

We present two case studies to demonstrate the application of \iql{}
to real-world problems: one in medicine (clinical trial
data) and one in synthetic biology (wetlab data).
The datasets can be costly to obtain and researchers are interested in
understanding and analyzing their data.

In the first case, we show how anomaly detection in \iql{}
can be used to check for probable mislabelling of the data.
The anomalous rows can also be investigated further to understand the
reasons for the anomaly.
In the second case, we show how \iql{} can be used to generate
accurate synthetic data, capturing the complex relationships between
different host genes and experimental conditions.
Capturing these relationships with the model helps predict whether a
certain experimental condition or modification of the genome has
cascading downstream effects through the interrelations between the
genes. Such effects can render the cell toxic and kill the bacterium,
leading to a failed experiment.
The virtual wet lab allows researchers to check for such effects
before running costly experiments in the real world.

\begin{figure}
\captionsetup{skip=4pt}
\captionsetup[subfigure]{skip=4pt}

\begin{subfigure}[b]{0.49\textwidth}
     \centering
     \input{figures/iql-code/anom-detec-clinical}
     \caption{Find anomalous BMI values, defined by $ P(\textsc{BMI} \mid *) < 0.01$ for all BMI values within the 5th (20.3) and 95th (38.4) percentiles in the US.}
     \label{fig:case-anom:a}
 \end{subfigure}\hfill
 \begin{subfigure}[b]{0.49\textwidth}
     \centering
     \input{tables/anom.tex}
     \caption{Result from the query in \subref{fig:case-anom:a}. Eight anomalous participants are
     returned: all are clinically obese but report above-average health status and exercise.}
     \label{fig:case-anom:b}
\end{subfigure}

\captionsetup[subfigure]{skip=0pt,justification=centering}
 \begin{subfigure}[b]{0.35\textwidth}
     \centering
     \includegraphics[width=\textwidth]{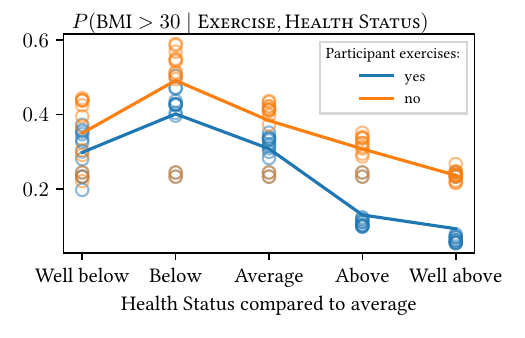}
     \caption{Conditonal probabilites encoded by the underlying model.}
     \label{fig:case-anom:c}
 \end{subfigure}\hfill
 \begin{subfigure}[b]{0.4\textwidth}
     \centering
     \includegraphics[width=\textwidth]{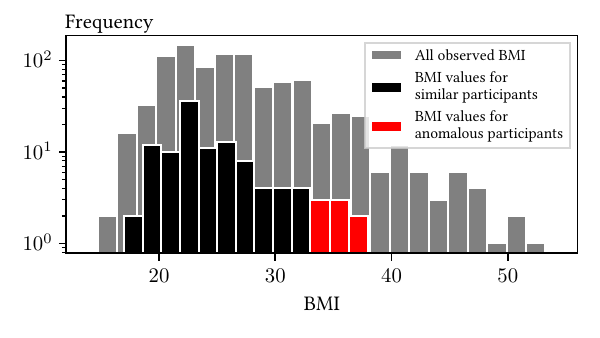}
     \caption{Compare anomalous BMI values to normal ones.}
     \label{fig:case-anom:d}
 \end{subfigure}\hfill
 \begin{subfigure}[b]{0.23\textwidth}
     \centering
     \includegraphics[width=\textwidth]{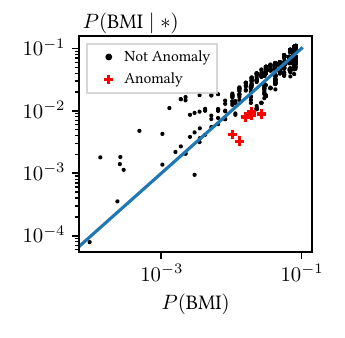}
     \caption{Compare conditional and marginal BMI.}
     \label{fig:case-anom:e}
 \end{subfigure}
\caption{Case study: Anomaly detection in clinical trials.}
\label{fig:case-anom}
\end{figure}

\begin{figure}[t]
\captionsetup[subfigure]{skip=0pt,belowskip=3pt,aboveskip=-3pt}
\captionsetup{skip=0pt}
\begin{subfigure}[b]{0.37\textwidth}
         \centering
         \input{figures/iql-code/gene-0}
         \medskip
         \caption{\footnotesize Overall population}
         \label{fig:case-gene:a}
\end{subfigure}
\hfill
\begin{subfigure}[b]{0.62\textwidth}
         \centering
         \includegraphics[width=\textwidth]{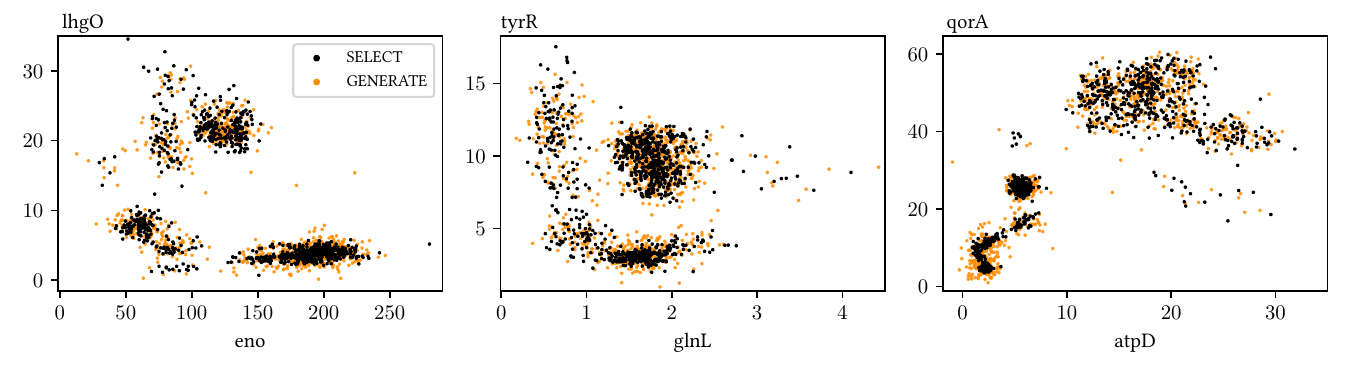}
         \caption{\footnotesize 6 genes from the overall population.}
         \label{fig:case-gene:b}
\end{subfigure}

\begin{subfigure}[b]{0.37\textwidth}
         \centering
         \input{figures/iql-code/gene-1}
         \medskip
         \caption{\footnotesize Add IPTG but not Arabinose}
          \label{fig:case-gene:c}
\end{subfigure}
\hfill
\begin{subfigure}[b]{0.62\textwidth}
         \centering
         \includegraphics[width=\textwidth]{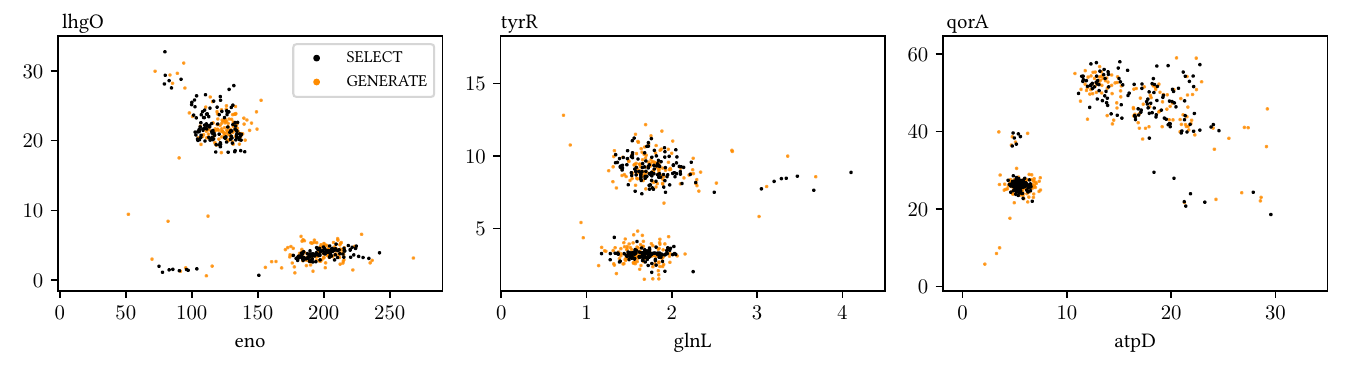}
         \caption{\footnotesize 6 genes from \subref{fig:case-gene:c}.}
         \label{fig:case-gene:d}
\end{subfigure}

\begin{subfigure}[b]{0.37\textwidth}
         \centering
         \input{figures/iql-code/gene-2}
         \medskip
         \caption{\footnotesize Add both IPTG and Arabinose}
         \label{fig:case-gene:e}
\end{subfigure}
\hfill
\begin{subfigure}[b]{0.62\textwidth}
         \centering
         \includegraphics[width=\textwidth]{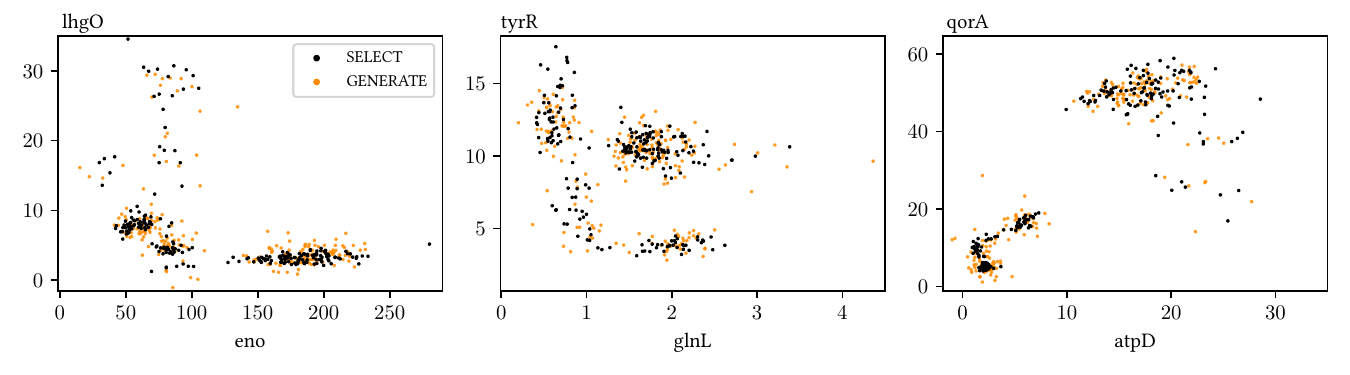}
         \caption{\footnotesize 6 genes from \subref{fig:case-gene:e}.}
         \label{fig:case-gene:f}
\end{subfigure}
\caption{Case study: Conditional synthetic data generation for a virtual wet lab.}
\label{fig:case-gene}

\medskip


\begin{subfigure}[t]{\textwidth}
\centering
\includegraphics[width=.75\textwidth]{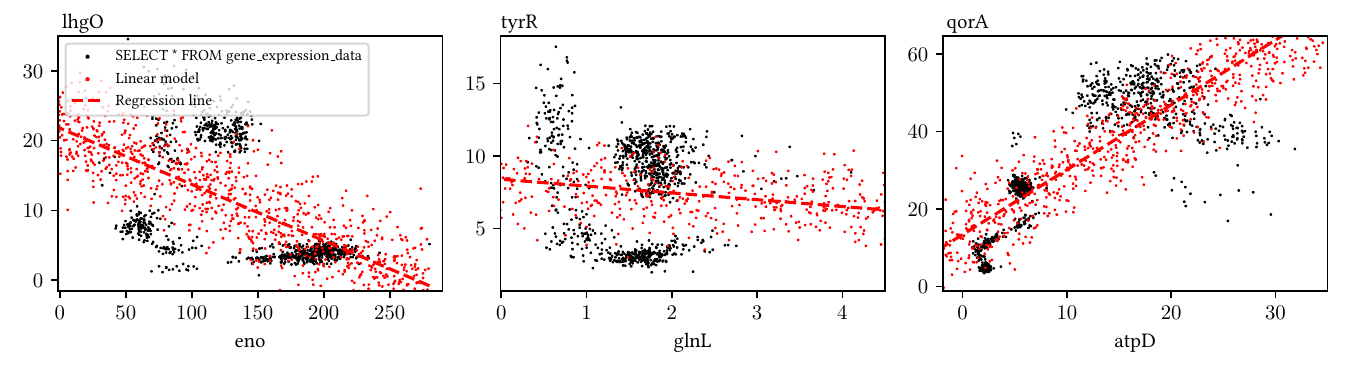}
\caption{\footnotesize Real and generated data from
bivariate linear models (compare to \cref{fig:case-gene:b}).}
\label{fig:case-gene-baseline-comparison:linear}
\end{subfigure}

\begin{subfigure}[t]{\textwidth}
\centering
\includegraphics[width=.75\textwidth]{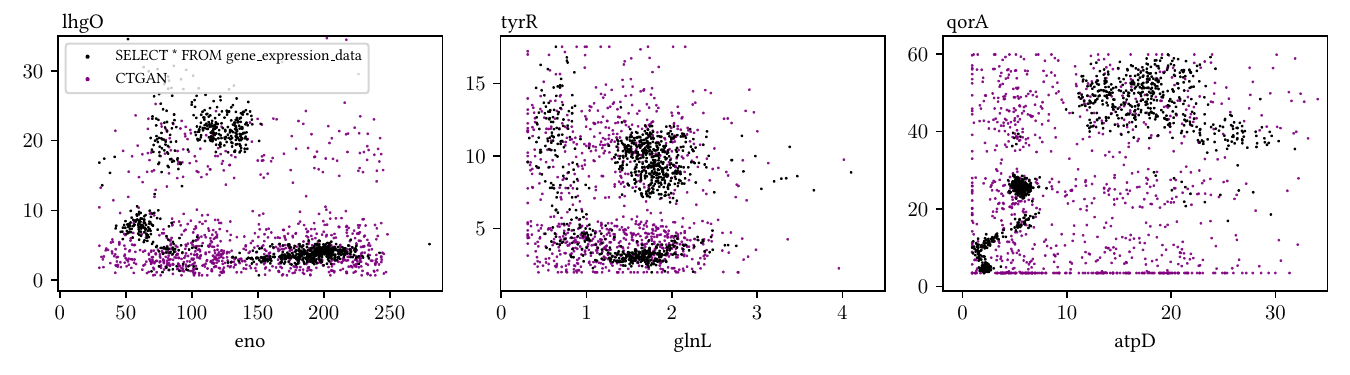}
\caption{\footnotesize Comparison of real and generated data from
CTGAN (compare to \cref{fig:case-gene:b}).}
\label{fig:case-gene-baseline-comparison:gan}
\end{subfigure}

\caption{Linear models and conditional generative adversarial networks (CTGAN~\citep{xu2020modeling})
produce less accurate synthetic virtual wet lab data as compared to
the synthetic data from \iql{} shown in \cref{fig:case-gene:b}.
In \subref{fig:case-gene-baseline-comparison:gan}, the default model
and inference parameters in the open source implementation of CTGAN
is used.}
\label{fig:case-gene-baseline-comparison}
\end{figure}

\paragraph{Anomaly detection in clinical trials.}
The (BEAT19) clinical trial~\citep{beat19} contains data about COVID-19 and records behavior, environment variables, and treatments.
\Cref{fig:case-anom:a} shows a query used for anomaly detection~\citep{BMIcalculator}.
For each row, it computes the model likelihood of the value BMI given the other values of the row. 
The trial participants labeled anomalous (\cref{fig:case-anom:b}) all
report above-average or well-above-average health and that they
exercise, while meeting the World Health Organization's definition of
clinical obesity~\citep{BMIdef}.
\cref{fig:case-anom:d} compares the overall population in the
trial (grey), anomalous individuals (red), and the subset of
the population with the same behavioral covariates (exercise, health status, etc.) as the anomalous individuals (black).
For similar individuals, the data generally suggests a lower BMI.
We can also compare the marginal and the conditional probability of
BMI values in the table of clinical trial records
(\cref{fig:case-anom:d}).
Anomalous data (red) is lower than the diagonal line,
highlighting the ``contextualization'' of BMI values that happens by
conditioning the models: the BMI values are less likely given the
context of the other values in the row, while not necessarily extreme.
To demonstrate this effect, we first apply a$\where$filter that removes
BMI values outside of the 5th and the 95th percentile, excluding one-dimensional extreme values (\cref{fig:case-anom:a}).
We then compute the conditional probabilities of the BMI values in
each row for the remaining data and return anomalies.
\cref{fig:case-anom:c} shows the posterior predictive over the ensemble of models (line) and for each individual model (dots) for a BMI above 30 given exercise and health status.

\paragraph{Conditional synthetic data generation for virtual wet lab.}
\Cref{fig:case-gene} shows synthetic gene expression data generated
using \iql{}, given a dataset from an experiment testing genetic
circuits~\citep{nielsen2016genetic} in \textit{Escherichia coli}.
The synthetic data aligns with the overall population characteristics
(\cref{fig:case-gene:a,fig:case-gene:b}) and accurately reflects the
outcomes of actual experimental interventions
(\crefrange{fig:case-gene:c}{fig:case-gene:f}).
In synthetic biology, the prospect of implementing genetic circuits
has fundamental implications for medical device
engineering~\citep{mansouri2022therapeutic},
bio-sensing~\citep{wang2013modular} and environmental
biotechnology~\citep{xue2022mercury}.
These circuits require input which is typically provided by adding
inducer substances to the culture mediums where the organisms are
grown.
Our figures show the effect of adding two such inducer substances,
Arabinose and IPTG on 6 different host genes.

Producing standard RNA sequencing data can be
costly~\citep{lohman2016evaluation}, especially for new, engineered
organisms that are not mass-produced.
When it is produced, RNA sequencing will yield measurements for gene
expressions for thousands of annotated host
genes~\citep{blattner1997complete}.
These genes are highly interrelated, and knowledge of the relations is
only partially available~\citep{keseler2021ecocyc}.
Thus, the application of generative models to these data presents a
challenging high-dimensional modeling problem, further compounded by
the inherent non-linearity in the data, as illustrated in
\cref{fig:case-gene-baseline-comparison:linear,fig:case-gene-baseline-comparison:gan}.

The most popular approach to modeling gene expression data is linear
regression~\citep{eck2008determining}, as models are easy to interpret and readily available in data analysis libraries.
For non-numerical data, linear regression requires analysis-specific
re-coding of discrete values.
That aside, the low capacity of the model means that it fails to
faithfully reproduce in the actual wet lab data, as shown in
\cref{fig:case-gene-baseline-comparison:linear}.
Conditional generative adversarial networks(CTGAN)~\citep{xu2020modeling}, though more complex, are also an appropriate baseline because they are domain-general and effective at modeling multivariate, heterogeneous data.
However, GANs are hard to interpret and as RNA sequencing data
acquisition is so costly, the number of available training examples
(943) renders it unsuitable for CTGANs.
\cref{fig:case-gene-baseline-comparison:gan} depicts this model class
failing to accurately model the gene expression data.

\section{Related Work}
\label{sec:related-work}

\paragraph{Probabilistic databases.}
Probabilistic databases systems~\citep{suciu2011probabilistic,van2017query} develop efficient algorithms for inference queries on discrete distributions over databases, often based on variants of weighted model counting, for which hardness complexity results were shown and algorithms were developed for tractable cases and efficient approximations.
\citet{cambronero2017} integrate probabilities into a relational
database system to support imputation, while \citet{hilprecht2020} use
probabilistic circuits to improve query performance.
\citet{jampani2011} use probabilistic databases to support random data
generation and simulation.
\citet{cai2013simulation} provides Gibbs sampling support in the space of database tables to a SQL-like language, enabling bayesian machine learning workload such as linear regression or latent Dirichlet allocation.
These languages are typically extensions to SQL or relational
algebra but with limited support for probabilistic models, which they tradeoff for performance.
\citet{schaechtle2022bayesian} presents a preliminary design for
extending SQL to support probabilistic models of tabular data.
Our work differs in that it presents \begin{enumerate*}[label=(\arabic*)]
\item a formalization of the system;
\item a denotational semantics;
\item soundness guarantees for the system;
\item a unified interface that probabilistic models implement;
\item a lowering transform and target lowering language;
\item an extensive performance evaluation; and
\item two new case studies on real-world data.
\end{enumerate*}

\paragraph{Semantics of probabilistic databases.}
\citet{barany2017declarative} and \citet{grohe2022generative} give a semantic
account to probabilistic databases by giving a probabilistic semantics
and guarantees to an extension of Datalog.
\citet{dash2021monads} give a monadic account and denotational
semantics for measurable queries in probabilistic databases.
Their semantics of SQL-like expressions inspired the semantics of our
table expressions.
\citet{grohe2022infinite} established a formal framework for reasoning
about infinite probabilistic databases.
\citet{benzaken2019coq} formalized the semantics of SQL in Coq while
\citet{borya2023formalisation} formalized relational algebra and a
SQL-like syntax using a model checker.

\paragraph{Probabilistic program synthesis.}
\iql{} has been designed with the possibility to leverage powerful
probabilistic program synthesis techniques based on
Bayesian~\citep{mansinghka2016crosscat,saad2019bayesian,adams2010learning}
or
non-Bayesian~\citep{chasins2017data,gens2013learning,grosse2012exploiting,nori2015efficient,hwang2011inducing}
probabilistic model discovery.
The \ami{} provides a unifying approach to expressing powerful Bayesian inference workflows in these probabilistic programs using a high-level SQL-like
language.
Extending the interface to handle synthesized models of time
series~\citep{saad2023sequential,saad2018trcrpm} and/or relational
data~\citep{saad2021hierarchical} is a promising avenue for future
work.

\paragraph{Probabilistic programming systems.}
While we used a Clojure version~\citep{CharchutMEngThesis} of CrossCat
\citep{mansinghka2016crosscat} in our experiments, our system supports
any probabilistic program that satisfies the \tableModel{} interface.
We can thus reuse models written in the variety of PPLs
developed in the literature, such as models written in languages supporting
approximate inference \citep{towner2019gen,
ge2018turing,carpenter2017stan,wood2014new,bingham2019pyro,milch20071,salvatier2016probabilistic}
and exact inference
\citep{saad2021sppl,holtzen2020scaling,gehr2016psi,zaiser2023exact}.
Our model interface is inspired by the SPPL interface
\citep{saad2021sppl} and the CGPM interface
\citep{saad2016probabilistic}.
\citet{gordon2014tabular} propose a probabilistic programming system
using a functional syntax similar to the stochastic lambda calculus,
specialized to inference over relational databases, implemented on top
on Infer.net. It can perform inference tasks such as linear regression
and querying for missing values which enable data imputation,
classification, or clustering.
\citet{borgstrom2016fabular} present a probabilistic DSL and semantics
for regression formulas in the style of the formula DSL in R.
Domain-specific PPLs for tabular data have also been designed to solve
tasks such as data cleaning \citep{lew2021pclean,rekatsinas2017holoclean}.

\paragraph{BayesDB}
Although BayesDB~\citep{mansinghka2015bayesdb} was
motivated by similar goals as \iql{},  \iql{} introduces novel
semantics concepts and soundness theorems that BayesDB did not.
\iql{} also improves upon BayesDB in terms of expressiveness and
performance, as shown in \cref{sec:perf-eval}.
For example, \iql{} queries can be nested and interleaved with SQL,
and also combine results from multiple models.
\iql{} also provides an exact inference engine for a broad class of
sum-product probabilistic programs~\citep{saad2021sppl}.
\revision{ BayesDB, on the other hand, has interesting features that
\iql{} does not yet support such as iterating over model and columns
(e.g. to find pairs of columns with the highest mutual
information)~\citep{saad2017dependencies} and similarity search
between rows~\citep{saad2017search}.
BayesDB also has a ``meta-modeling'' DSL~\citep{saad2016bayesdb} for composing probabilistic
programs from various sources. }

\paragraph{Automated Machine Learning.} Several systems
\citep{agtabular,ledell2020,thornton2013,bergstra2015,olson2016,jin2019}
have been developed to automate the use of discriminative machine
learning methods for analyzing tabular data.
Unlike \iql{}, they do not support the use of generative probabilistic
programs for tabular data satisfying a unified interface (for
sampling, conditioning, and evaluating probabilities or densities)
which enables a single model to be reused across many different tasks.

\section{Contributions}

\iql{} specializes probabilistic programming languages to applications
with tabular data.
It is differentiated from general purpose PPLs in three main ways:
\begin{itemize}[wide]
\item \textbf{Through the \ami{}, \iql{} enables multi-language
workflows.}
Users from different domains and with different expertise should be
able to use probabilistic models for their queries without having to
learn all the details of the PPL in which the model is written.
The \ami{} enables this separation of concerns by providing a
well-specified interface.
It enables the integrating probabilistic models of tabular
data in different languages, as it can be implemented in either a
general-purpose or domain-specific PPL (\cref{sec:implementations-atmi}).
There is no standard way to jointly query models in different PPLs or
use the result of a query in one language against a model in another
language.
As different PPLs focus on different workloads, users of \iql{} can
work with several models written in different PPLs.
\iql{} thus provides a natural multi-language workflow, and our
experiments already use multiple backends (Gen.clj, SPPL, and
ClojureCat).
\item \textbf{\iql{} enables declarative querying.}
No current PPL offers a simple declarative syntax for evaluating
complex queries (e.g., containing elaborate joins and nested selects)
interleaving calls on probabilistic models and database tables.
A number of PPLs provide declarative syntax for specifying and
conditioning models, but the user must decide which operations
on what conditional distributions to evaluate and then
manually combine the results of these operations.
\iql{} relieves the users of such concerns, reducing the chances of
programming errors.
\item \textbf{\iql{} enables reusable performance optimizations.}
Widely used database management systems (DBMS) have been optimized by many
engineer-hours of effort over several decades.
These optimizations are highly reusable because
they are independent of the application domain
and specific languages that the DBMS interfaces with.
\iql{} enables analogous optimizations for workloads that interleave
ordinary database queries with probabilistic inference and generative
modeling.
\iql{}'s optimizations can carry over to many
domains and workflows, avoiding the need for project-specific
performance optimizations involving probabilistic models of
tabular data.
\end{itemize}

We see two interesting avenues for \iql{} to impact database
applications and design.
\paragraph{Integration of \iql{} with database management systems (DBMS)}
First, \iql{} could serve as a query language,
allowing users to query generative models of tabular data directly
from the DBMS.
One use case of rapidly increasing practical importance is querying
synthetic data, generated on the fly to meet user-specified
privacy-utility trade-offs, instead of querying real data that cannot
be shared due to privacy constraints.
Other potential applications for
synthetic data include testing, performance tuning, and sensitivity
analysis of end-to-end data analysis workflows.
In all these settings,
\iql{} implementations could also draw on performance engineering
innovations from DBMS engines, optimized further using the generative
models themselves (e.g., to reduce variance for stratified sampling
approximations to SQL aggregates \citep{AgarwalMPMMS13}).

\paragraph{Modularized development of queries and models}
\iql{} introduces abstractions that isolate query developers and query
users from model developers.
This separation of concerns is analogous to the
physical data independence property achieved by relational databases
\citep{codd1970relational}.
Most database users do not need to know
the details of how data is stored and indexed to be able to query it
efficiently, but some experts do understand how to tune indices to
ensure that databases meet the necessary performance constraints.
Most \iql{} users need not be experts on the details of the
algorithms, modeling assumptions, and software pipelines that produced
the underlying generative models.
Expert statisticians and generative modelers can
still ensure the models are of sufficient quality and tune trade-offs
between performance, maintenance costs, and accuracy, improving models
without invalidating user workflows.
With \iql{}, both typical users and experts can more easily and
interactively query generative models to test their validity, both
qualitatively and quantitatively.
This division of responsibility between users, generative modelers,
and probabilistic programming system developers could potentially help
our society more safely and productively broaden the deployment of
generative models for tabular data.

\section*{Data-Availability Statement}

An artifact providing a version of \href{https://github.com/OpenGen/GenSQL.query}{GenSQL}, and reproducing our experiments, is available \cite{huot_gensql_artifact}.

\begin{acks}
MIT contributors would like to acknowledge support from
\grantsponsor{DARPA}{DARPA}{10.13039/100006502}, under the Machine
Common Sense (MCS) program (Award ID: \grantnum{DARPA}{030523-00001})
Synergistic Discovery and Design (SD2) program (Contract No.
\grantnum{DARPA}{FA8750-17-C-0239}), and Compositionally Organized
Learning To Reason About Novel Experiences (COLTRANE) grant (Contract No.
\grantnum{DARPA}{140D0422C0045}); and unrestricted gifts from
Google, an anonymous donor, and the Siegel Family Foundation.
F.~Saad acknowledges support from the
\grantsponsor{NSF}{National Science Foundation}{https://doi.org/10.13039/100000001}
(Award ID: \grantnum{NSF}{2311983}).
The authors wish to thank the anonymous referees for their valuable
comments and suggestions.
We have also benefited greatly from conversations with and feedback
from Pierre Glaser, Younesse Kaddar, João Loula, Timothy J. O'Donnell, and Fabian Zaiser.
\end{acks}

\bibliographystyle{ACM-Reference-Format}
\bibliography{cited}

\newpage
\appendix

\section{Further experiments and comparisons}

\subsection{Code comparison with Scikit-learn}
\label{sub:code-comparison-with-scikit-learn}

\revision{ \Cref{fig:scikit-learn} shows the implementation of a
single-line \iql{} query
\begin{equation*}
\select \text{class} \from \generate \text{model} \given \text{petal
width}= 0.1 \limit 100
\end{equation*}
in Scikit-learn~\citep{pedregosa2011scikit} on the Iris data from the
UCI ML repository.
Even though the implementation contains data loading and
preprocessing, model building, and a few comments, the model querying
alone is more than 50 lines long and clearly error prone.
\iql{} offers a significant advantage in simplicity over such
baselines. }

\revision{ The reason is that Scikit-learn is a Python library that
focuses on predictive modeling, not a platform for interactive
querying of probabilistic models.
We now detail some aspects of the implementation in Scikit-learn.
Heterogeneously typed mixture models are closely related to the \ami{}
specification we use for many experiments in the model.
The Scikit-learn implementations does not natively support discrete
columns.
Instead, Scikit-learn advises to apply column transformations like
one-hot encoding.
This makes sense for discriminative machine learning pipelines but
less so for multivariate generative models where users want to refer
to columns and their discrete values repeatedly and in two directions
(what goes into the model for conditioning and what comes out).
More importantly, conditioning models on partial information is not
supported in Scikit-learn.
The \lstinline{predict_proba} method for mixture models in
Scikit-learn returns updated weights for latent components which are
at the core of conditioning such a model.
Yet, the \lstinline{predict_proba} method crashes with partial input.
Users have to implement such methods manually before they can refer to
the sample method and post-process its output for the readability of
discrete columns. }

\begin{figure}
    \footnotesize
\begin{python}
import numpy as np
import pandas as pd
from scipy.stats import norm
from sklearn.compose import ColumnTransformer
from sklearn.preprocessing import OneHotEncoder
from sklearn.mixture import GaussianMixture
from sklearn.pipeline import make_pipeline
from ucimlrepo import fetch_ucirepo

###### Fetch Iris data from UCI ML repo #####
iris = fetch_ucirepo(id=53)
data = iris.data["original"]

###### Model building #####
# Define the preprocessing for the categorical features.
categorical_features = ["class"]
categorical_transformer = OneHotEncoder(handle_unknown="ignore")

# Create the ColumnTransformer to apply the transformations to the correct
# columns.
preprocessor = ColumnTransformer(
    transformers=[("categorical-columns", categorical_transformer, categorical_features)],
    remainder="passthrough")

# Use SKlearn to create a mixture model with heterogenous types;
# Gaussians are a bad choice of primitive here but this is
# supposed to be an illustrative example... SKlearn does not actually
# support heterogenously typed mixtures.
pipeline = make_pipeline(
    preprocessor, GaussianMixture(n_components=3, covariance_type="diag")
    ).fit(data)

###### Model querying: GENERATE class UNDER model GIVEN petal width = 0.1 #####
def get_normal_params(component_idx, pipeline, condition_name):
    # Get index of the column in the sklearn internals. The first columns
    # are the one-hot encoded columns, then we add the original index.
    col_index = data["class"].unique().shape[0] + data.columns.get_loc(condition_name)
    # 2. Read out mu.
    m = pipeline.named_steps["gaussianmixture"].means_[component_idx, col_index]
    # 3. Read out sigma.
    sigma = np.sqrt(
        pipeline.named_steps["gaussianmixture"].covariances_[component_idx, col_index])
    return {"mu": m, "sigma": sigma}

def get_comp_score(val, params):
    return norm.pdf(val, params["mu"], params["sigma"])
\end{python}
    \caption{Scikit-learn implementation of a simple \iql{} query.}
    \label{fig:scikit-learn}
\end{figure}

\begin{figure}\ContinuedFloat
    \footnotesize
\begin{python}
def condition_model(pipeline, condition_name, condition_value):
    unnormalized_updated_weights = pipeline.named_steps["gaussianmixture"].weights_
        * np.asarray([
            get_comp_score(
                condition_value, get_normal_params(0, pipeline, condition_name)),
            get_comp_score(
                condition_value, get_normal_params(1, pipeline, condition_name)),
            get_comp_score(
                condition_value, get_normal_params(2, pipeline, condition_name)),
        ])

    updated_weights = unnormalized_updated_weights / sum(unnormalized_updated_weights)
    pipeline.named_steps["gaussianmixture"].weights_ = updated_weights
    return pipeline

# Sample 100 times from conditional model.
conditioned_pipeline = condition_model(
    pipeline, condition_name="petal width", condition_value=0.1)
synethetic_data_with_unreadable_categoricals, _ = conditioned_pipeline.named_steps[
    "gaussianmixture"].sample(100)

# Back transform one-hot encoded class column to readiable iris names.
print(preprocessor
    .named_transformers_["categorical-columns"]
    .inverse_transform(X=synethetic_data_with_unreadable_categoricals[:, -3:]))
\end{python}
\caption{Scikit-learn implementation of a simple \iql{} query (continued).}
\end{figure}

\subsection{Comparison against a baseline using approximate inference}
\label{sub:comparison-against-baselines-using-approximate-inference}

\revision{
In this section, we present some more simple experiments comparing
\iql{} using an exact backend (SPPL) with a Gen.clj baseline using approximate
inference.
The results obtained are expected and confirm the fact that for simple
but non-trivial models, exact inference can be both faster and more
accurate than approximate inference.
This is especially true for rare events for which getting accurate
estimates can be crucial in domains such as risk assessment, fraud
detection, and rare disease diagnosis.
The fact that in \iql{} one can easily change the backend models is
therefore a significant advantage, where one can use approximate
models when the exact model is too slow, or switch to an exact model
when the approximate model is not accurate enough.
We compare the normalized standard deviation of the estimator and the
runtime of the queries. 
We use the same probabilistic model for both
backends, and we compare the results on a $\probof c^0$ query on a
possibly conditioned model. }

\cref{fig:benchmark-variance-genclj} presents a runtime and normalized
standard deviation of the estimator  comparison. 
We compare an exact and an approximate backend using importance sampling on a query
$$\probof c^0 \under m \given c^0_i$$ with varying $c^0_i$.
The same probabilistic model is used, using exact inference for \iql{}
and importance sampling (blue: 5000, red: 10000 importance samples) in
Gen.clj.
\iql{} is faster by orders of magnitude and more accurate than the
baseline Gen.clj model (runtime: in the presence of at least one
condition).


\begin{figure}
\captionsetup{skip=4pt}
\includegraphics[width=\textwidth]{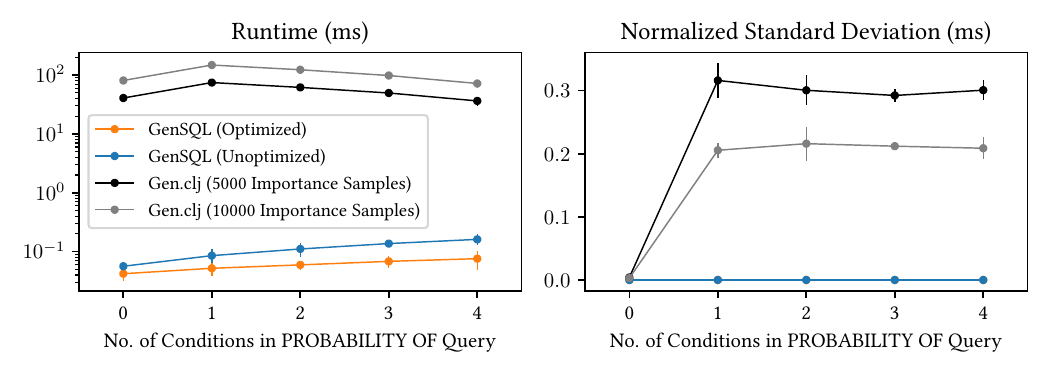}
    \caption{Comparing the effect on runtime (left) and on standard deviation of the estimated probability (right) between \iql{} and Gen.clj~\citep{towner2019gen} on a simple probability query, varying the number of conditions. }
    \label{fig:benchmark-variance-genclj}
\end{figure}

\subsection{Gen.clj code for emission functions}
\label{sub:gen-clj-code-for-emission-functions}

\cref{fig:gen-clj-emission} shows the Gen.clj code for a simple population model and a harmonized model with emission functions that uses age categories as opposed to integers for age.
Such emissions functions are useful for harmonizing models with different data sources. 
For instance, one model can be trained on data with age as an integer, while another model can be trained on data with age categories. 
These two values are not directly comparable, but emission functions can be used to harmonize them, and therefore allow for querying across models trained on different data sources.
While the given example is simple, they are essential for real-world applications where data sources are often heterogeneous.

\begin{figure}
\input{figures/gen-code/emission-function}
    \caption{Gen.clj code for a simple population model and a harmonized model with emission functions.}
    \label{fig:gen-clj-emission}
\end{figure}
\revision{
\section{Details of the optimizations}
\label{sec:correctness-optimizations} }

A dominating computational factor in \iql{} queries for
exact models is often model conditioning, analogously to the role of
join computations in SQL.
For these models, this usually amounts to a program transformation
\citep{saad2021sppl}, or e.g., a costly matrix inversion
\citep{mackay1998introduction}.
Likewise, for both exact and approximate models, a typical \iql{}
workload will involve repeated likelihood evaluations and querying a
conditioned \tableModel{} repeatedly over the rows of a data table.
We can optimize this by caching the results of the likelihood
evaluations and the conditioned \tableModels{} in a way that would not
typically appear in other general purpose PPLs.

\revision{ A second class of optimizations exploits the conditional
independence relations between subsets of the columns of probabilistic
models.
We can simplify a query of the form $\probof c_1 \under \id$ $\given
c_2 \wedge \id.y ~\op~ e$ to $\probof c_1 \under \id \given c_2$ if
$y$ is independent from the joint distribution of the variables
appearing in $c_1$ and $c_2$.
This independence check can sometimes be performed by simply looking
at the structure of the program.
For probabilistic programs produced by CrossCat, we can read that a
variable $y$ is independent from variables $x,z$ if $y$ is not in the
same cluster as $x$ and $z$.
If we see these programs as Bayesian networks, the reason is that no
column-variables have nodes as children, and this is what we exploit
in our implementation.
More generally, one could use more sophisticated algorithms to detect
independence relations, for example in Bayesian networks based on
d-separation \citep{pearl1988probabilistic}.
However, there may be more independence relations that cannot be
directly detected, even in the case of exact models expressed as
sum-product expressions, and in that case more independence relations
can perhaps be detected by using more sophisticated algorithms. }

\revision{ Together, these optimizations can lead to significant
speedups in the evaluation of \iql{} queries.
One reason is that these optimizations can synergize with each other.
To see this, imagine as a first example if $\id \given c \wedge \id.y
= 3$ and $\id \given c \wedge \id.y = 4$ both get simplified to $\id
\given c$.
Then, caching the first conditioned model will hit for the second
query.
As a second example, consider a subquery $\id \given \id.x = (\probof
\id'.y = 3 \under \id')$ appearing at least twice in a query.
If $\id'$ only approximates the computation of the probability query,
then it's unlikely for $\id'$ to compute the same approximation twice.
In that case, caching the result instead allows to compute only one
conditioned model for $\id$, instead of two. }

\revision{
\subsection{Caching}
}

\revision{
\paragraph{Caching conditioned models}
\iql{} is agnostic to the implementation of the \ami{} methods, as
long as they satisfy the interface described in
\cref{sec:impl-table-model}.
In particular, the $\simulate, \logpdf$ and $\prob$ methods are free
to construct and cache any intermediate data structures, such as a
data structure that represents the conditioned model. }

\revision{
\paragraph{Caching exact $\prob$ and $\logpdf$ computations}
A cache hit typically occurs in \iql{} queries when scalar expressions
$e_1$ and $e_2$ in an event reduce to the same value.
To see why, let $\Gamma \vdash c_1 := c' \wedge (\id,i) ~\op~ e_1$
and $\Gamma \vdash c_2 := c' \wedge (\id,i) ~\op~ e_2$ be two lowered
events that differ only in the expressions $e_1$ and $e_2$.
Let $\Gamma \vdash \prob_{\id}(c^0,c^1,c_1)$ and $\Gamma \vdash
\prob_{\id}(c^0,c^1,c_2)$ both be well-typed programs.
Also assume $C[]$ is a given program with a hole such that $\Gamma
\vdash C[\prob_{\id}(c^0,c^1,c_1)]$ and $\Gamma\vdash
C[\prob_{\id}(c^0,c^1,c_2)]$ are well-typed.
Finally, suppose that for some evaluation $\gamma$ of the context
$\Gamma$, $\semex{e_1}(\gamma) = \semex{e_2}(\gamma)$.
Then, we will have $\semex{\prob_{\id}(c^0,c^1,c_1)}(\gamma) =
\semex{\prob_{\id}(c^0,c^1,c_2)}(\gamma)$.
Therefore, we can cache the result of $\prob_{\id}(c^0,c^1,c_1)$ and
reuse it instead of doing the full computation of
$\prob_{\id}(c^0,c^1,c_2)$.
More generally, the following proposition justifies the correctness of
caching in \iql{}. }

\revision{
\begin{proposition}[Correctness of caching (exact computations)] Let
$P$ be either $\prob_{\id}$ or $\logpdf_{\id}$.
Let $C[]$ be a program with a hole such that $\Gamma \vdash
C[P(c^0,c^1,c)]$ is well-typed.
Let $\gamma$ be an evaluation of the context $\Gamma$, and let
$v:=\semex{c}(\gamma)$.
Then, $\Gamma \vdash C[v]$ is well-typed and
$\semex{C[P(c^0,c^1,c)]}(\gamma) = \semex{C[v]}(\gamma)$.
\end{proposition} }

\revision{
\begin{proof} This follows by a straightforward induction on the
structure of the program with a hole $C[]$.
If $C[]:=[]$, then the result follows by the definition of the
semantics and the fact that reals constants $v$ are in the language.
Otherwise, $C[]$ is obtained using a rule from the type system in
\cref{fig:lowered-syntax-full}.
By induction hypothesis, all the subprograms of $C[]$ are well-typed
and the result holds for them.
Therefore, the typing result holds for $C[]$.
In addition, by compositionality of the semantics, the semantic
equality also holds for $C[]$.
\end{proof} }
\revision{ If we take
$v:=\semex{P(c^0,c^1,c_1)}(\gamma)=\semex{P(c^0,c^1,c_2)}(\gamma)$ in
the above proposition, then
\begin{equation*}
\semex{C[P(c^0,c^1,c_1)]}(\gamma) = \semex{C[v]}(\gamma)=
\semex{C[P(c^0,c^1,c_2)]}(\gamma)
\end{equation*}
and we have formally recovered the example from above with events
differing in a scalar $e_1$ and $e_2$.
This is a view on caching as partial evaluation. }

\revision{
\paragraph{Caching of approximate $\prob$ and $\logpdf$ computations} In
the case where $\prob$ or $\logpdf$ are approximated, the argument
using the reduction to a value $v$ above does not directly hold.
Indeed, in general we can have $\semex{c_1}(\gamma) =
\semex{c_2}(\gamma)$ but $\semap{c_1}(\gamma) \neq
\semap{c_2}(\gamma)$.
Caching in this case \emph{does change} the semantics of the program,
but it will not change the asymptotic guarantees of the program. }

\revision{ The key intuition that makes it hold is that if two
sequences of random variables $(x_n)_n$ and $(y_n)_n$ converge to the
same $x$, then for every continuous function $f$, both sequences
$(f(x_n,y_n))_n$ and $(f(x_n,x_n))_n$ will converge to $f(x,x)$.
Here, $x_n$ and $y_n$ are two approximations of the same quantity that
appears twice in a program. The value $x$ is the true value of the
quantity under consideration, and the denotation of the program is
$f(x_n,y_n)$.
With this view, caching consists of using the same approximation $x_n$
twice instead of independently recomputing an approximation $y_n$.
This operation is then valid as long as programs are continuous at
$(x, x)$, which is the case for \emph{safe} (see
\cref{sub:appendix-approx-guarantee}) programs in \iql{}.
In other words, the caching operation may change the approximate
semantics but not the asymptotic guarantees of the program, as long as
the replaced approximation sequence converges to the same value, and
that the query using these approximations is safe. }

\revision{
\begin{proposition}[Correctness of caching (approximate computations)]
Suppose that the following four queries are safe: $\Gamma;\Delta
\vdash t_1$, $\Gamma;\Delta \vdash t_2$, $\Gamma;[] \vdash C[t_1,t_2]$
and $\Gamma;[] \vdash C[t_1,t_1]$.
Let $\gamma$ be an evaluation context for $\Gamma$, and
$\delta,\delta'$ the

obtained evaluations of $\Delta$ for $t_1,t_2$ when evaluating $C[t_1,t_2]$ in $\gamma$. Further assume that $\delta=\delta'$, and that almost surely
\begin{equation*}
\lim_{n\to \infty} ~\semap{\tran{\delta}{t_1}}(\gamma)_n = \lim_{n\to \infty} ~\semap{\tran{\delta}{t_2}}(\gamma).
\end{equation*}
Then, almost surely
\begin{equation*}
\lim_{n\to \infty}~ \semap{\tran{[]}{C[t_1,t_2]}}(\gamma)_n = \lim_{n\to \infty}~\semap{\tran{[]}{C[t_1,t_1]}}(\gamma)_n
\end{equation*}
\end{proposition}
}

\revision{
\begin{proof}
Almost surely,
\begin{align*}
&\lim_n ~\semap{\tran{[]}{C[t_1,t_2]}}(\gamma)_n &\\
&= \sem{C[t_1,t_2]}(\gamma,[]) &\text{\cref{thm:consistent-ami-guarantee} as $C[t_1,t_2]$ safe}\\
&= \sem{C}[\sem{t_1}(\gamma,\delta),\sem{t_2}(\gamma,\delta)] &\text{Compositionality of $\sem{-}$}\\
&= \sem{C}[\sem{t_1}(\gamma,\delta),\lim_{n\to \infty} ~\semap{\tran{\delta}{t_2}}(\gamma)] &\text{\cref{thm:consistent-ami-guarantee} as $t_2$ safe} \\
&= \sem{C}[\sem{t_1}(\gamma,\delta),\lim_{n\to \infty} ~\semap{\tran{\delta}{t_1}}(\gamma)] &\text{Hypothesis on $t_1,t_2$} \\
&= \sem{C}(\gamma,[])[\sem{t_1}(\gamma,\delta),\sem{t_1}(\gamma,\delta)] &\text{\cref{thm:consistent-ami-guarantee} as $t_1$ safe} \\
&= \sem{C[t_1,t_1]}(\gamma,[]) &\text{Compositionality of $\sem{-}$} \\
&= \lim_n~ \semap{\tran{[]}{C[t_1,t_1]}}(\gamma)_n &\text{\cref{thm:consistent-ami-guarantee} as $C[t_1,t_1]$ safe}
\end{align*}
\end{proof}
}

\revision{
\subsection{Independence simplification}
}

\revision{
The formal correctness of independence simplification hinges on the repeated
application of the following proposition.
For ease of expression, we use the following notation in the statement and proof of the proposition.
\begin{itemize}
\item The symbols $\overline{i}$ and $\overline{j}$ are respectively syntactic
sugar for the sequences $i_1, \ldots, i_k$ and $j_1, \ldots, j_k$. Similarly,
$x_{\overline{i}}$ is sugar for the sequence $x_{i_1}, \ldots, x_{i_k}$, and
$\pi_{\overline{i}}$ is sugar for $\pi_{i_1} \otimes \cdots \otimes \pi_{i_k}$.
We also take $\sem{\sigma_{\overline{i}}}$ as sugar for the product
$\sem{\sigma_{i_1}}\times\cdots\times\sem{\sigma_{i_k}}$
and $\sem{t_{\overline{i}}}(\gamma, \delta)$ as sugar for the sequence
$\sem{t_1}(\gamma, \delta), \ldots, \sem{t_k}(\gamma, \delta)$.
\item Given $m \in \sem{M[\id]\{\cols\}}$, we let
\begin{equation*}
p_{\overline{i}|\overline{j}}(x_{\overline{i}} | x_{\overline{j}}) = \Dis(m, \pi_{\overline{j}},x_{\overline{j}})\dotpdf(x_{\overline{i}}).
\end{equation*}
\item As in \cref{sub:language}, $\cols$ is sugar for
$\col_1:\sigma_1,\ldots,\col_n:\sigma_n$. We additionally use
$\cols_{\overline{i}}$ as sugar for
$\col_{i_1}:\sigma_{i_1},\ldots,\col_{i_k}:\sigma_{i_k}$.
\end{itemize}
}
\revision{
\begin{proposition}[Correctness of independence simplification]
Suppose $m \in \sem{M[\id]\{\cols\}}$ and that under $m\dotmeas$ the distributions
of $\cols_{\overline{i}}$ and $\col_j$ are conditionally independent given
$\cols_{\overline{j}}$; i.e. $p_{\overline{i},j|\overline{j}}(x_{\overline{i}}, x_j|x_{\overline{j}})=p_{\overline{i}|\overline{j}}(x_{\overline{i}}|x_{\overline{j}})$, or equivalently, for all measurable functions $f:
\sem{\sigma_{\overline{i}}} \to \RR$ and $g: \sem{\sigma_j}
\to \RR$, we have
\[
\int f(x_{\overline{i}})g(x_j)
    p_{\overline{i},j|\overline{j}}(x_{\overline{i}}, x_j|x_{\overline{j}})dx_{\overline{i}}dx_j =
\left(\int f(x_{\overline{i}})
    p_{\overline{i}|\overline{j}}(x_{\overline{i}}|x_{\overline{j}})dx_{\overline{i}}\right)
\left(\int g(x)p_{j|\overline{j}}(x|x_{\overline{j}})dx\right).
\]
Let $c = \id.\col_{i_1} \op_1 t_1 \land \cdots \land \id.\col_{i_k} \op_k t_k$,
and let $(\Gamma, \Delta)$ be a pair of contexts such that  $\Gamma; \Delta
\vdash c_0 : C^0\{\cols_{\overline{j}}\}$, and $\Gamma, \Delta \vdash c_1 :
C^1\{\cols\}$. Further, let $(\gamma, \delta) \in \sem{\Gamma} \times
\sem{\Delta}$ such that $\sem{\id}(\gamma, \delta) = \mu$.
If $\col_j \not\in \mathbf{vars}(c_0) \cup \mathbf{vars}(c_1)$ and
\[
\Gamma; \Delta \vdash
    (\probof c \under \id \given (\id.\col_j~\op'~t') \land c_0 \given c_1) : \tau
\]
for some $\tau \in \{\posreal, \ranged(0, 1)\}$, then
\begin{align*}
\sem{\probof c \under \id \given (\id.\col_j~\op'~t') \land c_0 \given c_1}(\gamma, \delta) &= \\
&\hspace{-16em}\sem{\probof c \under \id \given c_0 \given c_1}(\gamma, \delta).
\end{align*}
\end{proposition}
\begin{proof}
For ease of notation suppose
$
c_0 = \id.\col_{j_1} = t_1' \land \cdots \land \id.\col_{j_l} = t_l'
$.
We have two cases:
\begin{description}
\item[Case I:] $c : C^0$. In this case we have
\begin{align*}
\sem{\probof c \under \id \given (\id.\col_j~\op'~t') \land c_0 \given c_1}(\gamma, \delta) & \\
&\hspace{-26em}=
\frac{1}{\mu(\sem{c_1}(\gamma,\delta))}
\prod_{\alpha=1}^k
    \mathbf{1}_{\pi_{i_\alpha}(\sem{c_1}(\gamma,\delta))}(\sem{t_\alpha}(\gamma, \delta))
\prod_{\beta=1}^l
    \mathbf{1}_{\pi_{j_\beta}(\sem{c_1}(\gamma,\delta))}(\sem{t_\beta'}(\gamma, \delta))\\
&\hspace{-20em}\times
    p_{\overline{i}|j,\overline{j}}(\sem{t_{\overline{i}}}(\gamma, \delta)
        |\sem{t'}(\gamma, \delta), \sem{t_{\overline{j}}'}(\gamma, \delta)) \\
&\hspace{-26em}=
\frac{1}{\mu(\sem{c_1}(\gamma,\delta))}
\prod_{\alpha=1}^k
    \mathbf{1}_{\pi_{i_\alpha}(\sem{c_1}(\gamma,\delta))}(\sem{t_\alpha}(\gamma, \delta))
\prod_{\beta=1}^l
    \mathbf{1}_{\pi_{j_\beta}(\sem{c_1}(\gamma,\delta))}(\sem{t_\beta'}(\gamma, \delta))\\
&\hspace{-20em}\times
    p_{\overline{i}|\overline{j}}(\sem{t_{\overline{i}}}(\gamma, \delta)
        |\sem{t_{\overline{j}}'}(\gamma, \delta)) \quad\text{by conditional independence} \\
&\hspace{-26em}=
\sem{\probof c \under \id \given c_0 \given c_1}(\gamma, \delta).
\end{align*}
\item[Case II:] $c : C^1$.
Let
\begin{align*}
A_\alpha &:=
    \left\{x \in \sem{\sigma_{i_\alpha}} \middle|
\sem{\op_\alpha}(\gamma, \delta)(x, \sem{t_\alpha}(\gamma, \delta)) = \true\right\}
& 1 \leq \alpha \leq k \\
A &:= A_1 \times \cdots \times A_k.
\end{align*}
We have
\begin{align*}
\sem{\probof c \under \id \given (\id.\col_j~\op'~t') \land c_0 \given c_1}(\gamma, \delta) & \\
&\hspace{-26em}=
\frac{1}{\mu(\sem{c_1}(\gamma,\delta))}
\int \prod_{\alpha=1}^k
    \mathbf{1}_{\pi_{i_\alpha}(\sem{c_1}(\gamma,\delta))}(\sem{t_\alpha}(\gamma, \delta))
\prod_{\beta=1}^l
    \mathbf{1}_{\pi_{j_\beta}(\sem{c_1}(\gamma,\delta))}(\sem{t_\beta'}(\gamma, \delta))\\
&\hspace{-20em}\times
\mathbf{1}_A(x_{\overline{i}})
    p_{\overline{i}|j,\overline{j}}(x_{\overline{i}}
     |\sem{t'}(\gamma, \delta), \sem{t_{\overline{j}}'}(\gamma, \delta))dx_{\overline{i}}\\
&\hspace{-26em}=
\frac{1}{\mu(\sem{c_1}(\gamma,\delta))}
\int \prod_{\alpha=1}^k
    \mathbf{1}_{\pi_{i_\alpha}(\sem{c_1}(\gamma,\delta))}(\sem{t_\alpha}(\gamma, \delta))
\prod_{\beta=1}^l
    \mathbf{1}_{\pi_{j_\beta}(\sem{c_1}(\gamma,\delta))}(\sem{t_\beta'}(\gamma, \delta))\\
&\hspace{-20em}\times
\mathbf{1}_A(x_{\overline{i}})
    p_{\overline{i}|\overline{j}}(x_{\overline{i}}|\sem{t_{\overline{j}}'}(\gamma, \delta))dx_{\overline{i}} \quad\text{by conditional independence}\\
&\hspace{-26em}=\sem{\probof c \under \id \given c_0 \given c_1}(\gamma, \delta).
\end{align*}
\end{description}
\end{proof}
}

\revision{
\section{Implementations of the \ami{}}
\label{sec:implementations-atmi}
}

\revision{
\subsection{\ami{} in terms of sum-product expressions}
Sum-product networks (SPNs) are a useful family of distributions that are
closed under marginalization and conditioning. 
Sum-product expressions (SPEs)~\citep{saad2021sppl} generalize SPNs to gain more expressivity while still allowing for exact marginalization and conditioning.
The idea is that if finitely many distributions are closed under marginalization and conditioning, so are their independent products and weighted mixtures.
%
%
On top of having an API method for sampling, SPEs also have the API methods
$\mathbf{prob}$ for computing the probability of an event and
$\mathbf{logpdf}$ for computing the (marginal) log-density of the distribution
represented by the SPE.
SPEs also have access to the following functions, which are typically implemented as program transforms: 
}
\revision{
\begin{align*}
\cond_0 &: C^0\{\cols'\} \to \SPE \to \SPE \\
\cond_1 &: C^1\{\cols\} \to \SPE \to \SPE \\
\marginalize_{\cols'}&: \SPE \to \SPE.
\end{align*}
The functions $\cond_0$ and $\cond_1$ take in an SPE and an event-0 or event,
respectively, and return a new SPE representing the conditional distribution of
the input SPE given the provided event or event-0.
The family of methods $\marginalize_{\cols'}$ take an SPE and return a new SPE whose distribution is the marginal distribution of $\cols'$ under the input SPE.
They can for instance be used to implement marginal densities in terms of densities.
}

\revision{
With these, we can implement the \ami{} methods as follows:
\begin{align*}
\simulate_\id(c^0, c^1) &= \mathbf{sample}~\$~\cond_0~c^0~\$~\cond_1~c^1~M_\id \\
\logpdf_\id(c^0, c^1, c^0_2) &= \mathbf{logpdf}~c^0_2~\$~\cond_0~c^0~\$~\cond_1~c^1~M_\id \\
\prob_\id(c^0, c^1, c^1_2) &= \mathbf{prob}~c^1_2~\$~\cond_0~c^0~\$~\cond_1~c^1~M_\id.
\end{align*}
}

\revision{
\subsection{\ami{} in terms of truncated multivariate Gaussians}
It is well-known that the family of multivariate Gaussian distributions is
closed under conditioning on a subset of the Gaussian variates. Moreover, if
$\mathbf{x}\sim \mathcal{N}(\mu, \Sigma)$ is a $d$-dimensional multivariate
Gaussian random variable, $A$ is a $m \times d$ matrix, and $\mathbf{u},
\mathbf{l} \in \RR^m$, then conditional on the event $\{\mathbf{l} \leq
A\mathbf{x} \leq \mathbf{u}\}$ the random variable $\mathbf{x}$ follows a
truncated multivariate Gaussian distribution.
}

\revision{
If we focus on the set of events of the form $\{\mathbf{l} \leq A\mathbf{x}
\leq \mathbf{u}\}$, then similar to sum-product expressions, the family of
truncated multivariate Gaussians will be closed under conditioning. This
results from the fact that multivariate Gaussians are closed under conditioning
on event-0s and conditioning a truncated multivariate Gaussian on an event of
the above form is equivalent to truncating it further. That is, if we represent
a truncated multivariate Gaussian by the tuple $(\mu, \Sigma, A, \mathbf{l},
\mathbf{u})$, we can define the events
\begin{align*}
\cond &: C^0\{\cols'\} \to \TMVG \to \TMVG \\
\truncate &: C^1\{\cols\} \to \TMVG \to \TMVG,
\end{align*}
where $\TMVG$ denotes the type of truncated multivariate Gaussian terms.
}

\revision{
Moreover, it is possible to define distribution functions for multivariate
Gaussians distributions allowing us to calculate the probability mass assigned
to hyperrectangles in the domain of the distribution. 
We note that such calculations require the evaluation of the distribution function at every
corner of the hyperrectangles, and the number of these corners can be 
exponential in the dimension of the hyperrectangle. 
Furthermore, there are several exact methods for exactly calculating the normalizing constant of truncated multivariate Gaussian variables. 
We refer the reader to the textbook
\cite{genz2009computation} for the details, and assume access to methods
$\mathbf{logpdf}$ and $\mathbf{prob}$ for calculating log-marginal densities and probabilities, respectively.
}

\revision{
Using the functions $\truncate$, $\cond$, $\mathbf{logpdf}$, and
$\mathbf{prob}$ the implementation of the \ami{} methods for truncated
multivariate Gaussians is almost identical to that of SPEs, as follows.
\begin{align*}
\simulate_\id(c^0, c^1) &= \mathbf{sample}~\$~\cond_0~c^0~\$~\truncate~c^1~M_\id \\
\logpdf_\id(c^0, c^1, c^0_2) &= \mathbf{logpdf}~c^0_2~\$~\cond~c^0~\$~\truncate~c^1~M_\id \\
\prob_\id(c^0, c^1, c^1_2) &= \mathbf{prob}~c^1_2~\$~\cond~c^0~\$~\truncate~c^1~M_\id.
\end{align*}
}

\revision{
\subsection{\ami{} in terms of ancestral sampling}
}

\revision{
In general, exact Bayesian inference in probabilistic models is intractable.
As such, practitioners often rely on variational or Monte Carlo methods for
approximate inference. Probabilistic programming languages that support
programmable inference allow their users to seamlessly incorporate
approximate inference methods into their workflow. We now describe how \iql{}'s
\ami{} can be implemented in such probabilistic programming languages. As a
simple representative example of an approximate inference algorithm we chose
ancestral sampling, a simple sequential Monte Carlo method.
}

\revision{
We do not introduce ancestral sampling in detail here, and refer the unfamiliar
reader to \cite{robert1999monte}. 
As in the previous sections, we describe our implementation in a superset of the lowered language. 
We assume the implementation language contains a parametric type
$\Model\{\cols\}$ representing a probabilistic model describing a joint
distribution on the columns $\cols$. 
Given this type, the type signature of the ancestral sampling algorithm is given by
$$
\ancestral: \Model\{\cols\} \to C^0\{\cols'\} \to C^1\{\cols\} \to
((\sigma_1, \ldots, \sigma_n), \real).
$$
In the above type signature, the first argument denotes the probabilistic
model under consideration. 
The remaining arguments denote the event-0 and event on which the model is being conditioned. 
The algorithm returns a weighted sample $(x, w)$, where $w$ is the log-importance weight of $x$.
}

\revision{
Recall that the approximate implementations of the \ami{} methods are indexed
by a ``compute budget'' $n$. Here, $n$ will denote the number of independent
ancestral samples we will generate to perform the desired computation.
}

\revision{
Throughout, we suppose $M_\id$ is a probabilistic model in the host language
implementing the row model specified by $\id$.
}

\revision{
To implement $\simulate_\id(c^0, c^1)$, we first generate $n$ independent
ancestral samples by calling $\ancestral~ M_\id~ c^0~c^1$ for $n$ times. 
If the importance weights of all the generated particles are zero (i.e. the return second element of all the returned tuples are $-\infty$) then the generated sample is incompatible with the event described by $c^1$. 
As such, in this case we return a null sample $(\star, \ldots, \star)$. 
Note that as the number of particles $n$ increases, the probability of generating such a collection of
particles goes to zero, and so the implementation of simulate will generate a
non-null sample almost surely. 
On the other hand, if at least one importance
weight is non-zero, we re-sample one of the generated particles based on its
weight and return it. 
The sampling distribution of the return value of this
procedure converges to the conditioned row model by an elementary statistical
argument \cite{robert1999monte}. The Haskell-style pseudo-code of the
$\simulate$ method is given below.
\begin{align*}
\simulate_\id(c^0, c^1) &= \llet \mathrm{particles} = \replicate~n~\$~\ancestral~M_\id~c^0~c^1~\iin \\
&\hspace{4mm}\llet \mathrm{logits} = \map \pi_2~\mathrm{particles}~\iin \\
&\hspace{8mm}\iif \mathbf{all}~ (==-\infty)~ \mathrm{logits}~\tthen (\star,\ldots,\star) \eelse \\
&\hspace{12mm}\llet i = \mathbf{Categorical}~\$~\mathbf{logitsToProbs}~\mathrm{logits}~\iin\\
&\hspace{16mm}\pi_1\ \mathrm{particles}_i
\end{align*}
}

\revision{
To implement $\logpdf(c^0, c^1, c^0_2)$ we rely on the fact that the expected
value of an importance weight is equal to the joint density of the variables
whose values are given \cite{lew2023probabilistic}. 
Using this fact, we first
define a helper function $\mathbf{logmarginal}$ which estimates the log
marginal density of a given assignment of variables to values expressed as an
event-0. That is, we define
$$
\mathbf{logmarginal}: \Model\{\cols\} \to C^0\{\cols'\} \to C^1\{\cols\} \to
\real,
$$
$$
\mathbf{logmarginal}~M~c^0~c^1 =
\mathbf{logmeanexp}~\$~\map~\pi_2~(\replicate~n~\$~\ancestral~M~c^0~c^1)
$$
where
$
\mathbf{logmeanexp}~\mathrm{xs} := 
(\mathbf{logsumexp}~\mathrm{xs}) - (\log~\$~\mathbf{length}~\mathrm{xs})
$.
Now, as conditional densities are given as a ratio of marginal densities we
will have the implementation of $\logpdf_\id(c^0,c^1, c^0_2)$ as follows
$$
\logpdf_\id(c^0,c^1, c^0_2) = (\mathbf{logmarginal}~M_\id~(c^0\land c^0_2)~c^1)
- (\mathbf{logmarginal}~M_\id~c^0~c^1).
$$
}

\revision{
Note that as $n \to \infty$, the return value of this implementation will
almost surely converge to the correct value. 
This is because by the strong law
of large numbers the results of $\mathbf{logmarginal}$ will almost surely
converge to the correct value, and as division and logarithm are continuous
functions the value of $\logpdf_\id(c^0,c^1, c^0_2)$ will converge to the
correct value.
}

\revision{
Lastly, to implement the $\prob$ method we rely on the fact that we can perform
Monte Carlo integration using importance samples. 
Note that the probability of
an event under a probability measure is given by the integral of the indicator
function of that event under the same measure. 
We have the implementation
\begin{align*}
\prob_\id(c^0, c^1, c^1_2) &= \llet \mathrm{particles} = \replicate~n~\$~\ancestral~M_\id~c^0~c^1~\iin \\
&\hspace{8mm}\exp\left(
\mathbf{logmeanexp}~[w \times (\mathbf{logindicator}~c^1~x) | (x, w) \in \mathrm{particles}]\right.\\
&\hspace{16mm}- \left.\mathbf{logmarginal}~M_\id~c^0~c^1\right),
\end{align*}
where
$\mathbf{logindicator}~c^1~x := \iif~x \in c^1~\tthen~0~\eelse -\infty$.

Note that we need the correction term $\mathbf{logmarginal}~M_\id~c^0~c^1$, as
Monte Carlo integral estimators obtained via ancestral sampling estimate the
integral under consideration up to a constant factor which is the joint density
of the given variables.
}

\section{Lowering Details}
\label{sec:lowering-appendix}

\subsection{Normalization of \iql{} Queries}
\label{sub:normalization}

Our implementation allows arbitrary nested conditionings on events and event-0 which do not fit the restriction from the type system presented in \cref{sec:language}.
Our normalization has two stages (\cref{fig:algo-norm,fig:algo-norm-2}), which rewrite queries so they satisfy the formalization presented in \cref{sec:language}, and simplifies the queries to a normal form.

The first part of the normalization does three things.
\begin{enumerate}
    \item It partially evaluates$\rename$ clauses on models.
    \item It aggregates events in$\given$clauses.
    \item It removes self-contradictory statements. For instance, $\probof \id.x=7 \wedge \id.y=8 \under \id \given
    \id.y=4$ should technically return 0 as the event is impossible, but
    the normalization will return $\probof \id.x=7 \under \id \given
    \id.y=4$ instead.
    These can also appear through nested$\given$ clauses, e.g. $(\id
    \given \id.\col_1 = 7) \given \id.\col_1 = 8$.
    In this case, the normalization pass will remove the second$\given$
    clause.
    In other words, variables which are conditioned by a 0-event behave in
    the density of the model like variables which have been
    marginalized-out and are not present in the model anymore.
    The same logic applies to the probability of an event under a model,
    but this case does not need a special treatment.
    Ideally, an implementation should issue a warning to the user letting
    them know that there was probably a mistake in the query.
 \end{enumerate}

The second part of the normalization pass is only needed for the implementation and ensures that we do not evaluate a density at a point where some of the indices have been conditioned on.
This avoids another subtle issue when conditioning on event-0.
Briefly, when conditioning on an event-0 such as $\id.\col = 7$, we
are changing the base measure on $\col$ to be a Dirac $\delta$ measure
at $7$, and the density evaluation would incorrectly assume that it is
still the original base measure.
In the second part of the normalization, a temporary name $\true$ is
used which is not part of the language syntax. 
It gets eliminated during the normalization.
After applying the normalization passes, we obtain the following normal
forms:

\begin{proposition}
    \label{prop:norm_form_\tableModel}
    The rewrites from \cref{fig:algo-norm} are confluent, terminating, and lead to the following normal forms, where$\rename$and$\given$clauses are optional:
    \begin{itemize}
        \item \tableModels:$\rename (\id \given c^0 \given c^1) \as \id'$.
        \item Probability queries:$\probof c^i_1 \under (\id \given c^0 \given c^1)$.
        \item Generate queries:$\generate (\id \given c^0 \given c^1) \limit e$ and \\ $t \modjoin (\id \given c^0 \given c^1)$.
    \end{itemize}
    After the rewrites from \cref{fig:algo-norm-2}, we can further assume that the variables in $c^0_1$ and $c^0$ are disjoint.
\end{proposition}
\begin{proof}
Termination: we define the following valuation function $\val$ on expressions:
\begin{itemize}
\item $\val(\id) = 1$
\item $\val(\rename m \as \id') =1+\val(m)$
\item $\val(m \given c^0) = 2* \val(m)$
\item $\val(m \given c^1) = 2* \val(m)+1$
\item $\val(\probof c^i \under m) =\val(m)$
\item $\val(\generate m \limit e) =\val(m)$
\item $\val(t \modjoin m) =\val(m)$
\end{itemize}
Every rewrite rule strictly decreases the valuation of the expression, and thus the rewrite system terminates.

Confluence (sketch): this stems from the fact that the rewrites are generated from all critical pairs and the rules for the self-critical pairs are associative, which can be checked by inspection. 

Disjointness (sketch): a straightforward induction on the normalization rules including $P$ show that the variables in $c^0_1$ and $c^0$ are disjoint after $P$ is applied.
\end{proof}

\begin{figure}
    \fbox{
        \parbox{\textwidth}{
    \centering
    \small
   \begin{tabular}{lcl}
    $(\rename m \as \id') \given c^i$ & $\leadsto$ & $\rename (m \given c^i[\id'/\id]) \as \id' $\\
        $\rename (\rename m \as \id') \as \id''$ & $\leadsto$ & $\rename m \as \id''$ \\
       $(m \given c^0) \given {c^0}'$ & $\leadsto$ & $m \given S(c^0, {c^0}')$\\
        $(m \given c^1) \given {c^1}'$ & $\leadsto$ & $m \given c^1\wedge {c^1}'$\\
         $(m \given c^1) \given c^0$ & $\leadsto$ &  $(m \given c^0) \given c^1$ \\
       $\probof c^i \under (\rename m \as \id') $ & $\leadsto$ &  $\probof c^i[\id'/\id] \under m$ \\
        $\generate (\rename m \as \id) \limit e$ & $\leadsto$ &  $\generate m \limit e $ \\
        $t \modjoin (\rename m \as \id) $ & $\leadsto$ &  $ t \modjoin m$ \\
        $S(c^0, [])$ & $\leadsto$ & $c^0$ \\
        $S(c^0:\densitytype\{\cols\}, (\id.\col = e) \wedge {c^0}')$ & $\leadsto$ & $S(c^0 \wedge (\id.\col = e), {c^0}')$ \quad if $\col \not\in \cols$ \\
        && $S(c^0, {c^0}')$ \hspace{2.2cm} otherwise \\
   \end{tabular}
        }}
    \caption{Normalization rules for \tableModels{} (first phase).}
    \label{fig:algo-norm}
\end{figure}

\begin{figure}
    \fbox{
        \parbox{\textwidth}{
            \centering
\begin{tabular}{lcl}
    $\probof c^0_2 \under (\id \given c^0 \given c^1)$ & $\leadsto$ & $\probof P(c^0_2,c^0) \under$ \\
    && $\id \given c^0 \given c^1$ \\
    $P(c^0_2 \wedge (\id.\col = e), c^0:\densitytype\{\cols\})$ & $\leadsto$ &   $P(c^0_2, c^0)$ \hspace{2.2cm} if $\col \in \cols$ \\
    && $P(c^0_2, c^0)\wedge (\id.\col = e)$ \quad otherwise \\
    $P([], c^0)$ & $\leadsto$ & $\true$ \\
    $\true \wedge c^0$ & $\leadsto$ & $c^0$ \\
    $\probof \true \under m$ & $\leadsto$ & $1$
\end{tabular}
        }}
\caption{Normalization rules for \tableModels{} (second phase).}
    \label{fig:algo-norm-2}
\end{figure}

\Cref{fig:lowered-syntax-full} shows the full syntax of the lowered
language.
\Cref{fig:lower-exact-semantics} shows the measure semantics of the
lowered language for \tableModels{} satisfying the exact \ami{}.

\subsection{Lowered language}

\cref{fig:lower-exact-semantics} shows the full syntax of the lowered language. Compared to the abbrieviated syntax in \cref{fig:lowered-syntax}, we have added the builtin functions. 
\cref{fig:lower-exact-semantics} presents the exact semantics of the lowered language. 
The map $l:\measm \bag\ X\to \bag~\measm X$ in the semantics of $\mapreduce$ is the distributive law for the point process monad
\citep{dash2021monads}. 
The semantics of builtin functions is their standard mathematical counterpart.

\begin{figure}
    \fbox{
        \parbox{\textwidth}{
            \small
\begin{align*}
    \text{base type } \sigma &::= \sigma_c \gor  \sigma_d  
    &\hspace{-2mm}\text{ground type } \sigma_g&::= \sigma \gor (\sigma_1, \ldots, \sigma_n) 
    &\text{event type } \evtype &::= \eventtype[\sigma_g] \gor \densitytype[\sigma_g]
    \\
    \text{type } \tau &::= \bag[\sigma_g]
    &\text{operator } \op\ &::= + \gor  - \gor  \times \gor \div \gor \land \gor  \lor \gor\ = 
    &\text{\tableModel{} } \modtype &::= M[\sigma_g]
\end{align*}
\begin{align*}
    \text{builtin function } f\ &::= \mapreduce \gor \map \gor \filter \gor \replicate\gor \hostexp\gor \singleton\gor \lowerjoin\\
    \text{special primitives } &::=\simulate\gor \prob\gor \logpdf \\
    \text{term } t &::= c \gor \id \gor f(t_1, \ldots, t_n) \gor x \gor (t_1, \ldots, t_n) \gor \pi_i\ t \gor t_1\ op\ t_2 \\
    \end{align*}
    \vspace{-8mm}
    $$
    \frac{\Gamma \vdash t:\sigma_i}
    {\Gamma,\id: M[(\sigma_1, \ldots,\sigma_n)] \vdash (\id, i) = t : \densitytype[\sigma_i]}
    \quad
    \frac{\Gamma\vdash t_1 : \densitytype[\sigma_g^1]\quad \Gamma\vdash t_2 : C^0[\sigma_g^2]}
    {\Gamma \vdash t_1 \land t_2: \densitytype[\sigma_g^1, \sigma_g^2]}
    $$
    $$
    \frac{\Gamma\vdash t:\sigma_i\quad \op\in\{=,<,>\}\quad (\sigma_i,\op)\neq (\sigma_c,=)}
         {\Gamma, \id: M[(\sigma_1, \ldots,\sigma_n)] \vdash (\id, i)\ \op\ t : \eventtype[(\sigma_1,\ldots,\sigma_n)]}
    \quad
    \frac{\Gamma\vdash t_1 : \eventtype[\sigma_g]\quad \Gamma\vdash t_2 : \eventtype[\sigma_g]\quad \op\in\{\land,\lor \}}
    {\Gamma \vdash t_1\ \op\ t_2: \eventtype[\sigma_g]} 
    $$
    $$
    \frac{\Gamma\vdash t_1: \nat\ \ \Gamma\vdash t_2: \bag[\sigma_g]}
         {\Gamma\vdash \replicate(t_1, t_2): \bag[\sigma_g]}
\quad
         \frac{\Gamma\vdash t_1: \bag[(\sigma_1, \ldots, \sigma_n)]\quad \Gamma\vdash t_2: \bag[(\sigma_{n+1}, \ldots, \sigma_{n+m})]}
         {\Gamma\vdash\lowerjoin(t_1,t_2) : \bag[(\sigma_1, \ldots, \sigma_{n+m})]}
    $$
    $$
    \frac{\Gamma, x:\sigma_g^2\vdash t_1: \sigma_g^1 \quad
          \Gamma \vdash t_2: \bag[\sigma_g^2]}
         {\Gamma\vdash \map (x.t_1)\ t_2: \bag[\sigma_g^1]}
         \quad
         \frac{\Gamma, \id: M[\sigma_g]\vdash 
         c^i: C^i[(\sigma_1, \ldots,\sigma_n)]}
              {\Gamma, \id: M[\sigma_g]\vdash \simulate_{\id}(c^0,c^1): \bag[\sigma_g]}
         $$
         $$
        \frac{}{\Gamma, \id: \bag[\sigma_g]\vdash \id: \bag[\sigma_g]}
         \quad
         \frac{\Gamma, \id: M[\sigma_g]\vdash 
         c^i: C^i[\sigma_g]\quad \Gamma, \id: M[\sigma_g]\vdash c^0_1: \densitytype[\sigma_g]}
              {\Gamma, \id: M[\sigma_g]\vdash \logpdf_{\id}(c^0, c^1,c^0_1): \real}
         $$
         $$
         \frac{\Gamma, \id: M[\sigma_g]\vdash 
         c^i: C^i[\sigma_g]\quad \Gamma, \id: M[\sigma_g]\vdash c^1_1: \eventtype[\sigma_g]}
              {\Gamma, \id: M[\sigma_g]\vdash \prob_{\id}(c^0, c^1,c^1_1): \ranged(0,1)}
              \quad 
              \frac{\Gamma,x:\sigma_g \vdash t_1:\bool\quad \Gamma\vdash t_2:\bag[\sigma_g]}
              {\Gamma\vdash \filter (x.t_1)\ t_2: \bag[\sigma_g]}
              $$ 
    $$
    \frac{\Gamma \vdash t : \sigma_g}
         {\Gamma \vdash \singleton(t) : \bag[\sigma_g]}
    \quad
    \frac{\Gamma,x:\sigma_g^2 \vdash t_1 : \bag[\sigma_g^1] \quad
          \Gamma \vdash t_2 : \bag[\sigma_g^2]}{
          \Gamma \vdash \mapreduce (x.t_1)\ t_2 : \bag[\sigma_g^1]
    }
    \quad \frac{\Gamma \vdash t_i: \sigma}{\Gamma \vdash t_1\ \op\ t_2: \sigma}(\op:\sigma,\sigma\to\sigma)
    $$
    $$
    \frac{}{\Gamma \vdash c:\sigma}
    \quad
    \frac{}{\Gamma, x:\tau \vdash x:\tau}
    \quad
    \frac{\Gamma\vdash x:\real}{\Gamma\vdash \hostexp(x):\posreal}
    \quad 
    \frac{\Gamma\vdash t:(\sigma_1,\ldots,\sigma_n)}{\Gamma\vdash \pi_i(t): \sigma_i}
    \quad 
    \frac{\Gamma\vdash t_i:\sigma_i\quad i=1..n}{\Gamma\vdash (t_1,\ldots,t_n):(\sigma_1,\ldots,\sigma_n)}
    $$
         }}
    \caption{Full syntax and type system of the lowered language.}
    \label{fig:lowered-syntax-full}
\end{figure}

\begin{figure}
\fbox{
\parbox{\textwidth}{
\small
 \begin{center}
\textbf{Semantics of types and contexts}
\end{center}
\begin{align*}
&\semex{\sigma} = \sem{\sigma} \qquad \semex{(\sigma_1,\ldots,\sigma_n)} = \semex{\sigma_1} \times \ldots \times \semex{\sigma_n} \qquad
\semex{\bag[\sigma_g]} = \measm\bag(\sem{\sigma_g}) \\
&\semex{M[\sigma_g]} = \admissible \left(\semex{\sigma_g}\right) \qquad
\semex{\Gamma, x:\tau} = \semex{\Gamma} \times \semex{\tau} \qquad \semex{[]} = 1
\end{align*}
\begin{center}
\textbf{Semantics of terms}
\end{center}
\begin{align*}
\semex{(t_1,\ldots, t_n)}(\gamma) &=
        (\semex{t_1}(\gamma),\ldots, \semex{t_n}(\gamma)) \\
    \semex{\pi_i\ t}(\gamma) &= \pi_i(\semex{t}(\gamma)) \\
    \semex{\op(t_1, t_2)}(\gamma) &= \op(\semex{t_1}(\gamma), \semex{t_2}(\gamma)) \\
\semex{\id}(\gamma) &= \gamma(\id) \\
\semex{\simulate_\id(c^0,c^1)}(\gamma) &= 
\llet (\pi,v) = \semex{c^0}(\gamma) \ \iin \llet m = \Dis(\gamma(\id),\pi,v)\ \iin\\
&\quad \llet E = \semex{c^1}(\gamma)\  \iin \cond(m,E)\\
\semex{\logpdf_\id(c^0,c^1,c^0_2)}(\gamma) &= \llet (\pi,v) = \semex{c^0_2}(\gamma) \ \iin \log(\semex{\simulate_\id(c^0,c^1)}(\gamma)\dotpdf(v)) \\
\semex{\prob_\id(c^0,c^1,c^1_2)}(\gamma) &=  \semex{\simulate_\id(c^0,c^1)}(\gamma)\dotmeas(\semex{c^1_2}(\gamma)) \\
\semex{(\id,i) ~\op~ t: \mathcal{C}^1[\sigma_g]}(\gamma) &= \{(x_1,\ldots,x_n) \in \semex{\sigma_g} ~\mid~ x_i ~\op_l~ \semex{t}(\gamma)\}\\
\semex{(\id,i) = t :\mathcal{C}^0[\sigma_g]}(\gamma) &= (\pi_i,\semex{t}(\gamma))\\
\semex{c^1_1\wedge c^1_2 }(\gamma) &= \semex{c^1_1}(\gamma) \cap \semex{c^1_2}(\gamma)\\
\semex{c^1_1\vee c^1_2}(\gamma) &= \semex{c^1_1}(\gamma) \cup \semex{c^1_2}(\gamma)\\
\semex{c^0_1\wedge c^0_2 }(\gamma) &= \llet_{1\leq i\leq 2}(f_i,v_i) = \semex{c^0_i}(\gamma)\ \iin\ \big(\lambda x.(f_1(x),f_2(x)),(v_1,v_2)\big) \\
\semex{\mapreduce (x.t_1)\ t_2}(\gamma) &= \begin{aligned}[t]
    l \big[
        \semex{t_2}(\gamma)
        \bind \left(\lambda S.\left\{\lambda x'.\semex{t_1}(\gamma[x\mapsto x'])\ y\ \middle|\ y \in S\right\}\right)
    \big]
    \end{aligned}\\
\semex{\map (x.t_1)\ t_2}(\gamma) &= \begin{aligned}[t]
    \semex{t_2}(\gamma)
    \bind \left(\lambda S. \return \left\{\lambda x'.\semex{t_1}(\gamma[x\mapsto x'])\ y\ \middle|\ y \in S\right\}\right)
    \end{aligned}\\
\semex{\filter (x.t_1)\ t_2}(\gamma) &= \begin{aligned}[t]
    \semex{t_2}(\gamma)
    \bind \left(\lambda S.\return \left\{y \in S \middle| \lambda x'.\semex{t_1}(\gamma[x\mapsto x'])\ y\right\}\right)
    \end{aligned}\\
\semex{\replicate(t_1, t_2)}(\gamma) &= \begin{aligned}[t]
    &\llet m = \semex{t_1}(\gamma)\  \iin \\
    &\semex{t_2}(\gamma)
        \bind \lambda \mu. \textstyle \bigotimes_{i=1}^m \mu
        \bind \left(\lambda (A_1, \ldots, A_m). \return \cup_{i=1}^m A_i\right)
    \end{aligned} \\
\semex{\hostexp(t_1)}(\gamma) &= \exp\left(\semex{t_1}(\gamma)\right) \\
\semex{\singleton(t_1)}(\gamma) &= \return \left\{\semex{t_1}(\gamma)\right\} \\
\semex{\lowerjoin(t_1, t_2)}(\gamma) &= \begin{aligned}[t]
    &\semex{t_1}(\gamma) \otimes \semex{t_2}(\gamma)
    \bind \left(\lambda S, S'. \return (\maptwo \splat S\ S')\right)
    \end{aligned}\\
\end{align*}
}}
\caption{Exact semantics of the lowered language.}
\label{fig:lower-exact-semantics}
\end{figure}

\subsection{Proof of the exact guarantee theorem}
\label{sub:proof-exact}

First note that the interpretation of the contexts $\Gamma$ in \iql{} and the lowered language for closed expressions $\Gamma, [] \vdash t:\tau$ are the same, i.e. $\semex{\tranty{\Gamma}} = \sem{\Gamma}$.
This is checked by direct inspection. 
Given a local context $\Delta$ and an evaluation $\delta$ of that context, we write $\delta' = \tranty{\delta}$ to denote the corresponding evaluation of the lowered context $\tranty{\Delta}$. 
$\tranty{\delta}$ is defined as the identity on the table keys, and removes model keys.
For general expressions $\ctx \vdash t:\tau$, the context $\Gamma$ of the lowered and \iql{} program is not necessarily the same. 
In this case, $\tran{\delta_1}{t}$ will be in a context $\Gamma'$ that consists of $\tranty{\Gamma}$ as well as a new variable for each table in $\Delta$. 
That is, $\Gamma'$ adds a renamed version of each table type in $\Delta$ to $\tranty{\Gamma}$.
Likewise, an evaluation context $\gamma,\delta$ will translate to an evaluation context $\gamma'$ in the lowered language, where $\gamma'$ is the same as $\gamma$ on the variables in $\Gamma$, and for every key $k$ at position $i$ in $\delta_1$, $\gamma'(k)=\delta(\delta.\text{key}[i])$.
Then, the proof of the theorem is by induction on the structure of the \iql{} program. 
More precisely, we show the following. 
Let $\ctx \vdash t: \tau$ be a \iql{} program. 
Then, for all evaluation of the context $\gamma, \delta$, we have that $\sem{t}(\gamma, \delta)= \sem{\tran{\delta'}t}(\gamma')$, where $\delta'= \tranty{\delta}$.
With no loss of generality, we simplify the problem by assuming that the queries are normalized, i.e. models $m$ can be assumed to be of the form $\id \given c^0 \given c^1$, where$\given$clauses are optional.
In addition, we will only cover the case where these clauses are present, as the other cases are immediate simplifications.

\begin{itemize}
\item $t\equiv \generate (\id \given c^0 \given c^1) \limit e$ 
\begin{align*}
    &\semex{\tran{\delta'}{\generate (\id \given c^0 \given c^1) \limit e}}(\gamma') \\
    &=\semex{\replicate(\tran{\delta'}{ e}, \simulate_\id(\tran{\delta'}{ c^0},\tran{\delta'}{ c^1}))}(\gamma') \quad\quad\text{def }\tranty{-}\\
    &= \begin{aligned}[t]
        &\llet m = \semex{\tran{\delta'}{ e}}(\gamma')\  \iin \quad\quad\text{def }\semex{\replicate}\\
        &\semex{\simulate_\id(\tran{\delta'}{ c^0},\tran{\delta'}{ c^1})}(\gamma')
            \bind \lambda \mu. \textstyle \bigotimes_{i=1}^m \mu
            \bind \left(\lambda (A_1, \ldots, A_m). \return \cup_{i=1}^m A_i\right)
        \end{aligned} \\
    &= \begin{aligned}[t]
        &\llet m = \sem{e}(\gamma,\delta)\  \iin \quad\quad\text{I.H.}\\
        &\sem{\id \given c^0 \given c^1}(\gamma,\delta)
            \bind \lambda \mu. \textstyle \bigotimes_{i=1}^m \mu
            \bind \left(\lambda (A_1, \ldots, A_m). \return \cup_{i=1}^m A_i\right)
        \end{aligned} \\
    &= \sem{\generate (\id \given c^0 \given c^1) \limit e}(\gamma,\delta)\quad\text{def }\sem{\generate}
\end{align*}
\item $t\equiv t_1 \where e$.
\begin{align*}
    & \semex{\tran{\delta'}{t_1 \where e}}(\gamma') & \\
    & = \big(\lambda x.\filter(\lambda r. \semex{\tran{\delta'[\id\to r]}{e}}(\gamma'),x)\big)_*\semex{\tran{\delta'}{t_1}}(\gamma') &\text{def }\semex{-},\tranty{-} \\
    & =\big(\lambda x.\filter(\lambda r.\sem{e}(\gamma,\delta[\id\to r]),x)\big)_*\sem{t_1}(\gamma,\delta) &\text{I.H.}\\
    &= \sem{t_1 \where e}(\gamma,\delta) &\text{def }\sem{-}
\end{align*}
\item $t\equiv \select e \from t_1$
\begin{align*}
&\semex{\tran{\delta'}{\select e \from t_1}}(\gamma') \\
&=  \semex{\map(\lambda r. \tran{\delta'[\id\mapsto r]}{\overline{e}}, \tran{\delta'}{ t})}(\gamma') &\text{def }\tranty{} \\
&=\begin{aligned}[t]
    &\semex{\tran{\delta}{t_1}}(\gamma')
    \bind  \\
    &\left(\lambda S. \return \left\{\lambda x'.\semex{\tran{\delta'}{e}}(\gamma'[x\mapsto x'])\ y\ \middle|\ y \in S\right\}\right) 
    \end{aligned}&\text{def }\semex{\map}\\
&= \begin{aligned}[t]
    \sem{t_1}(\gamma,\delta)
    \bind \left(\lambda S. \return \left\{\lambda x'.\sem{e}(\gamma,\delta[x\mapsto x'])\ y\ \middle|\ y \in S\right\}\right) 
    \end{aligned}&\text{I.H.}\\
&= \sem{\select e \from t_1}(\gamma,\delta) &\text{def }\sem{\select}
\end{align*}
\item $t\equiv t_1 \modjoin m$ is similar to the $\generate$ case, and is omitted for brevity.
\item $t\equiv \id: T[\id]\{\cols\}$ is immediate from the definition of $\gamma$ and of the $\id$ rules in \iql{} and the lowered language.
\item $t\equiv \id \given c^0 \given c^1$. y induction hypothesis, we have that $\sem{c^0}(\gamma,\delta)=\semex{\tran{\delta'}{c^0}}(\gamma')$, and similarly for $c^1$. Then, by inspection of the semantics of \iql{} and the lowered language, we conclude that the semantics are the same.
\item $t\equiv \probof c^i_1 \under \id \given c^0 \given c^1$. 
By induction hypothesis, we have that $\sem{c^i_1}(\gamma,\delta)=\semex{\tran{\delta'}{c^i_1}}(\gamma')$, and similarly for $c^0$ and $c^1$. 
Then, by inspection of the semantics of \iql{} and the lowered language, we conclude that the semantics are the same.
\item $t\equiv t_1 \wedge t_2, t_1 \vee t_2, \id.\col = e, \id.\col ~\op~ e, op(t_1,\ldots,t_n), \rename m \as \id', \rename t \as \id'$. These cases are straightforward and immediately follow from the induction hypothesis.
\item $t\equiv t_1 \join t_2$. This case also simply follow from the induction hypothesis, as the $\lowerjoin$ operator mimics the behavior of the$\join$operator.
\end{itemize}

\subsection{Guarantee for approximate backend}
\label{sub:appendix-approx-guarantee}

In this section we describe assumptions on approximate \tableModels{} and
the approximate semantics of the lowered language. 
We then state and prove a soundness guarantee for the lowering transform with
respect to the approximate semantics.
As discussed in \cref{sub:approximate-backend}, queries with$\where$clauses that use approximate values do not necessarily converge to the correct limit
even as the approximations converge to the correct value. 
In this section, we present modifications to the \iql{} and lowered language
semantics that detect the types of queries for which this type of
behavior can happen.
We then give a proof by logical relations for a guarantee that precisely captures the notion of correctness in this setting.

\paragraph{Key insights in the proof}
The proof of the guarantee for the approximate backend is based on the following key points:
\begin{enumerate}
    \item A safe query ($\safe$ macro defined in \cref{fig:safe-macro} returns $\true$) is one for which the approximate semantics of the lowered language converges to the exact semantics almost surely.
    \item A query will be safe if its approximations are used in a continuous way. The key theorem is the continuous mapping theorem for random variables which states that if a sequence of random variables $x_n$ converges to $x$, then for every continuous function $f$, $f(x_n)$ converges to $f(x)$.
    \item To ensure this, we define macros $\safe, \continuous, \exact$ (\cref{fig:safe-macro}) that track sufficient restrictions on queries to ensure the above. We also define appropriate topologies and metrics on the various spaces of random variables to ensure that the continuous mapping theorem holds.
\end{enumerate}

\paragraph{Topological preliminaries}
To apply the continuous mapping theorem, we need to define topologies on the different spaces on which the random variables take values.
For the discrete spaces $\BB$, $\ZZ$, $\NN$, and $\str$, we use the discrete topology.
For the real spaces $\RR$ and $\RR^+$, we use the usual metric topology.
As the spaces also interpret $\Null$, we use the discrete topology on $\Null$, and the coproduct topology for interpreting base types.
For the space of bags, we use the topology of symmetric products \citep{hatcher2002algebraic}.
For the product spaces, we use the product topology.
All these spaces are metrizable. 
This is known to be sufficient for the continuous mapping theorem to hold \citep{billingsley2017probability}[Chapter 5, Section 26].
More precisely, we take the usual metric on base types, extend each of these metrics $d$ to $\Null$ by defining $d(\Null, \Null) = 0$, and $d(\Null, x) = 1$ for $x \neq \Null$.
We then take the product metric on product spaces, and the metric on bags as the symmetric product metric. 
The space of bags with topology of symmetric products on a metric space is known to be metrizable \citep{hatcher2002algebraic}.
With these topologies, scalar functions on continuous types are continuous if they're continuous in the usual way.
All operations on discrete types are continuous.
The extended scalar operations $\op_s$ are continuous if the corresponding scalar operation $\op$ is continuous (suffices to check that the preimage of the open $\{\Null\}$ is the open $\{\Null\}$).
Likewise, the extended scalar operations $\op_l$ are continuous if the corresponding scalar operation $\op$ is continuous (suffices to check that preimage of the open $\{\Null\}$ is the open $\{\}$).
Projections on the $i$-th component $\id.\col_i$ are continuous.
It remains to show that the operations on bags are continuous. 
Using the universal property of the bag construction as a colimit in the category of topological spaces and continuous maps, we can show that the bag construction is an endofunctor on the category of topological spaces and continuous maps. This directly implies that $\map$ is continuous. 
$\singleton$ is one of the cocone injection and is therefore continuous by construction. 
$\mapreduce$ is the composition of a map and a reduce operation, and is therefore continuous if union is continuous.
Likewise, for every $n$, $\replicate~n$ is simply an $n$-fold union of the input and is continuous if union is continuous.
To show that union is continuous, first note that bag is the countable coproduct of the bags of size $n$ for all $n$. 
As finite products distribute over coproducts in the category of topological spaces and continuous maps, we will have a map $\bag (X) \times \bag (X) \to \bag (X)$ if we have a map $\bag_n(X) \times \bag_m (X) \to \bag_{n+m} (X)$ by the universal property of colimits, and where $\bag_n (X)$ are bags of size $n$ of elements of $X$.
In addition, there's a natural isomorphism between $\bag_n (X) \times \bag_m (X)$ and $\bag_{n+m} (X)$, coming from the isomorphism $X^n \times X^m \simeq X^{n+m}$, which therefore extends to a map from $\bag_n (X) \times \bag_m (X) \to \bag X$, and we are done.

The map $\splat$ defined by
\begin{align*}
\splat: (X_1 \times\ldots\times X_n)\times(Y_1 \times\ldots\times Y_m) &\to (X_1 \times \ldots \times X_n \times Y_1\times \ldots \times Y_m) \\
(x_1, \ldots, x_n), (y_1, \ldots, y_m) &\mapsto (x_1, \ldots, x_n, y_1,
\ldots, y_m)
\end{align*}
is used to give the semantics of join. This map is a composition of the product
of the identity maps $\mathrm{id}_{X_1 \times\cdots\times X_n}$ and
$\mathrm{id}_{Y_1 \times\cdots\times Y_m}$ and the
isomorphism 
$$
(X_1 \times\cdots\times X_n)\times(Y_1 \times\cdots\times Y_m) \simeq
X_1 \times\cdots\times X_n\times Y_1 \times\cdots\times Y_m.
$$
Hence, it is continuous. 
More generally, given natural numbers $n,m$, there is a function $X^n\times Y^m \to (X\times Y)^{n\times m}$ which returns all possible combinations of the elements of $X$ and $Y$ in the $n\times m$-tuple. 
This function is continuous as a composition of identity, duplication, and permutation maps.
Post-composing with the injection $(X\times Y)^{n\times m} \to \bag(X\times Y)$, we get a continuous map $X^n\times Y^m \to \bag(X\times Y)$ which respects the symmetric quotient and therefore extends to a map $\bag_n(X)\times \bag_m(Y) \to \bag(X\times Y)$. 
We conclude as in the case for union that this extends to the continuous map $\lowerjoin:\bag(X)\times \bag(Y) \to \bag(X\times Y)$.

\paragraph{Static analysis for detecting safe terms \cref{fig:safe-macro}}
The $\safe$ macro detects safe terms for which the correctness guarantee for
the approximate \ami{} holds. 
The $\exact$ macro detects those scalar terms
whose value does not depend on the approximation used to implement the
\ami{} method. 
As we formalize below, for such terms the exact and approximate
semantics should be almost surely the same. 
The $\continuous$ macro detects those scalar terms that are continuous functions of the results of $\prob$ or $\logpdf$ methods. 

\begin{figure}[t]
    \captionsetup{skip=4pt}
\footnotesize
\setlength{\abovedisplayskip}{0pt}
\setlength{\belowdisplayskip}{0pt}
\captionsetup[subfigure]{skip=0pt,belowskip=1pt,aboveskip=1pt}

\setlength{\FrameSep}{1pt}
\begin{framed}
\begin{subfigure}[t]{\linewidth}
    \caption{\textbf{Safeness of Scalar Expressions}}
    \begin{align*}
    \exact(c) &= \continuous(c) = \safe(c) = \true \\
    \exact(\op(e_1, \ldots, e_n)) &=
        \exact(e_1)\land\ldots\land\exact(e_n) \\
    \safe(\op(e_1, \ldots, e_n)) &= \safe(e_1)\land\ldots\land\safe(e_n) \\
    \continuous(\op(e_1, \ldots, e_n)) &=
        \continuous(\op)\land \bigwedge_{1\leq i\leq n}\continuous(e_i)\\
    \exact(\id.\col) &= \false \\
    \continuous(\id.\col) &= \safe(\id.\col) = \true \\
    \exact(\probof c \under \id \given c^0 \given c^1) &= \false \\
    \safe(\probof c \under \id \given c^0 \given c^1) &= \safe(c) \land 
\safe(c^0) \land \safe(c^1) \\
    \continuous(\probof c \under \id \given c^0 \given c^1) &= \safe(c) \land
\safe(c^0) \land \safe(c^1)
    \end{align*}
\end{subfigure}

\begin{subfigure}[t]{\linewidth}
    \caption{\textbf{Safeness of Event and Event-0 Expressions}}
    \begin{align*}
    \safe(\id.\col_i ~\op~ e) &= \exact(e) \\
    \safe(c ~\op~ c') &= \safe(c) \land \safe(c') \\
    \end{align*}
\end{subfigure}

\begin{subfigure}[t]{\linewidth}
\caption{\textbf{Safeness of Table Expressions}}
\centering
\begin{align*}
    \safe(\id) &= \true \\ 
\safe(\rename t \as \id) &= \safe(t) \\
\safe(t_1 \join t_2) &= \safe(t_1) \land \safe(t_2) \\
\safe(t \modjoin \id \given c^0 \given c^1) &= \safe(t) \land \safe(c^0) \land
\safe(c^1) \\
\safe(t \where e) &= \false \\
\safe(\select \overline{e} \as \overline{\col} \from t) &=
\safe(t)\land\continuous(e_1)\land\ldots\land\continuous(e_n)
\end{align*}
\end{subfigure}
\end{framed}

\caption{Macros for detecting safe terms for the approximate \ami{} guarantee}
\label{fig:safe-macro}
\end{figure}

\paragraph{Random variable semantics}  
The random variable semantics for the approximate case is given in
\cref{fig:lower-approximate-semantics}. 
If $(\mathcal{X}, \Sigma)$ is a measurable space, then by $\stseq(\mathcal{X})$ we denote the set of sequences of $\mathcal{X}$-valued random variables.

\begin{figure}
    \captionsetup{skip=4pt}
\footnotesize

\setlength{\abovedisplayskip}{-5pt}
\setlength{\belowdisplayskip}{0pt}
\captionsetup[subfigure]{skip=0pt,belowskip=0pt,aboveskip=0pt}
\begin{framed}
\begin{subfigure}[t]{\linewidth}
\caption{\bfseries Semantics of Types and Contexts}
\begin{align*}
\semap{\Gamma, x:\tau} &= \semap{\Gamma} \times \semap{\tau}
    &\semap{\sigma_g} &= \stseq(\sem{\sigma_g}) \\
\semap{\bag[\sigma_g]} &= \stseq(\measm\bag\ \sem{\sigma_g})
    &\semap{M[\sigma_g]} &= \aadmissible \sem{\sigma_g}
\end{align*}

\bigskip

\caption{\bfseries Semantics of Terms}
\begin{align*}
\semap{(t_1,\ldots, t_n)}(\gamma, \omega)_n &=
    (\semap{t_1}(\gamma, \omega)_n,\ldots, \semap{t_n}(\gamma, \omega)_n) \qquad \semap{c}(\gamma, \omega)_n =c \\
\semap{\pi_i\ t}(\gamma, \omega)_n &= \pi_i(\semap{t}(\gamma, \omega)_n) \qquad \semap{x}(\gamma, \omega)_n = \gamma(x)(\omega)_n \\
\semap{\op(t_1, t_2)}(\gamma, \omega)_n &= \op(\semap{t_1}(\gamma, \omega)_n, \semap{t_2}(\gamma, \omega)_n) \\
\semap{\mapreduce (x.t_1)\ t_2}(\gamma, \omega)_n &= \begin{aligned}[t]
    l \big[
        \semap{t_2}(\gamma,\omega)_n
        \bind \left(\lambda S.\left\{\lambda x'.\semap{t_1}(\gamma[x\mapsto x'], \omega)_n\ y\ \middle|\ y \in S\right\}\right)
    \big]
    \end{aligned}\\
\semap{\map (x.t_1)\ t_2}(\gamma, \omega)_n &= \begin{aligned}[t]
    \semap{t_2}(\gamma,\omega)_n
    \bind \left(\lambda S. \return \left\{\lambda x'.\semap{t_1}(\gamma[x\mapsto x'], \omega)_n\ y\ \middle|\ y \in S\right\}\right)
    \end{aligned}\\
\semap{\filter (x.t_1)\ t_2}(\gamma, \omega)_n &= \begin{aligned}[t]
    \semap{t_2}(\gamma,\omega)_n
    \bind \left(\lambda S.\return \left\{y \in S \middle| \lambda x'.\semap{t_1}(\gamma[x\mapsto x'], \omega)_n\ y\right\}\right)
    \end{aligned}\\
\semap{\replicate(t_1, t_2)}(\gamma, \omega)_n &= \begin{aligned}[t]
    &\llet m = \semap{t_1}(\gamma,\omega)_n\  \iin \\
    &\semap{t_2}(\gamma,\omega)_n
        \bind \lambda \mu. \textstyle \bigotimes_{i=1}^m \mu
        \bind \left(\lambda (A_1, \ldots, A_m). \return \cup_{i=1}^m A_i\right)
    \end{aligned} \\
\semap{\hostexp(t_1)}(\gamma, \omega)_n &= \exp\left(\semap{t_1}(\gamma,\omega)_n\right) \\
\semap{\singleton(t_1)}(\gamma, \omega)_n &= \return \left\{\semap{t_1}(\gamma,\omega)_n\right\} \\
\semap{\lowerjoin(t_1, t_2)}(\gamma, \omega)_n &= \begin{aligned}[t]
    &\semap{t_1}(\gamma, \omega)_n \otimes \semap{t_2}(\gamma, \omega)_n
    \bind \left(\lambda S, S'. \return (\maptwo \splat S\ S')\right)
    \end{aligned}\\
     \semap{(\id,i) ~\op~ t: \mathcal{C}^1[\sigma_g]}(\gamma, \omega)_n &= \{(x_1,\ldots,x_n) \in \sem{\sigma_g} ~\mid~ x_i ~\op_l~ \semap{t}(\gamma, \omega)_n\}\\
     \semap{(\id,i) = t :\mathcal{C}^0[\sigma_g]}(\gamma, \omega)_n &= (\pi_i,\semap{t}(\gamma, \omega)_n)\\
     \semap{c^1_1\wedge c^1_2 }(\gamma, \omega)_n &= \semap{c^1_1}(\gamma, \omega)_n \cap \semap{c^1_2}(\gamma, \omega)_n\\
     \semap{c^1_1\vee c^1_2}(\gamma, \omega)_n &= \semap{c^1_1}(\gamma, \omega)_n \cup \semap{c^1_2}(\gamma, \omega)_n\\
     \semap{c^0_1\wedge c^0_2 }(\gamma, \omega)_n &= \llet_{1\leq i\leq 2}(f_i,v_i) = \semap{c^0_i}(\gamma, \omega)_n\ \iin\ \big(\lambda x.(f_1(x),f_2(x)),(v_1,v_2)\big) \\
     \semap{\simulate_\id(c^0,c^1)}(\gamma, \omega) &=
   \left\{\mu^n_{\id; \semap{c^0}(\gamma)_n, \semap{c^1}(\gamma)_n}(\omega)\right\}\\
\semap{\logpdf_\id(c^0,c^1,c^0_2)}(\gamma, \omega) &=
   \left\{L^n_{\id;\semap{c^0}(\gamma)_n, \semap{c^1}(\gamma)_n,
    \semap{c^0_2}(\gamma)_n}(\omega)\right\} \\
\semap{\prob_\id(c^1_2, c^0,c^1)}(\gamma, \omega) &=
     \left\{P^n_{\id;\semap{c^0}(\gamma)_n, \semap{c^1}(\gamma)_n, \semap{c_2^1}(\gamma)_n}(\omega)\right\}
\end{align*}
\end{subfigure}
\end{framed}
\caption{Approximate semantics of the lowered language.}
\label{fig:lower-approximate-semantics}
\end{figure}

For the case of the \ami{} methods, we assume that the sequences of random
variables denoting each method converge to the correct value in the following
sense. 

\begin{definition}[Asymptotically Sound Approximate \ami{} Implementation]
We say that an implementation of the \ami{} is asymptotically sound, if for all
environments $\Gamma$, terms $\Gamma \vdash c^i_j: C^i[(\sigma_1, \ldots, \sigma_n)]$, and all $\gamma \in \sem{\Gamma}$, we have
\begin{align*}
\semap{\logpdf_\id(c^0_1, c^1_1, c^0_2)}(\gamma)_n &\to
\semex{\logpdf_\id(c^0_1, c^1_1, c^0_2)}(\gamma) \\
\semap{\prob_\id(c^0_1, c^1_1, c^1_2)}(\gamma)_n &\to
\semex{\prob_\id(c^0_1, c^1_1, c^1_2)}(\gamma) \\
\semap{\simulate_\id(c^0_1, c^1_1)}(\gamma)_n &\to
\semex{\simulate_\id(c^0_1, c^1_1)}(\gamma),
\end{align*}
$\mathbb{P}$-almost surely.
\end{definition}
The implementation of such asymptotically sound estimators in presence of
nested conditioning is considered in the Bayesian inference literature
\citep{lew2023probabilistic,rainforth2018nesting}.
We denote by $\aadmissible X$ the set of asymptotically sound approximate \ami{} implementation where $\id$ is a model on $X$. 
Given an evaluation context $\gamma$, we denote by $\gamma_n$ the following mapping. If $\gamma = \{\id_1 \mapsto a^1, \ldots, \id_m \mapsto a^m\}$, then $\gamma_n := \{\id_1 \mapsto a^1_n, \ldots, \id_m \mapsto a^m_n \ldots\}$, where $a^i_n$ is $a^i$ if $\id_i$ is a table type, and the $n$-th element of the sequence $a^i$ if $\id_i$ is a model type.
We can then prove the following lemma by simple induction on the structure of the term.
\begin{lemma}
If a term $\Gamma;\Delta\vdash t:\tau $ is exact then $\mathcal{P}$-almost surely
$\semap{\tran{\delta}{t}}(\gamma)_n =
\semex{\tran{\delta}{t}}(\gamma_n)$ for all $n$ and all evaluation contexts $\gamma$, $\delta$.
\label{lem:sa-safe-exact}
\end{lemma}

The proof of soundness of the translation $\tranty{\cdot}$ relies on logical
relations given in \cref{fig:logical-relations}.
The logical relations are in terms of convergence of sequences taking values in
$\sem{\tau}$ for different types $\tau$. 
We assume $\BB$, $\ZZ$, $\NN$, and $\str$ have the
discrete topology and $\RR$ and $\RR^+$ have their usual metric topology.
Product spaces are endowed with the usual product topology and the space of
bags is endowed with the topology of symmetric products \citep{hatcher2002algebraic}.
The fundamental lemma of logical relations, which implies the soundness of our translation, is as follows.

\begin{figure}

\begin{framed}
\small
\centering
\textbf{Logical Relations}
\begin{align*}
(x, \{X_n\}) \in R_\sigma & \iff _n \to x \ \ \text{a.s.} \\
((f, v), (g, \{X_n\})) \in R_{C^0[\cols]} &\iff
f = g \land (v = X_n\ \text{almost surely for all $n \in \NN$}) \\
(E, \{E_n\}) \in R_{C^1[\cols]} &\iff
\text{$E = E_n$ almost surely for all $n \in \NN$} \\
(\mu, \{\mu_n\}) \in R_{\tabtype} & \iff 
\mu_n(\omega) \Rightarrow \mu\ \text{for $\mathbb{P}$-almost all $\omega$}
\end{align*}
\end{framed}
\caption{Logical relations used for the proof of soundness of $\tranty{\cdot}$
with respect to the approximate semantics.}
\label{fig:logical-relations}
\end{figure}

\begin{lemma}[Fundamental lemma of logical relations]
Suppose
$$
\Gamma := x_1:\tau_1, \ldots, x_n:\tau_n
\qquad\qquad
\Delta := x_1':\tau_1', \ldots, x_m':\tau_m'
$$
and $\Gamma; \Delta \vdash t:\tau$ where $t$ is a normalized \iql{} term which is safe (continuous for a scalar type, safe for a table, event or event-0 type). 
Take any $a_i \in \sem{\tau_i}$, $a_i' \in \sem{\tau_i'}$, $b_i \in
\semap{\tranty{\tau_i}}$, and $b_i' \in \semap{\tranty{\tau_i'}}$, such that
$(a_i, b_i) \in R_{\tau_i}$ and $(a_i', b_i') \in R_{\tau_i'}$.  
Moreover, suppose that the implementation of the \ami{} methods are asymptotically sound.
Then, for all $\delta_{\mathcal{T}}$ of the form
$$
\delta_{\mathcal{T}} = \{x_i' \mapsto r_i | 1 \leq i \leq m\}
$$
where the $r_i$s are variable names in the lowered language, if we let
$$
\gamma := \{x_i \mapsto a_i | 1\leq i\leq n\}
\qquad
\delta := \{x_i' \mapsto a_i' | 1\leq i\leq m\}
\qquad
\delta' := \{ r_i \mapsto b_i' | 1\leq i\leq m\}
$$
we will have
$$
(\sem{t}(\gamma, \delta),
\semap{\tran{\delta_{\mathcal{T}}}{t}}(\tranty{\gamma}\cup \delta'))
\in R_{\tau}.
$$
\end{lemma}
\begin{proof}[proof sketch]
We prove the lemma case by case for normalized \iql{} queries.

\textbf{Scalar Expressions}
\begin{description}
\item[constants] If $t$ is a constant of type $\sigma$, then then
$\tran{\delta_\mathcal{T}}{t} = t$ by the lowering rules and so
$$
\sem{t}(\gamma, \delta) = t = \semap{t}(\tranty{\gamma}) = 
\semap{\tran{\delta_\mathcal{T}}{t}}.
$$
As constant sequences are almost surely convergent to their value, we can
conclude that $(\sem{t}(\gamma, \delta),
\semap{\tran{\delta_\mathcal{T}}{t}}(\tranty{\gamma}\cup\delta')) \in
R_\sigma$.

\item[column values] If $t$ is of the form $\id.\col$, then by the typing rules
of \iql{} the local context $\Delta$ should contain $T[\id]\{\cols\}$ as its
last entry. Hence, either by our assumption on the context, or by the induction hypothesis, $\id$ is a safe term. As projection is a continuous operation, we have that 
\begin{align*}
&\semap{\tran{\delta_\mathcal{T}}{\id.\col_i}}(\tranty{\gamma})_n \\
&= \pi_i(\semap{\tran{\delta_\mathcal{T}}{\id}}(\tranty{\gamma})(\tranty{\gamma}))_n \\
&\toas \pi_i(\sem{\id}(\gamma, \delta)) \\
&= \sem{\id.\col_i}(\gamma, \delta)
\end{align*}

\item[mathematical operations] If $t$ is of the form $\op(e_1, \ldots, e_n)$
where $\Gamma;\Delta\vdash e_i:\sigma_i$ and $\continuous(\op)$, by definition of $R_{\sigma_i}$ and our assumption we must have
$\semap{\tran{\delta_\mathcal{T}}{e_i}}(\tranty{\gamma})_n \toas \sem{e_i}$.
By the lowering rules we have
$$
\semap{\tran{\delta_\mathcal{T}}{\op(e_1, \ldots, e_n)}}(\tranty{\gamma}) =
\op(\semap{\tran{\delta_\mathcal{T}}{e_1}}(\tranty{\gamma}),
\ldots,
\semap{\tran{\delta_\mathcal{T}}{e_n}}(\tranty{\gamma})),
$$
As $\op$ is continuous, we can conclude that
$$
\semap{\tran{\delta_\mathcal{T}}{\op(e_1, \ldots, e_n)}}(\tranty{\gamma})_n
\toas 
\sem{\op(e_1, \ldots, e_n)}(\gamma, \delta).
$$
Hence, by definition of $R_\sigma$ our claim holds.

\item[probability of an event] If $t$ is of the form $\probof c \under m$ with
$\Gamma, \Delta\vdash c: C^1\{\cols\}$, then by normalization assumption $t$
must necessarily be in the normal form
$$
\probof c \under \id \given c^0 \given c^1.
$$
Now,
{\small
\begin{align*}
&\semap{\tran{\delta_\mathcal{T}}{\probof c \under \id \given c^0 \given c^1}}
    (\tranty{\gamma})_n \\
&\hspace{2em}= \semap{\prob\tran{\delta_\mathcal{T}}{c^0}
    \tran{\delta_\mathcal{T}}{c^1}
    \tran{\delta_\mathcal{T}}{c}}(\tranty{\gamma})_n
& \text{(lowering rules)} \\
&\hspace{2em}\toas \semex{\prob\tran{\delta_\mathcal{T}}{c^0}
    \tran{\delta_\mathcal{T}}{c^1}
    \tran{\delta_\mathcal{T}}{c}}(\tranty{\gamma})
& \text{(sound \ami{} implementation)} \\
&\hspace{2em}= \semex{\tran{\delta_\mathcal{T}}{\probof c \under \id \given c^0 \given c^1}}
    (\tranty{\gamma})
& \text{(lowering rules)} \\
&\hspace{2em}= \sem{\probof c \under \id \given c^0 \given c^1}(\gamma).
& \text{(exact guarantee)} \\
\end{align*}
}
This is a base case for the induction.

\item[likelihood of an event-0] If $t$ is of the form $\probof c \under m$ with
$\Gamma, \Delta\vdash c: C^0\{\cols'\}$, then the proof is almost identical to
the case where $\Gamma, \Delta\vdash c: C^1\{\cols\}$.
\end{description}

\textbf{Event Expressions}
\begin{description}
\item[primitive events] Suppose $\Gamma; \Delta\vdash t: C^1\{\cols\}$ is of
the form $\id.\col_i\ \op\ e$ where
$$
\cols = \col_1:\sigma_1, \ldots, \col_m:\sigma_m.
$$
For ease of notation we let $\sigma_g := (\sigma_1, \ldots, \sigma_m)$. By our
safety assumption and the definition of the $\safe$ macro, the subterm $e$ must
be exact. By Lemma \ref{lem:sa-safe-exact}, we must then have
$\semap{\tran{\delta_\mathcal{T}}{e}}(\tranty{\gamma})_n = 
\semex{\tran{\delta_\mathcal{T}}{e}}(\tranty{\gamma})$ for all $n$ almost
surely. Hence, by definition of the lowering rules and the approximate
semantics
\begin{align*}
\semap{\tran{\delta_\mathcal{T}}{\id.\col_i\ \op\ e}}(\tranty{\gamma})_n &= 
\semap{(\id, i)\ \op\ \tran{\delta_\mathcal{T}}{e}}(\tranty{\gamma})_n \\
&= \{(x_1, \ldots, x_m) \in \sem{\sigma_g} |
    x_i\ \op\ \semap{\tran{\delta_\mathcal{T}}{e}}(\tranty{\gamma})_n\} \\
&= \{(x_1, \ldots, x_m) \in \sem{\sigma_g} |
    x_i\ \op\ \semex{\tran{\delta_\mathcal{T}}{e}}(\tranty{\gamma})\} \\
&= \semex{(\id, i)\ \op\ \tran{\delta_\mathcal{T}}{e}}(\tranty{\gamma}) \\
&= \semex{\tran{\delta_\mathcal{T}}{\id.\col_i\ \op\  e}}(\tranty{\gamma}),
\end{align*}
almost surely. Applying the soundness theorem for the exact \ami{}
implementations, we are done.

\item[conjunctions and disjunctions] If $\Gamma;\Delta\vdash t:C^1\{\cols\}$ is
of the form $c_1^1 \land c_2^1$, then by our inductive assumption we must have
$$
\sem{c_i^1}(\gamma, \delta) = 
\semap{\tran{\delta_\mathcal{T}}{c_i^1}}(\tranty{\gamma})_n
$$
almost surely for all $n$. Hence,
\begin{align*}
\semap{\tran{\delta_\mathcal{T}}{c_1^1\land c_2^2}}(\tranty{\gamma})_n
&= \semap{\tran{\delta_\mathcal{T}}{c_1^1}\land\tran{\delta_\mathcal{T}}{c_2^1}}
(\tranty{\gamma})_n \\
&= \semap{\tran{\delta_\mathcal{T}}{c_1^1}}(\tranty{\gamma})_n
\cap\semap{\tran{\delta_\mathcal{T}}{c_2^1}}(\tranty{\gamma})_n \\
&= \sem{c_1^1}(\gamma, \delta) \cap \sem{c_2^1}(\gamma, \delta) \\
&= \sem{c_1^1 \land c_2^1}(\gamma, \delta)
\end{align*}

The proof is almost identical when $\Gamma;\Delta\vdash t:C^1\{\cols\}$ is
of the form $c_1^1 \lor c_2^1$.
\end{description}

\textbf{Event-0 Expressions}
\begin{description}
\item[primitive event-0s] Suppose $\Gamma; \Delta\vdash t: C^0\{\col_i\}$ is of
the form $\id.\col_i=e$ where $\col_i:\sigma$. By our safety assumption and
the definition of the $\safe$ macro, the subterm $e$ must
be exact. By Lemma \ref{lem:sa-safe-exact}, we must then have
$\semap{\tran{\delta_\mathcal{T}}{e}}(\tranty{\gamma})_n = 
\semex{\tran{\delta_\mathcal{T}}{e}}(\tranty{\gamma})$ for all $n$ almost
surely. Hence, by definition of the lowering rules and the approximate
semantics
\begin{align*}
\semap{\tran{\delta_\mathcal{T}}{\id.\col_i= e}}(\tranty{\gamma})_n &= 
\semap{(\id, i)=\tran{\delta_\mathcal{T}}{e}}(\tranty{\gamma})_n \\
&= (\pi_i, \semap{\tran{\delta_\mathcal{T}}{e}}(\tranty{\gamma})_n) \\
&= (\pi_i, \semex{\tran{\delta_\mathcal{T}}{e}}(\tranty{\gamma})) \\
&= \semex{(\id, i)=\tran{\delta_\mathcal{T}}{e}}(\tranty{\gamma}) \\
&= \semex{\tran{\delta_\mathcal{T}}{\id.\col_i\ \op\  e}}(\tranty{\gamma}),
\end{align*}
almost surely. By the soundness theorem for the exact \ami{} implementations,
we are done.

\item[conjunctions] If $\Gamma;\Delta\vdash t:C^0\{\cols\}$ is
of the form $c_0^1 \land c_0^1$, then by our inductive assumption we must have
$$
\sem{c_i^0}(\gamma, \delta) = 
\semap{\tran{\delta_\mathcal{T}}{c_i^0}}(\tranty{\gamma})_n
$$
almost surely for all $n$. Hence,
\begin{align*}
\semap{\tran{\delta_\mathcal{T}}{c_1^0\land c_2^0}}(\tranty{\gamma})_n
&= \semap{\tran{\delta_\mathcal{T}}{c_1^0}\land\tran{\delta_\mathcal{T}}{c_2^0}}
(\tranty{\gamma})_n \\
&= \llet_{1\leq i\leq 2} (f_i,v_i) = \semap{
    \tran{\delta_\mathcal{T}}{c_i^0}
}(\tranty{\gamma})_n\ \iin\\
&\hspace{2em} \big(\lambda x.(f_1(x),f_2(x)),(v_1,v_2)\big) \\
&= \llet_{1\leq i\leq 2} (f_i,v_i) = \sem{c_i^0}(\gamma, \delta)\ \iin\\
&\hspace{2em} \big(\lambda x.(f_1(x),f_2(x)),(v_1,v_2)\big) \\
&= \sem{c_1^0\land c_2^0}(\gamma, \delta),
\end{align*}
almost surely.
\end{description}

\textbf{Table Expressions}
\begin{description}
\item[loaded tables] If $\Gamma;\Delta\vdash \id:T[\id]\{\cols\}$, then by assumption $\id$ is safe and there's nothing to show.

\item[renamed tables] If $t$ is of the form $\rename t' \as \id'$ then by the
typing rules of \iql{} we must have $\Gamma;\Delta\vdash t':T[?\id]\{\cols\}$,
and so by our inductive assumption
$$
\semap{\tran{\delta_\mathcal{T}}{t'}}(\tranty{\gamma})_n \toas 
\sem{t'}(\gamma, \delta).
$$
By definition of $\tran{\delta_\mathcal{T}}{-}$ we have 
\begin{align*}
\semap{\tran{\delta_\mathcal{T}}{\rename t' \as \id'}}(\tranty{\gamma})_n &=
\semap{\tran{\delta_\mathcal{T}}{t'}}(\tranty{\gamma})_n \\
&\toas \sem{t'}(\gamma, \delta) \\
&=\sem{\rename t' \as \id'}(\gamma, \delta).
\end{align*}

\item[joins] If $t$ is of the form $t_1 \join t_2$ then by the typing
rules of \iql{} we must have $\Gamma;\Delta\vdash t_i:T[?\id_i]\{\cols_i\}$,
and so by our inductive assumption
$$
\semap{\tran{\delta_\mathcal{T}}{t_i}}(\tranty{\gamma})_n \toas 
\sem{t}(\gamma, \delta).
$$
For ease of notation, we let $\sigma_g = (\sigma_1, \ldots, \sigma_n)$ and 
$\sigma_g' = (\sigma_1', \ldots, \sigma_m')$ where 
Now,
\begin{align*}
\semap{\tran{\delta_\mathcal{T}}{t_1 \join t_2}}(\tranty{\gamma})_n &=
\semap{\lowerjoin(\tran{\delta_\mathcal{T}}{t_1},
                  \tran{\delta_\mathcal{T}}{t_2})}(\tranty{\gamma})_n \\
&= \semap{t_1}(\tranty{\gamma})_n \otimes \semap{t_2}(\tranty{\gamma})_n
    \bind \\
&\qquad\left(\lambda S, S'. \return (\maptwo \splat S\ S')\right)
\end{align*}
As described in our topological preliminaries above, the mappings $\splat$, $\lowerjoin$ and $\maptwo$ are continuous.
Using the continuous mapping theorem and the basic properties of weak
convergence of measures we can pass the above to the limit and apply our
inductive hypothesis to get
\begin{align*}
\semap{\tran{\delta_\mathcal{T}}{t_1 \join t_2}}(\tranty{\gamma})_n
&\Rightarrow
\sem{t_1}(\gamma, \delta) \otimes \sem{t_2}(\gamma, \delta)
    \bind \\
&\qquad\left(\lambda S, S'. \return (\maptwo \splat S\ S')\right) \\
&= \sem{t_1 \join t_2}(\gamma, \delta).
\end{align*}

\item[filtered tables] If $t$ is of the form $t \where e$ then by definition of
the $\safe$ macro $t$ has to be unsafe. Hence, there is nothing to show in this
case.

\item[generated tables] If $t$ is of the form $\generate m \limit e$, by
the normalization assumption $t$ must be of the form $\generate \id \given c^0
\given c^1 \limit e$. By the lowering rules,
\begin{align*}
\semap{\tran{\delta_\mathcal{T}}{\generate\id \given c^0 \given c^1 \limit
e}}(\tranty{\gamma})_n &= \\
&\hspace{-20em}\semap{\replicate
    \tran{\delta_\mathcal{T}}{e},
    \simulate_\id(\tran{\delta_\mathcal{T}}{c^0},
                  \tran{\delta_\mathcal{T}}{c^1})}(\tranty{\gamma})_n
\end{align*}
Now, by definition of $\safe$ and $\exact$ and Lemma \ref{lem:sa-safe-exact}
we must have
\begin{align*}
\semap{\tran{\delta_\mathcal{T}}{c^0}}(\tranty{\gamma})_n &=
\semex{\tran{\delta_\mathcal{T}}{c^0}}(\tranty{\gamma}) \\
\semap{\tran{\delta_\mathcal{T}}{c^1}}(\tranty{\gamma})_n &=
\semex{\tran{\delta_\mathcal{T}}{c^1}}(\tranty{\gamma}) \\
\semap{\tran{\delta_\mathcal{T}}{e}}(\tranty{\gamma})_n &=
\semex{\tran{\delta_\mathcal{T}}{e}}(\tranty{\gamma}),
\end{align*}
almost surely for all $n$. 
Hence, using the definition of the approximate semantics we can see that
\begin{align*}
&\semap{\replicate
    \tran{\delta_\mathcal{T}}{e},
    \simulate_\id(\tran{\delta_\mathcal{T}}{c^0},
                  \tran{\delta_\mathcal{T}}{c^1})}(\tranty{\gamma})\\
&= \semex{\replicate
    \tran{\delta_\mathcal{T}}{e},
    \simulate_\id(\tran{\delta_\mathcal{T}}{c^0},
                  \tran{\delta_\mathcal{T}}{c^1})}(\tranty{\gamma})_n
\end{align*}
almost surely for all $n$. Applying the correctness result of the exact \ami{}
we are done.
\item[selects] By assumption, the term $e$ in the select is continuous, and by assumption the term $t$ is safe. Hence, the continuous mapping theorem directly applies as the semantics of select is the pushforward by the semantics of $e$ of the semantics of $t$.
\item[generative joins] by assumption, $\simulate_\id$ converges to the exact semantics almost surely. Singleton and pairing are continuous, as well as join and mapreduce, and so the whole translation of the generative join is continuous and converges to the exact semantics almost surely.

\end{description}

\textbf{Row Model Expressions} 
As the lowering is defined on normalized queries, all \tableModel{} terms
are of the form $\id\given c^0\given c^1$, where the$\given$clause are optional.
Therefore, these terms are never directly returned by the lowering transformation, and are only used in$\probof$and $\generate$expressions.
Hence, the proof is complete.
\end{proof}

Using the fundamental lemma, we can prove the following result on the soundness
of the lowering transformation in presence of approximate semantics.
\begin{theorem}[Consistent \ami{} Guarantee]
    \label{thm:consistent-ami-guarantee-appx}
Let $\Gamma, [] \vdash t: T[?\id]\{\cols\}$ be a safe query and
suppose the \ami{} methods have asymptotically sound implementations.
Then, for all evaluation context $\gamma$, we have 
$$
\lim_n (\semap{\tran{[]}{t}}(\tranty{\gamma})) = \sem{t}(\lim_n \gamma)
$$
$\mathbb{P}$-almost surely.
\end{theorem}
\begin{proof}
Since the \ami{} methods are assumed to have
asymptotically sound implementations and the queyr safe, by the fundamental lemma we must have, for all $\gamma$ and $\delta=[]$, that
\[
(\sem{t}(\gamma, []),
\semap{\tran{[]}{t}}(\tranty{\gamma}, [])) \in R_{T[?\id]\{\cols\}}
\]
By the second part of the definition of the logical relation for table types, we have that almost surely $\lim_n (\semap{\tran{[]}{t}}(\gamma)) = \sem{t}(\lim_n \gamma)$, as required.
\end{proof}

We finally note that we gave a sufficient but not necessary condition for the soundness of the approximate semantics of the lowered language, and leave the development of a more permissive safe macro for future work.

\section{Full language}
\label{sec:full-language}

Our implementation of \iql{} contains several more primitives and
operations than the core language described in the main text.
Here, we present a fuller language more representative of the
implementation.
Yet, we still omit several constructs such as SQL's LEFT JOIN, RIGHT
JOIN, and OUTER JOIN which can be expressed with the primitives we
present.
These constructs are omitted for brevity and because they are not
essential to the core functionality of the language, but they can of
course be convenient and have optimized implementations in practice.

We assume given a set of aggregates $A\in \mathcal{A}$.
For aggregates $A$, we write $A:\sigma_1\to\sigma_2$ to indicate that
the aggregate operates on elements on type $\sigma_1$ and returns
elements of type $\sigma_2$.
\Cref{fig:aggr-type} shows a list of common aggregates and their
types.$\dedup$removes duplicates rows in a table, while$\duplicate$creates copies of each row in a table.
$\with t_1 \as \id': t_2$ is a convenient notation for introducing a
new table name $\id'$ bound to the value of $t_1$ in the context of a
table expression $t_2$.
It is a SQL equivalent of a let binding.
It can be used jointly with$\modjoin$to generate multiple possible
completions of a row under a model.
$\mi$computes the conditional mutual information between two sets of
columns in a table, conditioned on an event or event-0 $c^i$ under a
model $m$.
In general, event if marginal densities can be computed exactly for a
model $m$, the (conditional) mutual information may not be computable
in closed form, and is typically approximated using Monte Carlo
methods~\citep{saad2017dependencies}.
We can for instance use$\mi$as a compact way to express the
computation presented in \cref{fig:example-cmi-section2}.

The example in \cref{fig:example-cmi-section2} uses a syntax closer to SQL in its use of$\group\by$. It is possible to express the same query using the syntax in \cref{fig:syntax-iql-full} by using its$\group\by$operator, as follows.
Let $t$ be the query from lines 2-17 in \cref{fig:example-cmi-section2}. Then the query in \cref{fig:example-cmi-section2} can be expressed as
$\group t \by [\text{table.weight} \as \text{weight}] \aggr \Avg(\text{table.log\_pxy\_div\_px\_py}) \as \text{mutual\_info}$.

\begin{figure}
    \centering
\begin{tabular}{|c|c|}
    \hline
    Aggregate & Types \\
     $\Sum$ & $\sigma_c\to\real$, $\inte\to\inte$, $\nat\to\nat$ \\
      $\Avg$ & $\inte\to\real$, $\nat\to\real$,  $\sigma_c\to\sigma_c$  \\
      $\Max, \Min$ & $\sigma_c\to\sigma_c,\inte\to\inte$,  $\nat\to\nat$, $\bool\to\bool$ \\
      $\Count,\Countdis$ & $\sigma\to\nat$ \\
      $\Concat$ & $\str\to\str$ \\ \hline
\end{tabular}
    \caption{Supported types of aggregate operators.}
    \label{fig:aggr-type}
\end{figure}

\begin{figure}
    \fbox{
        \parbox{\textwidth}{
\begin{align*}
    \text{Continuous base types}\ \sigma_c &::= \real \gor \posreal \gor \ranged(a,b) \\
    \text{Discrete base types}\ \sigma_d   &::= \inte \gor \str \gor \nat\gor  \bool \gor \categorical(\textsc{n}_1,\ldots,\textsc{n}_k) \\
    \text{Base Types} \ \sigma &::= \sigma_c \gor \sigma_d \\
    \text{Table Types} \ \tabtype &::= T[?\id]\tabty{\col_1: \sigma_1,\ldots,\col_n: \sigma_n} \\
    \text{\tableModel{} Types}\ \modtype &::=   M[?\id]\tabty{\col_1: \sigma_1,\ldots,\col_n: \sigma_n}\\
    \text{Event Types}\ \evtype &::= \eventtype\tabty{\col_1: \sigma_1,\ldots,\col_n: \sigma_n} \gor \densitytype\tabty{\col_1: \sigma_1,\ldots,\col_n: \sigma_n} \\
    \text{Aggregates} \ A &::= \Sum \gor \Avg \gor \Max \gor \Min \gor \Count \gor \Concat \gor \Countdis\\
    \text{Operations}\ \op &::= + \gor * \gor \div \gor \log \gor \exp \gor \sqrt{\cdot} \gor\ >\ \gor\ <\ \gor\ =\ \gor \ldots \\
    \text{Table Expressions} \ t &::= \id
    \gor t_1\union t_2  \gor t_1\ \join\ t_2 \gor \rename t \as \id  \\
    &\quad \gor \dedup t \gor t \duplicate e \timess  \gor t\ \where\ e \gor \with t_1 \as \id' : t_2\\
    \oldcons{&\quad\gor \select e_1 \as \col_1,\ldots,e_n \as \col_n \from t\\}
    \oldcons{&\quad\gor \group t \by [e_1 \as \col_1,\ldots,e_n\as \col_n] \\
    &\quad\quad\aggr A_1(e'_1)\as\col'_1,\ldots,A_m(e'_m)\as\col'_m\\}
    &\quad\gor \generate m \limit e \gor t \modjoin m \given c^i\\
    \newcons{& \quad\gor \select e_1 \as \col_1,\ldots,e_n \as \col_n \from \ttable t, \model m\\}
    \text{\tableModel{} Expressions}\ m &::= \id \gor m \given c^i
    \gor \rename m \as \id\\
     \text{Scalar Expressions}\ e &::=  \id.\col \gor  \op(e_1,\ldots,e_n) \gor\probof c^i \under m \\
    &\quad \gor \mi (\id.\cols, \id.\cols', c^i) \under m \\
     \text{Event Expressions}\ c^1&::= \bigwedge_{1\leq i\leq 2} c^1_i \gor  \bigvee_{1\leq i\leq 2} c^1_i \gor \id.\col\ \op\ e \\
      \text{Event-0 Expressions}\ c^0 &::=  \bigwedge_{1\leq i\leq 2} c^0_i \gor \id.\col = e
\end{align*}
        }}
    \caption{Full syntax of \iql{}.}
    \label{fig:syntax-iql-full}
\end{figure}

\begin{figure}
    \fbox{
        \parbox{\textwidth}{
        \centering
    \begin{tabular}{c}
        $\ctx \vdash t:T[?\id]\{\cols\}$
        \\ \hline
        $\ctx \vdash \dedup t:T[?\id]\{\cols\}$
     \end{tabular}
     \quad
     \begin{tabular}{c}
        $\ctx \vdash t_1:T[\id_1]\{\cols\}$ \quad  $\ctx \vdash t_2:T[\id_2]\{\cols\}$
      \\ \hline
      $\ctx \vdash t_1 \union t_2:T[]\{\cols\}$
\end{tabular}
 \medskip

 \begin{tabular}{c}
    $\ctx \vdash t:T[?\id]\{\cols\}$    \quad
    $\ctx \vdash e:\nat$
    \\ \hline
    $\ctx \vdash t \duplicate e \timess :T[?\id]\{\cols\}$
 \end{tabular}
 \medskip

 \begin{tabular}{c}
    $\ctx \vdash t_1:T[?\id]\{\cols\}$ \quad
    $\Gamma, T[\id'']\{\cols\}; \Delta \vdash t_2:T[?\id']\{\cols\}$
    \quad $\id''$ fresh
    \\ \hline
    $\ctx \vdash \with t_1 \as \id'' : t_2 :T[?\id']\{\cols\}$
 \end{tabular}
   \medskip

   \oldcons{\begin{tabular}{c}
    $\ctx \vdash t:T[?\id]\{\cols\}$ \quad
    $\ctx, T[?\id]\{\cols\}\vdash e_i:\sigma_i$ for $1\leq i\leq n$\\
    $\ctx, T[?\id]\{\cols\}\vdash e'_j:\sigma'_j$ \quad $A_j:\sigma_j'\to\sigma_j''$ for $1\leq j\leq m$  \\
    $\widebar{e} \as \widebar{\col} := [e_1 \as \col_1,\ldots,e_n\as \col_n] $ \quad $\widebar{\col}\cap \widebar{\col'}=\emptyset$ \\
    $\widebar{A(e')} \as \widebar{\col'}:= A_1(e'_1)\as\col'_1,\ldots,A_m(e'_m)\as\col'_m$
    \\ \hline
    $\ctx \vdash \group t \by \widebar{e} \as \widebar{\col} \aggr \widebar{A(e')} \as \widebar{\col'}$ \\
    $:T[]\{\col_1:\sigma_1,\ldots,\col_n:\sigma_n,\col'_1:\sigma''_1,\ldots,\col'_m:\sigma''_m\}$
 \end{tabular}
 \medskip}

 \begin{tabular}{c}
        $\ctx \vdash m:M[?\id]\{\cols\}$ \quad
        $\ctx,  M[?\id]\{\cols\} \vdash c^i: \evtype $\quad $\cols', \cols''\subseteq \cols$
       \\ \hline
       $\ctx \vdash \mi (\id.\cols', \ \id.\cols'', c^i)\under m :\real$
    \end{tabular}
        }}
    \caption{Type system for the \iql{} expressions omitted in the main text.}
    \label{fig:type-system-iql-full}
\end{figure}

\begin{figure}
    \fbox{
        \parbox{\textwidth}{
    \begin{align*}
        \sem{\dedup t}(\gamma,\delta) &= \Big(\lambda x.\fold \big(\lambda r,y. \{r\} \cup \filter (\lambda r'. r=r')\ y\big)\ \emptyset \ x\Big)_*
        \sem{t}(\gamma,\delta) \\
        \sem{t_1 \union t_2} (\gamma,\delta) &=
        (\lambda x,y.\ x\cup y)_*
        (\sem{t_1}(\gamma,\delta) \otimes  \sem{t_2}(\gamma,\delta))\\
        \sem{t \duplicate e \timess} (\gamma,\delta) &=
         \llet n=\sem{e}(\gamma,\delta)\ \iin
         \Big(\lambda y.\bigcup_{1\leq i\leq n} y\Big)_*
         \sem{t}(\gamma,\delta) \\
         \sem{\with t_1 \as \id'' : t_2} (\gamma,\delta) &= \sem{t_1}(\gamma,\delta) \bind (\lambda x. \sem{t_2}(\gamma[\id''\mapsto x], \delta))
    \end{align*}
    \begin{align*}
        &\sem{ \group t:T[?\id]\{\cols\} \by \widebar{e} \as \widebar{\col} \aggr \widebar{A(e')} \as \widebar{\col'}} (\gamma,\delta) = \\
        &\Big(\lambda x.\llet y= \map (\lambda r.\big(\sem{\widebar{e}}(\gamma, \delta[\id\mapsto r]),\sem{\widebar{e'}}(\gamma, \delta[\id\mapsto r])\big))\ x \\
        &\quad \llet keys = \map (\lambda (a,b). b)\ y \\
        &\quad \llet bags = \map (\lambda a. (a,\fold (\lambda (a',b'). \iif a=a'\tthen \{b'\}\ \eelse \{\})\ \{\}\ y))\ keys \\
        &\quad \map (\lambda (a,s). (a,\widebar{A}(s)))\ bags \Big)_* \sem{t}(\gamma,\delta)
        \end{align*}
        where $\sem{\widebar{e}}(\gamma, \delta) = (\sem{e_1}(\gamma, \delta),\ldots,\sem{e_n}(\gamma, \delta))$ and
         $\sem{\widebar{e'}}(\gamma, \delta) = (\sem{e_1'}(\gamma, \delta),\ldots,\sem{e_m'}(\gamma, \delta))$
        }}
    \caption{Denotational semantics for the \iql{} expressions omitted in the main text.}
    \label{fig:semantics-iql-full}
\end{figure}

\section{List of \iql{} queries}
\label{appx:list-of-queries}

\Cref{fig:examples} showcases the modelling capabilities of \iql{} on
a variety of data analysis examples. \Cref{benchmark-queries-1} shows
the queries used from the benchmark in \cref{tab:runtime}.

\begin{figure}[!t]
    \captionsetup[subfigure]{skip=0pt,belowskip=0pt,aboveskip=0pt}
\captionsetup[sub]{skip=2pt}

\begin{subfigure}[b]{.48\linewidth}
    \centering
        \input{figures/iql-code/mi-est}
\caption{Entropy on conditioned model}
\label{fig:entropy}
\end{subfigure}
\hfill
\begin{subfigure}[b]{.48\linewidth}
    \centering
        \input{figures/iql-code/kl-est}
\caption{Kullback-Leibler divergence}
\label{fig:kl}
\end{subfigure}

\medskip

\begin{subfigure}[b]{.48\linewidth}
    \centering
        \input{figures/iql-code/nested-query}
\caption{Conditioning a model on output of another}
\label{fig:cond-on-cond}
\end{subfigure}
\hfill
\begin{subfigure}[b]{.48\linewidth}
    \centering
        \input{figures/iql-code/map-est}
\caption{Maxiumum-a-posteriori}
\label{fig:map}
\end{subfigure}

\medskip

\begin{subfigure}[b]{.48\linewidth}
    \centering
        \input{figures/iql-code/anom-detec}
\caption{Anomaly detection}
\label{fig:anom-detection}
\end{subfigure}
\hfill
\begin{subfigure}[b]{.48\linewidth}
    \centering
        \input{figures/iql-code/prediction-new}
\caption{Prediction}
\label{fig:preiction}
\end{subfigure}

\medskip

\begin{subfigure}[b]{.48\linewidth}
    \centering
        \input{figures/iql-code/likely-synthetic-data}
\caption{Likely synthetic data}
\label{fig:likely-data}
\end{subfigure}
\hfill
\begin{subfigure}[b]{.48\linewidth}
    \centering
        \input{figures/iql-code/cond-synth-gen}
\caption{Conditional synthetic data generation}
\label{fig:cond-synthesis-gen}
\end{subfigure}

\medskip

\begin{subfigure}[b]{.48\linewidth}
    \centering
        \input{figures/iql-code/imputation}
\caption{Imputation}
\label{fig:imputation}
\end{subfigure}
\hfill
\begin{subfigure}[b]{.48\linewidth}
    \centering
        \input{figures/iql-code/cmi-est}
\caption{Conditional mutual information}
\label{fig:cmi}
\end{subfigure}

\caption{Example queries in \iql{} for a variety of data analysis tasks.}
\label{fig:examples}
\end{figure}

\begin{figure}
            \centering
\begin{tabular}{cc}
    logpdf 1 &  
    \begin{minipage}{0.8\textwidth}
        \centering
    \begin{lstlisting}[style=iql-small]
        SELECT 
          PROBABILITY OF Period_minutes = 98.6 
          UNDER model GIVEN Country_of_Operator
        FROM data
\end{lstlisting}
    \end{minipage} \\ 
    logpdf 2 & 
    \begin{minipage}{0.8\textwidth}
        \centering
    \begin{lstlisting}[style=iql-small]
        SELECT
            PROBABILITY OF Period_minutes = 98.6 AND
                        Type_of_Orbit = \"Sun-Synchronous\"
            UNDER model GIVEN Country_of_Operator AND
                            Launch_Mass_kg
        FROM data
\end{lstlisting}
    \end{minipage} \\  
    logpdf 3 & 
    \begin{minipage}{0.8\textwidth}
        \centering
    \begin{lstlisting}[style=iql-small]
        SELECT
            PROBABILITY OF Period_minutes = 98.6 AND
                        Type_of_Orbit = \"Sun-Synchronous\" AND
                        Contractor = \"Lockheed Martin\"
            UNDER model GIVEN Country_of_Operator AND
                            Launch_Mass_kg AND
                            Inclination_radians
        FROM data
\end{lstlisting}
    \end{minipage} \\  
    logpdf 4 & 
    \begin{minipage}{0.8\textwidth}
        \centering
    \begin{lstlisting}[style=iql-small]
        SELECT
            PROBABILITY OF Period_minutes = 98.6 AND
                        Type_of_Orbit = \"Sun-Synchronous\" AND
                        Contractor = \"Lockheed Martin\" AND
                        Eccentricity = 0.001
            UNDER model GIVEN Country_of_Operator AND
                            Launch_Mass_kg AND
                            Inclination_radians AND
                            Apogee_km
        FROM data
\end{lstlisting}
    \end{minipage} \\  
    logpdf 5 &
    \begin{minipage}{0.8\textwidth}
        \centering
    \begin{lstlisting}[style=iql-small]
        SELECT
            PROBABILITY OF Period_minutes = 98.6 AND
                        Type_of_Orbit = \"Sun-Synchronous\" AND
                        Contractor = \"Lockheed Martin\" AND
                        Eccentricity = 0.001 AND
                        Purpose = \"Communications\"
            UNDER model GIVEN Country_of_Operator AND
                            Launch_Mass_kg AND 
                            Inclination_radians AND
                            Apogee_km AND
                            Power_watts
        FROM data
\end{lstlisting}
    \end{minipage} \\  
\end{tabular}
\caption{Queries from \cref{tab:runtime} in \iql{}.}
\label{benchmark-queries-1}
\end{figure}

\begin{figure}\ContinuedFloat
            \centering
    \begin{tabular}{cc}
    logpdf 6 & 
    \begin{minipage}{0.8\textwidth}
        \centering
    \begin{lstlisting}[style=iql-small]
        SELECT
            PROBABILITY OF
                Contractor = \"Microsat Systems Canada Inc\"
            UNDER model GIVEN 
                Country_of_Contractor
        FROM data
\end{lstlisting}
    \end{minipage} \\  
    logpdf 7 & 
    \begin{minipage}{0.8\textwidth}
        \centering
    \begin{lstlisting}[style=iql-small]
        SELECT
            PROBABILITY OF 
                Inclination_radians = 5.52 AND
                Operator_Owner = \"AMSAT-UK\"
            UNDER model GIVEN 
                Launch_Vehicle AND
                Eccentricity
        FROM data
\end{lstlisting}
    \end{minipage} \\  
    logpdf 8 & 
    \begin{minipage}{0.8\textwidth}
        \centering
    \begin{lstlisting}[style=iql-small]
        SELECT
            PROBABILITY OF 
                Purpose = \"Earth Observation/Research\" AND
                Period_minutes = 5512.43 AND
                Launch_Vehicle = \"Tsyklon 3\"
            UNDER model GIVEN 
                Eccentricity AND
                Dry_Mass_kg AND
                Launch_Mass_kg
        FROM data
\end{lstlisting}
    \end{minipage} \\  
    logpdf 9 & 
    \begin{minipage}{0.8\textwidth}
        \centering
    \begin{lstlisting}[style=iql-small]
        SELECT
            PROBABILITY OF 
                longitude_radians_of_geo = 2.19 AND
                Eccentricity = 0.00319 AND
                Inclination_radians = 20.67 AND
                Type_of_Orbit = \"Molniya\"
            UNDER model GIVEN 
                Launch_Mass_kg AND
                Launch_Vehicle AND
                Purpose AND
                Launch_Site
        FROM data
\end{lstlisting}
    \end{minipage} \\
        logpdf 10 & 
        \begin{minipage}{0.8\textwidth}
            \centering
        \begin{lstlisting}[style=iql-small,belowskip=0pt]
            SELECT
                PROBABILITY OF 
                    Period_minutes = 19529.87 AND
                    Type_of_Orbit = \"Deep Highly Eccentric\" AND
                    Launch_Site = \"Kodiak Launch Complex\" AND
                    Dry_Mass_kg = 5093.73 AND
                    Inclination_radians = 8.17
                UNDER model GIVEN 
                    Contractor AND
                    Launch_Mass_kg AND
                    Purpose AND
                    Perigee_km AND 
                    Power_watts
            FROM data
    \end{lstlisting}
        \end{minipage}
\end{tabular}
\caption{Queries from \cref{tab:runtime} in \iql{} (continued).}
\label{benchmark-queries-2}
\end{figure}


\end{document}